\documentclass{amsart}
\usepackage{amssymb,amsmath,amsfonts}
\usepackage{amsaddr}
\usepackage{geometry} 
\usepackage{multicol}
\usepackage{graphicx}
\usepackage{enumerate,enumitem}
\usepackage{verbatim,setspace}
\usepackage{pdfsync}
\usepackage{gastex}
\usepackage[lowtilde]{url}
\usepackage{xcolor}

\newtheorem{theorem}{Theorem}
\newtheorem{corollary}[theorem]{Corollary}
\newtheorem{lemma}[theorem]{Lemma}
\newtheorem{proposition}[theorem]{Proposition}

\newtheorem{remarkthm}[theorem]{Remark}
\newtheorem{definition}[theorem]{Definition}
\newenvironment{remark}{\begin{remarkthm}\normalfont\quad}{\end{remarkthm}}
\renewcommand{\le}{\leqslant}
\renewcommand{\ge}{\geqslant}

\newcommand{\eps}{\varepsilon}

\newcommand{\Sig}{\Sigma}
\newcommand{\sig}{\sigma}
\newcommand{\noin}{\noindent}
\newcommand{\Bbf}{\mathbf{B}_{\mathrm{bf}}}
\newcommand{\Vbf}{\mathbf{W}^{\le 5}_{\mathrm{bf}}}
\newcommand{\Wbf}{\mathbf{W}^{\ge 6}_{\mathrm{bf}}}
\newcommand{\e}[1]{\hat{#1}}

\newcommand{\cD}{{\mathcal D}}

\newcommand{\cT}{{\mathcal T}}
\newcommand{\cW}{{\mathcal W}}

\newcommand{\qedb}{\hfill$\blacksquare$}

\begin{document}
\title{Syntactic complexity of bifix-free regular languages}
\author{Marek Szyku{\l}a}
\address{Institute of Computer Science, University of Wroc{\l}aw,\\
Joliot-Curie 15, PL-50-383 Wroc{\l}aw, Poland}
\email{msz@cs.uni.wroc.pl}
\author{John Wittnebel}
\address{David R. Cheriton School of Computer Science, University of Waterloo,\\
Waterloo, ON, Canada N2L 3G1}
\email{jkwittnebel@hotmail.com}

\begin{abstract}
We study the properties of syntactic monoids of bifix-free regular languages.
In particular, we solve an~open problem concerning syntactic complexity:
We prove that the cardinality of the syntactic semigroup of a~bifix-free language with state complexity $n$ is at most $(n-1)^{n-3}+(n-2)^{n-3}+(n-3)2^{n-3}$ for $n\ge 6$.
The main proof uses a~large construction with the method of injective function.
Since this bound is known to be reachable, and the values for $n \le 5$ are known, this completely settles the problem.
We also prove that $(n-2)^{n-3} + (n-3)2^{n-3} - 1$ is the minimal size of the alphabet required to meet the bound for $n \ge 6$.
Finally, we show that the largest transition semigroups of minimal DFAs which recognize bifix-free languages are unique up to renaming the states.

\medskip\noin
\end{abstract}

\maketitle
\section{Introduction}

The \emph{syntactic complexity}~\cite{BrYe11} $\sig(L)$ of a~regular language $L$ is defined as the size of its syntactic semigroup~\cite{Pin97}.
It is known that this semigroup is isomorphic to the transition semigroup of the quotient automaton $\cD$ and of a~minimal deterministic finite automaton accepting the language.
The number $n$ of states of $\cD$ is the \emph{state complexity} of the language~\cite{Yu01}, and it is the same as the \emph{quotient complexity}~\cite{Brz10} (number of left quotients) of the language.
The \emph{syntactic complexity of a~class} of regular languages is the maximal syntactic complexity of languages in that class expressed as a~function of the quotient complexity~$n$.

Syntactic complexity is related to the Myhill equivalence relation \cite{Myh57}, and it counts the number of classes of non-empty words in a~regular language which act distinctly.
It provides a~natural bound on the time and space complexity of algorithms working on the transition semigroup.
For example, a~simple algorithm checking whether a~language is \emph{star-free} just enumerates all transformations and verifies whether none of them contains a~non-trivial cycle \cite{McSe71}.

Syntactic complexity does not refine state complexity, but used as an~additional measure it can distinguish particular subclasses of regular languages from
the class of all regular languages, whereas state complexity alone cannot.
For example, the state complexity of basic operations in the class of star-free languages is the same as in the class of all regular languages (except the reversal, where the tight upper bound is $2^{n-1}-1$ see \cite{BrSz15Aperiodic}).

Finally, the largest transition semigroups play an~important role in the study of \emph{most complex} languages~\cite{Brz13} in a~given subclass.
These are languages that meet all the upper bounds on the state complexities of Boolean operations, product, star, and reversal, and also have maximal syntactic semigroups and most complex atoms~\cite{BrTa14}.
In particular, the results from this paper enabled the study of most complex bifix-free languages \cite{FeSz17ComplexityOfBifixFree}.

A~language is \emph{prefix-free} if no word in the language is a~proper prefix of another word in the language.
Similarly, a~language is \emph{suffix-free} if there is no word that is a~proper suffix of another word in the language.
A~language is \emph{bifix-free} if it is both prefix-free and suffix-free.
Prefix-, suffix-, and bifix-free languages are important classes of codes, which have numerous applications in such fields as cryptography and data compression.
Codes have been studied extensively; see~\cite{BPR09} for example.

Syntactic complexity has been studied for a~number of subclasses of regular languages (e.g.,~\cite{BrLi15,BLL12,BLY12,BrSz15Aperiodic,HoKo04,IvNa14}).
For bifix-free languages, the lower bound $(n-1)^{n-3}+(n-2)^{n-3}+(n-3)2^{n-3}$ for the syntactic complexity for $n \ge 6$ was established in~\cite{BLY12}. The values for $n \le 5$ were also determined.

The problem of establishing tight upper bound on syntactic complexity can be quite challenging, depending on the particular subclass.
For example, it is easy for prefix-free languages and right ideals, while much more difficult for suffix-free languages and left ideals.
The case of bifix-free languages studied in this paper requires an~even more involved proof, as the structure of a~maximal transition semigroup is more complicated.

Our main contributions in this paper are as follows:
\begin{enumerate}
\item We prove that $(n-1)^{n-3}+(n-2)^{n-3}+(n-3)2^{n-3}$ is also an~upper bound for syntactic complexity for $n \ge 8$.
To do this, we apply the general method of injective function (cf. \cite{BrSz14a} and~\cite{BrSz15SyntacticComplexityOfSuffixFree}).
The construction here is much more involved than in the previous cases and uses a~number of tricks for ensuring injectivity.
\item We prove that the transition semigroup meeting this bound is unique for every $n \ge 8$.
\item We refine the witness DFA meeting the bound by reducing the size of the alphabet to $(n-2)^{n-3} + (n-3)2^{n-3} - 1$, and we show that it cannot be any smaller.
\item Using a~dedicated algorithm, we verify by computation that two semigroups $\Vbf$ and $\Wbf$ (defined below) are the unique largest transition semigroups of a~minimal DFA of a~bifix-free language, respectively for $n=5$ and $n=6,7$ (whereas they coincide for $n=3,4$).
\end{enumerate}
In summary, for every $n$ we have determined the syntactic complexity, the unique largest semigroups, and the minimal sizes of the alphabets required; this completely solves the problem for bifix-free languages.

These results have been announced in~\cite{SzWi17SyntacticComplexityOfBifixFree} with only proof ideas.

\section{Preliminaries}

Let $\Sigma$ be a~non-empty finite alphabet, and let $L \subseteq \Sigma^*$ be a~language.
If $w \in \Sigma^*$ is a~word, $L.w$ denotes the \emph{left quotient} or simply quotient of $L$ by $w$, which is defined by $L.w = \{u \mid wu \in L\}$.
The number of quotients of $L$ is its \emph{quotient complexity}~\cite{Brz10} $\kappa(L)$. 
From the Myhill-Nerode Theorem, a~language is regular if and only if the set of all quotients of the language is finite.
We denote the set of quotients of regular $L$ by $K=\{K_0,\dots,K_{n-1}\}$, where $K_0=L=L.\varepsilon$ by convention.

A~\emph{deterministic finite automaton (DFA)} is a~tuple $\cD=(Q, \Sigma, \delta, q_0,F)$, where $Q$ is a~finite non-empty set of \emph{states}, $\Sigma$ is a~finite non-empty \emph{alphabet}, $\delta\colon Q\times \Sigma\to Q$ is the \emph{transition function}, $q_0\in Q$ is the \emph{initial} state, and $F\subseteq Q$ is the set of \emph{final} states.
We extend $\delta$ to a~function $\delta\colon Q\times \Sig^*\to Q$ as usual.

The \emph{quotient DFA} of a~regular language $L$ with $n$ quotients is defined by
$\cD=(K, \Sigma, \delta_\cD, K_0,F_\cD)$, where $\delta_\cD(K_i,w)=K_j$ if and only if $K_i.w=K_j$, and $F_\cD=\{K_i\mid \eps \in K_i\}$.
Without loss of generality, we assume that $Q=\{0,\dots,n-1\}$.
Then $\cD=(Q, \Sigma, \delta, 0,F)$, where $\delta(i,w)=j$ if $\delta_\cD(K_i,w)=K_j$, and $F$ is the set of subscripts of quotients in $F_\cD$.
A~state $q \in Q$ is \emph{empty} if its quotient $K_q$ is empty.
The quotient DFA of $L$ is isomorphic to each complete minimal DFA of $L$.
The number of states in the quotient DFA of $L$ (the quotient complexity of $L$) is therefore equal to the state complexity of $L$.

In any DFA $\cD$, each letter $a\in \Sigma$ induces a~transformation on the set $Q$ of $n$ states.
We let $\cT_n$ denote the set of all $n^n$ transformations of $Q$; then $\cT_n$ is a~monoid under composition. 
The \emph{image} of $q\in Q$ under transformation $t$ is denoted by $qt$, and the \emph{image} of a~subset $S \subseteq Q$ is $St = \{qt \mid q \in S\}$.
If $s,t \in \cT_n$ are transformations, their composition is denoted by $st$ and defined by $q(st)=(qs)t$.
The identity transformation is denoted by $\mathbf{1}$, and we have $q\mathbf{1} = q$ for all $q \in Q$.
By $(S \to q)$, where $S \subseteq Q$ and $q \in Q$, we denote a~\emph{semiconstant} transformation that maps all the
states from $S$ to $q$ and behaves as the identity function for the states in $Q \setminus S$.
A~\emph{constant} transformation is the semiconstant transformation $(Q \to q)$, where $q \in Q$.
A~\emph{unitary} transformation is $(\{p\} \to q)$, for some distinct $p,q \in Q$; this is denoted by $(p \to q)$ for simplicity.

The \emph{transition semigroup} of $\cD$ is the semigroup of all transformations generated by the transformations induced by $\Sigma$.
Since the transition semigroup of a~minimal DFA of a~language $L$ is isomorphic to the syntactic semigroup of $L$~\cite{Pin97}, the syntactic complexity of $L$ is equal to the cardinality of the transition semigroup of $\cD$.

The \emph{underlying digraph} of a~transformation $t \in \cT_n$ is the digraph $(Q,E)$, where $E = \{(q,qt) \mid q \in Q\}$.
We identify a~transformation with its underlying digraph and use usual graph terminology for transformations:
The \emph{in-degree} of a~state $q \in Q$ is the cardinality $|\{p \in Q \mid pt = q\}|$.
A~\emph{cycle} in $t$ is a~cycle in its underlying digraph of length at least 2.
A~\emph{fixed point} in $t$ is a~self-loop in its underlying digraph.
The \emph{orbit} of a~state $q \in Q$ in $t$ is a~connected component containing $q$ in its underlying digraph, that is, the set $\{p \in Q \mid pt^i = qt^j\text{ for some }i,j \ge 0\}$.
Note that every orbit contains either exactly one cycle or one fixed point.
The \emph{distance} in $t$ from a~state $p \in Q$ to a~state $q \in Q$ is the length of the path in the underlying digraph of $t$ from $p$ to $q$, that is, $\min\{i \in \mathbb{N} \mid pt^i = q\}$, and is undefined if no such path exists.
If a~state $q$ does not lie in a~cycle, then the \emph{tree} of $q$ is the underlying digraph of $t$ restricted to the states $p$ such that there is a~path from $p$ to $q$.

\subsection{Bifix-free languages and semigroups}

Let $\cD_n=(Q, \Sigma, \delta, 0,F)$, where $Q = \{0,\ldots,n-1\}$, be a~minimal DFA accepting a~bifix-free language $L$, and let $T(\cD_n)$ be its transition semigroup. We also define $Q_M = \{1,\ldots,n-3\}$ (the set of the ``middle'' non-special states).

The following properties of bifix-free languages, slightly adapted to our terminology, are well known~\cite{BLY12}:
\begin{lemma}\label{lem:bifix-free}
A~minimal DFA $\cD_n=(Q, \Sigma, \delta, 0,F)$ of a~bifix-free languages $L$ satisfies the following properties:
\begin{enumerate}
\item There is an~empty state, which is $n-1$ by convention.
\item There exists exactly one final quotient, which is $\{\varepsilon\}$, and whose state is $n-2$ by convention, so $F=\{n-2\}$.
\item For $u,v\in \Sigma^+$, if $L.v\neq \emptyset$, then $L.v\neq L.uv$.
\item In the underlying digraph of every transformation of $T(\cD_n)$, there is a~path starting at $0$ and ending at $n-1$.
\end{enumerate}
\end{lemma}
Items~(1) and~(2) are sufficient and necessary conditions for a~prefix-free language, items~(1) and~(3) are sufficient and necessary conditions for a~suffix-free, and item~(4) follows from item~(3).
Following \cite{BrSz15SyntacticComplexityOfSuffixFree}, we say that an (unordered) pair $\{p,q\}$ of distinct states in $Q_M$ is \emph{colliding} (or $p$ \emph{collides} with $q$) in $T(\cD_n)$ if there is a~transformation $t \in T(\cD_n)$ such that $0t = p$ and $rt = q$ for some $r \in Q_M$.
A~pair of states is \emph{focused by} a~transformation $u \in \cT(n)$ if $u$ maps both states of the pair to a~single state $r \in Q_M \cup \{n-2\}$.
We then say that $\{p,q\}$ is \emph{focused to the state $r$}.
By Lemma~\ref{lem:bifix-free}(3), it follows that if $\{p,q\}$ is colliding in $T(\cD_n)$, then there is no transformation $u \in T(\cD_n)$ that focuses $\{p,q\}$.
Hence, in the case of bifix-free languages, colliding states can be mapped to a~single state only if this state is $n-1$.
In contrast to suffix-free languages, we do not consider the pairs from $Q_M \times \{n-2\}$ being colliding, as they cannot be focused.

For $n \ge 2$ we define the set of transformations
\begin{eqnarray*}
\Bbf(n) = \{t \in \cT_n & \mid & 0 \notin Qt\text{, }(n-1)t=n-1\text{, }(n-2)t=n-1\text{, and}\\
& & \text{for all }j\ge 1, 0t^j = n-1\text{ or }\forall q \in Q\setminus \{0,n-1\}\;0t^j \neq qt^j\}.
\end{eqnarray*}
In~\cite{BLY12} it was shown that the transition semigroup $T(\cD_n)$ of a~minimal DFA of a~bifix-free
language must be contained in $\Bbf(n)$.
It contains all transformations $t$ which fix $n-1$, map $n-2$ to $n-1$, and do not focus any
pair which is colliding from $t$.
The condition of fixing the empty state $n-1$ is obvious. State $n-2$ must be always mapped to $n-1$ because the language must be prefix-free, and focusing state $0$ with any other state is forbidden because the language must be suffix-free.
For $n \ge 5$, $\Bbf(n)$ is not a~semigroup, because compositions of its transformations may violate the condition about focusing state $0$.
For example, the transformations $t_1 = (0 \to 1)(1 \to 2)(Q \setminus \{0,1\} \to n-1)$ and $t_2 = (1 \to 2)(Q \setminus \{1,2\} \to n-1)$ are in $\Bbf(n)$, but their composition $t_1 t_2 = (0 \to 2)(1 \to 2)(Q \setminus \{0,1\})$ focuses the pair $\{0,1\}$, thus is not in $\Bbf(n)$.

Since $\Bbf(n)$ is not a~semigroup, no transition semigroup of a~minimal DFA of a~bifix-free language can contain all transformations from $\Bbf(n)$.
Therefore, its cardinality is not a~tight upper bound on the syntactic complexity of bifix-free languages.
A~lower bound on the syntactic complexity was established in~\cite{BLY12}.
We study the following two semigroups that play an~important role for bifix-free languages.

\subsubsection{Semigroup $\Wbf(n)$.}

For $n \ge 3$ we define the semigroup:
\begin{eqnarray*}
\Wbf(n) & = & \{t \in \Bbf(n) \mid 0t \in \{n-2,n-1\}\text{, or}\\
& & 0t \in Q_M\text{ and for all $q \in Q_M$ we have} qt \in \{n-2,n-1\}\}.
\end{eqnarray*}
The name of this semigroup follows from \cite{BLY12}, and the superscript denotes that it is, as we will show, the largest syntactic semigroup of a~bifix-free language when $n \ge 6$.

The following remark summarizes the transformations of $\Wbf(n)$ (illustrated in Fig.~\ref{fig:Wbf_transformations}):
\begin{remark}\label{rem:Wbf_transformations}
$\Wbf(n)$ contains all transformations that:
\begin{enumerate}[align=left,leftmargin=*]
\item[\bf(Type~1)] map $\{0,n-2,n-1\}$ to $n-1$, and $Q_M$ into $Q \setminus \{0\}$, or
\item[\bf(Type~2)] map $0$ to $n-2$, $\{n-2,n-1\}$ to $n-1$, and $Q_M$ into $Q \setminus \{0,n-2\}$, or
\item[\bf(Type~3)] map $0$ to a~state $q \in Q_M$, and $Q_M$ into $\{n-2,n-1\}$.\qedb
\end{enumerate}
\end{remark}
They are the only transformations in $\Wbf(n)$, and we will be referring to these three types of transformations.
\begin{figure}[ht]
\unitlength 8pt\scriptsize
\gasset{Nh=2.5,Nw=2.5,Nmr=1.25,ELdist=0.3,loopdiam=1.5}
\begin{center}\begin{picture}(14,16)(0,-4)
\node[Nframe=n](name)(1,10){(1):}
\node(0)(2,4){0}\imark(0)
\node(1)(6,8){$1$}
\node[Nframe=n,Nw=2,Nh=2](dots)(6,4){$\dots$}
\node[Nw=3.5,Nh=11.5,Nmr=1.25,dash={.5 .25}{.25}](QM)(6,4){}
\node(n-3)(6,0){$n$-$3$}
\node(n-2)(10,4){$n$-$2$}\rmark(n-2)
\node(n-1)(10,0){$n$-$1$}
\drawedge[curvedepth=-6,sxo=-.2,exo=.2](0,n-1){}
\drawedge(n-2,n-1){}
\drawloop[loopangle=270](n-1){}
\drawloop[loopangle=0,syo=3](QM){}
\drawedge(QM,n-2){}
\drawedge(QM,n-1){}
\end{picture}\begin{picture}(14,14)(0,-4)
\node[Nframe=n](name)(1,10){(2):}
\node(0)(2,4){0}\imark(0)
\node(1)(6,8){$1$}
\node[Nframe=n,Nw=2,Nh=2](dots)(6,4){$\dots$}
\node[Nw=3.5,Nh=11.5,Nmr=1.25,dash={.5 .25}{.25}](QM)(6,4){}
\node(n-3)(6,0){$n$-$3$}
\node(n-2)(10,4){$n$-$2$}\rmark(n-2)
\node(n-1)(10,0){$n$-$1$}
\drawedge[curvedepth=8,sxo=-.5,exo=.5](0,n-2){}
\drawedge(n-2,n-1){}
\drawloop[loopangle=270](n-1){}
\drawloop[loopangle=0,syo=3](QM){}
\drawedge(QM,n-1){}
\end{picture}\begin{picture}(14,14)(0,-4)
\node[Nframe=n](name)(1,10){(3):}
\node(0)(2,4){0}\imark(0)
\node(1)(6,8){$1$}
\node[Nframe=n,Nw=2,Nh=2](dots)(6,4){$\dots$}
\node[Nw=3.5,Nh=11.5,Nmr=1.25,dash={.5 .25}{.25}](QM)(6,4){}
\node(n-3)(6,0){$n$-$3$}
\node(n-2)(10,4){$n$-$2$}\rmark(n-2)
\node(n-1)(10,0){$n$-$1$}
\drawedge(0,dots){}
\drawedge(n-2,n-1){}
\drawloop[loopangle=270](n-1){}
\drawedge(QM,n-2){}
\drawedge(QM,n-1){}
\end{picture}\end{center}
\caption{The three types of transformations in $\Wbf(n)$ from Remark~\ref{rem:Wbf_transformations}.}\label{fig:Wbf_transformations}
\end{figure}
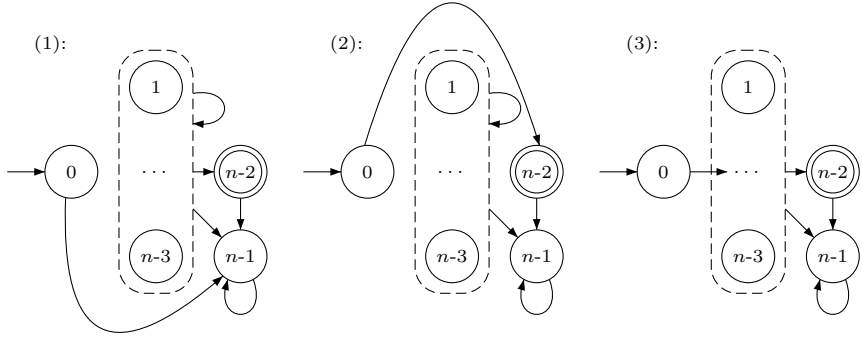

The cardinality of $\Wbf(n)$ is $(n-1)^{n-3}+(n-2)^{n-3}+(n-3)2^{n-3}$.

\begin{proposition}\label{pro:Wbf_unique}
$\Wbf(n)$ is the unique maximal transition semigroup of a~minimal DFA $\cD_n$ of a~bifix-free language in which there are no colliding pairs of states.
\end{proposition}
\begin{proof}
Since for any pair $p,q \in Q_M$ there is the transformation $(0 \to n-1)(\{p,q\} \to n-2)(n-2 \to n-1)$ in the semigroup, the pair $\{p,q\}$ cannot be colliding.
Therefore, there are no colliding pairs in $\Wbf(n)$.

Let $T(\cD_n)$ be a~transition semigroup in which there are no colliding pairs of states.
Consider $t \in T(\cD_n)$.
If $0t = n-1$ then $t \in \Wbf(n)$ as is a~transformation of Type~1.
If $0t = n-2$ then $t \in \Wbf(n)$ as is a~transformation of Type~2.
If $0t \in Q_M$, then $qt \in \{n-2,n-1\}$, as otherwise $\{0t,qt\}$ would be a~colliding pair, so $t$ is a~transformation of Type~3 from.
Therefore, $T(\cD_n)$ is a~subsemigroup of $\Wbf(n)$, and so $\Wbf(n)$ is unique maximal.
\end{proof}

In~\cite{BLY12} it was shown that for $n \ge 5$, there exists a~witness DFA of a~bifix-free language whose transition semigroup is $\Wbf(n)$ over an~alphabet of size $(n-2)^{n-3}+(n-3)2^{n-3}+2$ (and 18 if $n=5$).
Now we slightly refine the witness from \cite[Proposition~31]{BLY12} by reducing the size of the alphabet to $(n-2)^{n-3} + (n-3)2^{n-3} - 1$, and then we show that it cannot be any smaller.
\begin{definition}[Bifix-free witness]
For $n \ge 4$, let $\cW(n) = (Q,\Sigma,\delta,0,\{n-2\})$, where $Q = \{0,\ldots,n-1\}$ and $\Sigma$ contains the following letters:
\begin{enumerate}
\item $b_i$, for $1 \le i \le n-3$, inducing the transformations $(0 \to n-1)(i \to n-2)(n-2 \to n-1)$,
\item $c_i$, for every transformation of Type~(2) that is different from $(0 \to n-2)(Q_M \to n-1)(n-2 \to n-1)$,
\item $d_i$, for every transformation of Type~(3) that is different from $(0 \to q)(Q_M \to n-1)(n-2 \to n-1)$ for some state $q \in Q_M$.
\end{enumerate}
Altogether, we have $|\Sigma| = (n-3) + ((n-2)^{n-3}-1) + (n-3)(2^{n-3}-1) = (n-2)^{n-3} + (n-3)2^{n-3} - 1$.
For $n=4$ two letters suffice, since the transformation of $b_1$ is induced by $c_i d_i$, where $c_i\colon (0 \to 2)(2 \to 3)$ and $d_i\colon (0 \to 1)(1 \to 2)(2 \to 3)$.
\end{definition}

\begin{proposition}
The transition semigroup of $\cW(n)$ is $\Wbf(n)$.
\end{proposition}
\begin{proof}
Consider a~transformation $t$ of Type~1.
Let $S \subseteq Q_M$ be the states that are mapped to $n-2$ by $t$.
If $S = \emptyset$, then $t = s x$, where $s = (0 \to n-2)(n-2 \to n-1)$ is the transformation induced by some $c_i$, and $x$ is the transformation induced by some $c_j$ that maps $Q_M$ in the same way as $t$.
If $S \neq \emptyset$, then let $q \in Q_M$ be the state such that $q \notin Q_M t$.
Let $x$ be the transformation induced by $c_i$ that maps the states from $S$ to $q$ and $Q_M \setminus S$ in the same way as $t$.
Then $t = x b_q$, since $0 x b_q = n-1$, $S x b_q = q b_q = n-2$, and for $p \in (Q_M \setminus S)$ we know that $p x b_q = p t$.
Hence, we have all transformations of Type~1 in $\Wbf(n)$.

It remains to show how to generate the two missing transformations of Type~2 and Type~3 that do not have the corresponding generators $c_i$ and $d_i$, respectively.
Let $u = (0 \to q)(Q_M \to n-2)(n-2 \to n-1)$, which is induced by a~$d_i$.
Consider the transformation $t = (0 \to q)(Q_M \to n-1)(n-2 \to n-1)$.
Then $t = u v$, where $v = (0 \to n-1)(n-2 \to n-1)$ is of Type~1.
Consider the transformation $t = (0 \to n-2)(Q_M \to n-1)(n-2 \to n-1)$.
Then $t = u v$, where $v = (0 \to n-1)(q \to n-2)(n-2 \to n-1)$ is of Type~1.
\end{proof}

\begin{proposition}\label{pro:Wbf_alphabet_lower_bound}
For $n \ge 5$, at least $(n-2)^{n-3} + (n-3)2^{n-3} - 1$ generators are necessary to generate $\Wbf(n)$.
\end{proposition}
\begin{proof}
Consider a~transformation $t \in \Wbf(n)$ of Type~(2) that is different from $(0 \to n-2)(Q_M \to n-1)(n-2 \to n-1)$.
If $t$ were the composition of two transformations from $\Wbf(n)$, then either $t$ maps $0$ to $n-1$, or $t$ maps $Q_M$ into $\{n-2,n-1\}$.
Since neither is the case, $t$ must be a~generator.
There are $(n-2)^{n-3} - 1$ such generators.

Consider a~transformation $t \in \Wbf(n)$ of Type~(3) that is different from $(0 \to q)(Q_M \to n-1)(n-2 \to n-1)$ for some $q \in Q_M$.
Note that to generate $t$ a~transformation of Type~(3) must be used, but the composition of such a~transformation with any other transformation from $\Wbf(n)$ maps every state from $Q_M$ to $n-1$.
Hence, $t$ must be used as a~generator, and there are $(n-3)(2^{n-3}-1)$ such generators.

Consider a~transformation $t \in \Wbf(n)$ of Type~(1) of the form $(0 \to n-1)(q \to n-2)(n-2 \to n-1)$ for some $q \in Q_M$.
Note that to generate $t$, transformations of Type~(3) cannot be used because $Q_M$ is not mapped into $\{n-2,n-1\}$ if $|Q_M| \ge 3$.
Let $t = g_1 \dots g_k$, where $g_i$ are generators.
Since a~transformation of Type~(2) does not map $q$ to $n-2$, $g_k$ cannot be of Type~(2), and so must be of Type~(1).
Moreover $Q_M g_1 \dots g_{k-1} = Q_M$, as otherwise $t$ would map a~state $p \in Q_M$ to $n-1$.
Hence, $Q_M g_k = Q_M \setminus \{q\}$, and for every selection of $q$ there exists a~different $g_k$.
There are $(n-3)$ such generators.
\end{proof}

\subsubsection{Semigroup $\Vbf(n)$.}

For $n \ge 3$ we define the semigroup
\begin{eqnarray*}
\Vbf(n) = \{t \in \Bbf(n) & \mid & \text{for all }p,q \in Q_M\text{ where }p \neq q, pt = qt = n-1\text{ or }pt \neq qt\}.
\end{eqnarray*}
The name of this semigroup follows from \cite{BLY12}, and the superscript denotes that it is, as we will show, the largest syntactic semigroup of a~bifix-free language when $n \le 5$.
In this semigroup, there are all transformations from $\Bbf(n)$ that do not focus any pair.
By taking only such transformations, we are allowed to have all pairs of states possibly colliding.

\begin{proposition}\label{pro:Vbf_unique}
$\Vbf(n)$ is the unique maximal transition semigroup of a~minimal DFA $\cD_n$ of a~bifix-free language in which all pairs of states from $Q_M$ are colliding.
\end{proposition}
\begin{proof}
Let $p,q \in Q_M$ be two distinct states.
Then $\{p,q\}$ is colliding because of the transformation $(0 \to p)(p \to n-1)(n-2 \to n-1) \in \Vbf(n)$.
Therefore, all pairs of states from $Q_M$ are colliding.

Let $T(\cD_n)$ be a~transition semigroup with all colliding pairs of states.
Consider a~transformation $t \in T(\cD_n)$.
Then for every distinct $p,q \in Q_M$, we have $pt \neq qt$ or $pt = qt = n-1$, as otherwise $\{p,q\}$ would be focused.
By definition of $\Vbf(n)$, there are all such transformations $t$ in $\Vbf(n)$.
Therefore, $\Vbf(n)$ is unique maximal.
\end{proof}

In~\cite{BLY12} it was shown that for $n \ge 2$ there exists a~DFA for a~bifix-free language
whose transition semigroup is $\Vbf(n)$ over an~alphabet of size $(n-2)!$.
We prove that this is an~alphabet of minimal size that generates this transition semigroup.
\begin{proposition}\label{pro:Vbf_alphabet_lower_bound}
To generate $\Vbf(n)$ at least $(n-2)!$ generators must be used.
\end{proposition}
\begin{proof}
First we show that the composition of any two transformations $t,t' \in \Vbf(n)$ maps a~state different from $n-1$ to state $n-1$.
Suppose that $t$ does not map any state to $n-1$.
If $0t = n-2$, then $0tt' = n-1$.
If $0t \in Q_M$, then some state $q \in Q_M$ must be mapped either to $n-2$ or to $n-1$, and again $qtt' = n-1$.

Consider all transformations $t \in \Vbf(n)$ that map $Q_M \cup \{0\}$ onto $Q_M \cup \{n-2\}$, hence they must be bijections between these states.
There are $(n-2)!$ such transformations, and since they cannot be generated by compositions, they must be generators.
\end{proof}

\section{Upper bound on the syntactic complexity of bifix-free languages}

Our main result shows that the lower bound $(n-1)^{n-3}+(n-2)^{n-3}+(n-3)2^{n-3}$ on the syntactic complexity of bifix-free languages is also an~upper bound for $n \ge 8$.

We consider a~minimal DFA $\cD_n=(Q,\Sigma,\delta,0,\{n-2\})$, where $n-2$ is the only final state, and $n-1$ is the empty state, recognizing an~arbitrary bifix-free language.
Let $T(\cD_n)$ be the transition semigroup of $\cD_n$.
We will show that $T(\cD_n)$ is not larger than $\Wbf(n)$.

Note that the semigroups $T(\cD_n)$ and $\Wbf(n)$ share the set $Q$, and in both of them $0$, $n-2$, and $n-1$ play the role of the initial, final, and empty state, respectively.
When we say that a~pair of states from $Q$ is \emph{colliding} we always mean that it is colliding in $T(\cD_n)$.

First, we state the following lemma, which generalizes some arguments that we use frequently in the proof of the main theorem.
\begin{lemma}\label{lem:orbits}
Let $t,\e{t} \in T(\cD_n)$ and $s \in \Wbf(n)$ be transformations.
Suppose that:
\begin{enumerate}
\item All states from $Q_M$ whose mapping is different in $t$ and $s$ belong to $C$, where $C$ is either an~orbit in $s$ or is the tree of a~state in $s$.
\item All states from $Q_M$ whose mapping is different in $\e{t}$ and $s$ belong to $\e{C}$, where $\e{C}$ is either an~orbit in $s$ or is the tree of a~state in $s$.
\item The transformation $s^i t^j$, for some $i,j \ge 0$, focuses a~colliding pair whose states are in $C$.
\end{enumerate}
Then either $C \subseteq \e{C}$ or $\e{C} \subseteq C$.
In particular, if $C$ and $\e{C}$ are both orbits or both trees rooted in a~state mapped by $s$ to $n-1$, then $C = \e{C}$.
\end{lemma}
\begin{proof}
First observe that if $q \in (Q_M \cup \{n-2,n-1\}) \setminus C$ then $qs = qt$, since by~(1) state $q$ is mapped in the same way by $t$ as by $s$.
Also, $qs \in (Q_M \cup \{n-2,n-1\}) \setminus C$, since if $qs$ would be in $C$, then $q \in C$, because $C$ is an~orbit or a~tree and $qs$ is reachable from $q$.
Hence, for any $g=g_1 \dots g_k$, where $g_i = t$ or $g_i = s$, by a~simple induction we obtain $q g = q s^k = q t^k \in (Q_M \cup \{n-2,n-1\}) \setminus C$.
The same claim holds symmetrically for $\e{C}$.

Let $\{p_1,p_2\}$ be the colliding pair that is focused by $s^i t^j$ from~(3).
Suppose that $C \cap \e{C} = \emptyset$.
Since $p_1,p_2 \in C$, we know that $p_1,p_2 \in (Q_M \cup \{n-2,n-1\}) \setminus \e{C}$.
By the claim above for $\e{C}$, $p_1 s^i t^j = p_1 \e{t}^i t^j$, and $p_2 s^i t^j = p_2 \e{t}^i t^j$.
But this means that $\e{t}^i t^j$ focuses $\{p_1,p_2\}$, hence $t$ and $\e{t}$ cannot be both present in $T(\cD_n)$.

So it must be that $C \cap \e{C} \neq \emptyset$, since they are orbits or trees we have either $C \subseteq \e{C}$ or $\e{C} \subseteq C$.
\end{proof}

\begin{theorem}\label{thm:bifix-free_upper_bound}
For $n\ge 8$, the syntactic complexity of the class of bifix-free languages with $n$ quotients is $(n-1)^{n-3}+(n-2)^{n-3}+(n-3)2^{n-3}$.
\end{theorem}
\begin{proof}
We construct an~injective mapping $\varphi \colon T(\cD_n) \to \Wbf(n)$.
Since $\varphi$ will be injective, this will prove that $|T(\cD_n)| \le |\Wbf(n)| = (n-1)^{n-3}+(n-2)^{n-3}+(n-3)2^{n-3}$.

The mapping $\varphi$ is defined by 23 (sub)cases covering all possibilities for a~transformation $t \in T(\cD_n)$.
Let $t$ be a~transformation of $T(\cD_n)$, and $s$ be the assigned transformation $\varphi(t)$.
In every (sub)case we prove \emph{external injectivity}, which is that there is no other
transformation $\e{t}$ that fits in one of the previous (sub)cases and results in the same $s$, and
we prove \emph{internal injectivity}, which is that no other transformation $\e{t}$ that fits  the same (sub)case results in the same $s$.
All states and variables related to $\e{t}$ are always marked by a~hat.

In every (sub)case we observe some properties of the defined transformations $s$:
Property~(a) always says that a~colliding pair is focused by a~transformation of the form $s^i t^j$.
Property~(b) describes the orbits and trees of states which are mapped differently by $t$ and $s$; this is often for a~use of Lemma~\ref{lem:orbits}.
Property~(c) concerns the existence of cycles in $s$.

See the appendix for a~list and a~map of all (sub)cases.

\textbf{Supercase~1}: $t \in \Wbf(n)$.\\
We take $s = t$.
The internal and external injectivities are obvious.

\noindent For all the remaining cases let $p = 0t$.
Note that all $t$ with $p \in \{n-2,n-1\}$ fit in Supercase~1, and by the property of suffix-freeness, $pt^i \in \{n-2,n-1\}$ for some $i$.
Let $k \ge 0$ be the largest integer such that $pt^k \notin \{n-2,n-1\}$.

Then $pt^{k+1}$ is either $n-1$ or $n-2$, and we have two supercases covering these situations.

\textbf{Supercase~2}: $t \notin \Wbf(n)$ and $pt^{k+1} = n-1$.\\
Here we have the chain
$$0 \stackrel{t}{\rightarrow} p \stackrel{t}{\rightarrow} pt \stackrel{t}{\rightarrow} \dots \stackrel{t}{\rightarrow} pt^k \stackrel{t}{\rightarrow} n-1.$$
Within this supercase, we will always assign transformations $s$ focusing a~colliding pair.
Because the transformations assigned in Supercase~1 belong to $T(n)$, they do not focus any colliding pair, which makes them different from the transformations assigned now in Supercase~2.
Also, we will always have $0 s = n-1$.
We have the following cases covering all possibilities for $t$:

\textbf{Case~2.1}: $t$ has a~cycle.\\
Let $r$ be the minimal state among the states that appear in cycles of $t$, that is,
$$r = \min\{q\in Q \mid q\text{ is in a~cycle of }t\}.$$
Let $s$ be the transformation illustrated in Fig.~\ref{fig:case2.1} and defined by:
\begin{center}
  $0 s = n-1$, $p s = r$,\\
  $(p t^i) s = pt^{i-1}$ for $1\le i\le k$,\\
  $q s = q t$ for the other states $q\in Q$.
\end{center}
\begin{figure}[ht]
\unitlength 10pt\small
\gasset{Nh=2.5,Nw=2.5,Nmr=1.25,ELdist=0.3,loopdiam=1.5}
\begin{center}\begin{picture}(28,13)(0,-4)
\node[Nframe=n](name)(0,7){\normalsize$t\colon$}
\node(0)(2,0){0}\imark(0)
\node(p)(8,0){$p$}
\node[Nframe=n](pdots)(14,0){$\dots$}
\node(pt^k)(20,0){$pt^k$}
\node(n-1)(26,0){$n$-$1$}
\node(n-2)(26,4){$n$-$2$}\rmark(n-2)
\node(z)(12,4){$z$}
\node(r)(14,7){$r$}
\node[Nframe=n](rdots)(16,4){$\dots$}
\drawedge(0,p){}
\drawedge(p,pdots){}
\drawedge(pdots,pt^k){}
\drawedge(pt^k,n-1){}
\drawedge(n-2,n-1){}
\drawloop[loopangle=270](n-1){}
\drawedge[curvedepth=1](z,r){}
\drawedge[curvedepth=1](r,rdots){}
\drawedge[curvedepth=1](rdots,z){}
\end{picture}
\begin{picture}(28,13)(0,-4)
\node[Nframe=n](name)(0,7){\normalsize$s\colon$}
\node(0')(2,0){0}\imark(0')
\node(p')(8,0){$p$}
\node[Nframe=n](pdots')(14,0){$\dots$}
\node(pt^k')(20,0){$pt^k$}
\node(n-1')(26,0){$n$-$1$}
\node(n-2')(26,4){$n$-$2$}\rmark(n-2')
\node(z')(12,4){$z$}
\node(r')(14,7){$r$}
\node[Nframe=n](rdots')(16,4){$\dots$}
\drawedge[linecolor=red,dash={.5 .25}{.25},curvedepth=-3](0',n-1'){}
\drawedge[linecolor=red,dash={.5 .25}{.25}](pdots',p'){}
\drawedge[linecolor=red,dash={.5 .25}{.25}](pt^k',pdots'){}
\drawedge[linecolor=red,dash={.5 .25}{.25},curvedepth=3.5](p',r'){}
\drawedge[curvedepth=1](z',r'){}
\drawedge[curvedepth=1](r',rdots'){}
\drawedge[curvedepth=1](rdots',z'){}
\drawloop[loopangle=270](n-1'){}
\drawedge(n-2',n-1'){}
\end{picture}\end{center}
\caption{Case~2.1.}\label{fig:case2.1}
\end{figure}
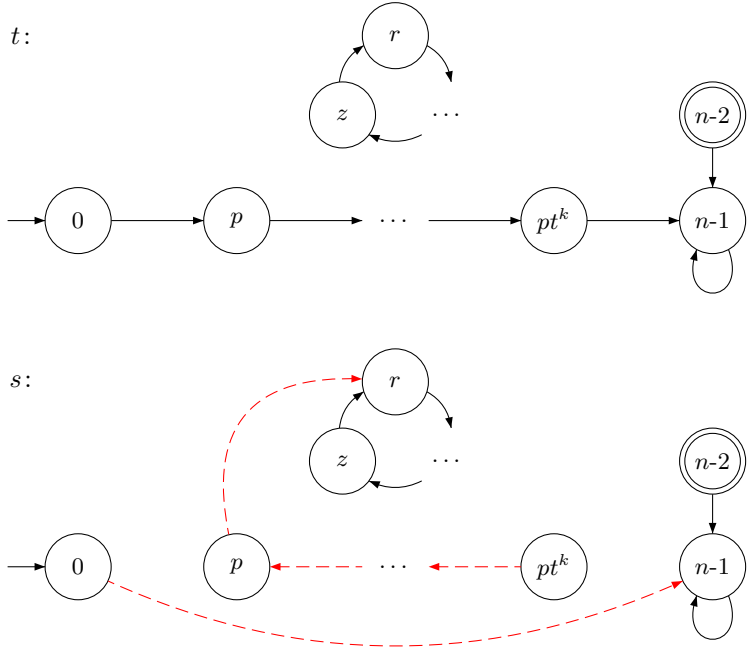

Let $z$ be the state from the cycle of $t$ such that $zt = r$. 
We observe the following properties:
\begin{enumerate}
\item[(a)] Pair $\{p,z\}$ is a~colliding pair focused by $s$ to state $r$ in the cycle, which is the smallest state of all states in cycles. 
This is the only colliding pair which is focused by $s$ to a~state in a~cycle.

\noindent\textit{Proof}: Note that $p$ collides with any state in a~cycle of $t$, in particular, with $z$.
The property follows because $s$ differs from $t$ only in the mapping of states $pt^i$ ($0 \le i \le k$) and $0$, and the only state mapped to a~cycle is $p$.

\item[(b)] All states from $Q_M$ whose mapping is different in $t$ and $s$ belong to the same orbit in $s$ of a~cycle.
Hence, all colliding pairs that are focused by $s$ consist only of states from this orbit.

\item[(c)] $s$ has a~cycle.

\item[(d)] For each $i$ with $1 \le i < k$, there is precisely one state $q$ colliding with $pt^{i-1}$ and mapped by $s$ to $pt^i$, and that state is $q=pt^{i+1}$.

\noindent\textit{Proof}: Clearly $q=pt^{i+1}$ satisfies this condition. Suppose that $q \neq pt^{i+1}$. Since $pt^{i+1}$ is the only state mapped to $pt^i$ by $s$ and not by $t$, it follows that $qt = qs = pt^i$. So $q$ and $pt^{i-1}$ are focused to $pt^i$ by $t$; since they collide, this is a~contradiction.
\end{enumerate}

\textit{Internal injectivity}:
Let $\e{t}$ be any transformation that fits in this case and results in the same $s$; we will show that $\e{t}=t$.
From~(a), there is the unique colliding pair $\{p,z\}$ focused to a~state in a~cycle, hence $\{\e{p},\e{z}\} = \{p,z\}$.
Moreover, $p$ and $\e{p}$ are not in this cycle, so $\e{p}=p$ and $\e{z}=z$, which means that $0t = 0\e{t} = p$.
Since there is no state $q \neq 0$ such that $qt=p$, the only state mapped to $p$ by $s$ is $pt$, hence $p\e{t} = pt$.
From~(d) for $i=1,\ldots,k-1$, state $pt^{i+1}$ is uniquely determined, hence $p\e{t}^{i+1} = pt^{i+1}$.
Finally, for $i=k$ there is no state colliding with $pt^{k-1}$ and mapped to $pt^k$, hence $p\e{t}^{k+1} = pt^{k+1} = n-1$.
Since the other transitions in $s$ are defined exactly as in $t$ and $\e{t}$, we have $\e{t}=t$.

\textbf{Case~2.2}: $t$ has no cycles, but $k \ge 1$.\\
Let $s$ be the transformation illustrated in Fig.~\ref{fig:case2.2} and defined by:
\begin{center}
  $0 s = n-1$, $p s = p$,\\
  $(p t^i) s = p t^{i-1}$ for $1\le i\le k$,\\
  $q s = q t$ for the other states $q\in Q$.
\end{center}
\begin{figure}[ht]
\unitlength 10pt\small
\gasset{Nh=2.5,Nw=2.5,Nmr=1.25,ELdist=0.3,loopdiam=1.5}
\begin{center}\begin{picture}(28,10)(0,-4)
\node[Nframe=n](name)(0,4){\normalsize$t\colon$}
\node(0)(2,0){0}\imark(0)
\node(p)(8,0){$p$}
\node[Nframe=n](pdots)(14,0){$\dots$}
\node(pt^k)(20,0){$pt^k$}
\node(n-1)(26,0){$n$-$1$}
\node(n-2)(26,4){$n$-$2$}\rmark(n-2)
\drawedge(0,p){}
\drawedge(p,pdots){}
\drawedge(pdots,pt^k){}
\drawedge(pt^k,n-1){}
\drawedge(n-2,n-1){}
\drawloop[loopangle=270](n-1){}
\end{picture}
\begin{picture}(28,10)(0,-4)
\node[Nframe=n](name)(0,4){\normalsize$s\colon$}
\node(0')(2,0){0}\imark(0')
\node(p')(8,0){$p$}
\node[Nframe=n](pdots')(14,0){$\dots$}
\node(pt^k')(20,0){$pt^k$}
\node(n-1')(26,0){$n$-$1$}
\node(n-2')(26,4){$n$-$2$}\rmark(n-2')
\drawedge[curvedepth=-3,linecolor=red,dash={.5 .25}{.25}](0',n-1'){}
\drawedge[linecolor=red,dash={.5 .25}{.25}](pdots',p'){}
\drawedge[linecolor=red,dash={.5 .25}{.25}](pt^k',pdots'){}
\drawloop(p'){}
\drawedge(n-2',n-1'){}
\drawloop[loopangle=270](n-1'){}
\end{picture}\end{center}
\caption{Case~2.2.}\label{fig:case2.2}
\end{figure}
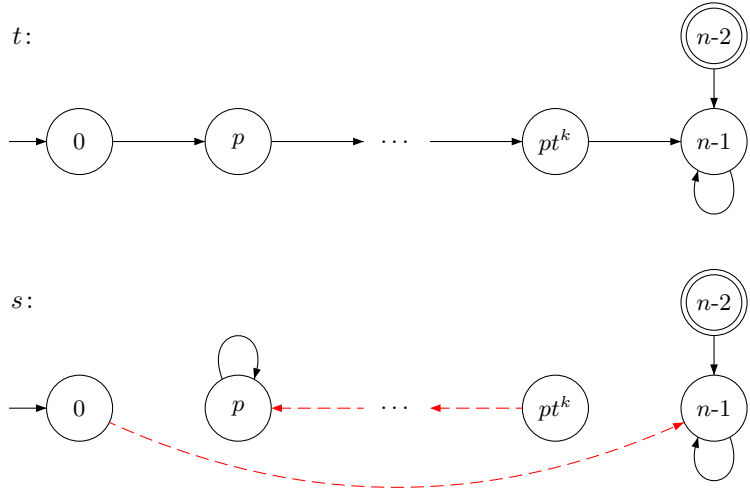

We observe the following properties:
\begin{enumerate}
\item[(a)] $\{p,pt\}$ is a~colliding pair focused by $s$ to a~fixed point of in-degree 2.
This is the only pair among all colliding pairs focused to a~fixed point.

\noindent\textit{Proof}: This follows from the definition of $s$, since any colliding pair focused by $s$ contains $pt^i$ ($0 \le i \le k$), and only $pt$ is mapped to $p$, which is a~fixed point.
Also, no state except $0$ can be mapped to $p$ by $t$ because this would violate suffix-freeness; so only $p$ and $pt$ are mapped by $s$ to $p$, and $p$ has in-degree 2.

\item[(b)] All states from $Q_M$ whose mapping is different in $t$ and $s$ belong to the same orbit in $s$ of a~fixed point.

\item[(c)] $s$ does not have any cycles, but has a~fixed point $p \neq n-1$ with in-degree $2$.

\item[(d)] For each $i$ with $1 \le i < k$, there is precisely one state $q$ colliding with $pt^{i-1}$ and mapped to $pt^i$, and that state is $q=pt^{i+1}$.

This follows exactly like Property~(d) from Case~2.1.
\end{enumerate}

\textit{External injectivity}:
Here $s$ does not have a~cycle in contrast to the transformations of Case~2.1.

\textit{Internal injectivity}:
Let $\e{t}$ be any transformation that fits in this case and results in the same $s$.
From~(a) there is the unique colliding pair $\{p,pt\}$ focused to the fixed point $p$, hence $\e{p} = p$ and $p\e{t} = pt$.
Then, from~(d), for $i=1,\ldots,k-1$ state $pt^{i+1}$ is uniquely defined, hence $p\e{t}^{i+1} = pt^i$.
Since the other transitions in $s$ are defined exactly as in $t$ and $\e{t}$, we have $\e{t} = t$.

\textbf{Case~2.3}: $t$ does not fit in any of the previous cases, but there exist at least two fixed points of in-degree~1.\\
Let the two smallest valued fixed points of in-degree 1 be the states $f_1$ and $f_2$, that is,
$$f_1 = \min\{q\in Q \mid q t = q, \forall_{q'\in Q \setminus \{q\}}\ q' t \neq q\},$$
$$f_2 = \min\{q\in Q\setminus\{f_1\} \mid q t = q, \forall_{q'\in Q \setminus \{q\}}\ q' t \neq q\}.$$
Let $s$ be the transformation illustrated in Fig.~\ref{fig:case2.3} and defined by
\begin{center}
  $0 s = n-1$, $f_1 s = f_2$, $f_2 s = f_1$, $p s = f_2$,\\
  $q s = q t$ for the other states $q\in Q$.
\end{center}
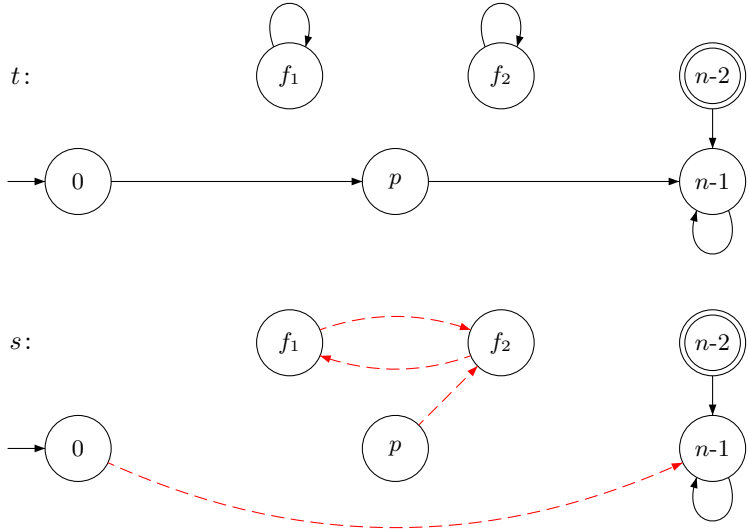
\begin{figure}[ht]
\unitlength 10pt\small
\gasset{Nh=2.5,Nw=2.5,Nmr=1.25,ELdist=0.3,loopdiam=1.5}
\begin{center}\begin{picture}(28,10)(0,-4)
\node[Nframe=n](name)(0,4){\normalsize$t\colon$}
\node(0)(2,0){0}\imark(0)
\node(p)(14,0){$p$}
\node(n-1)(26,0){$n$-$1$}
\node(n-2)(26,4){$n$-$2$}\rmark(n-2)
\node(f1)(10,4){$f_1$}
\node(f2)(18,4){$f_2$}
\drawedge(0,p){}
\drawedge(p,n-1){}
\drawedge(n-2,n-1){}
\drawloop[loopangle=270](n-1){}
\drawloop(f1){}
\drawloop(f2){}
\end{picture}
\begin{picture}(28,10)(0,-4)
\node[Nframe=n](name)(0,4){\normalsize$s\colon$}
\node(0')(2,0){0}\imark(0')
\node(p')(14,0){$p$}
\node(n-1')(26,0){$n$-$1$}
\node(n-2')(26,4){$n$-$2$}\rmark(n-2')
\node(f1')(10,4){$f_1$}
\node(f2')(18,4){$f_2$}
\drawedge[curvedepth=-3,linecolor=red,dash={.5 .25}{.25}](0',n-1'){}
\drawedge(n-2',n-1'){}
\drawloop[loopangle=270](n-1'){}
\drawedge[curvedepth=1,linecolor=red,dash={.5 .25}{.25}](f1',f2'){}
\drawedge[curvedepth=1,linecolor=red,dash={.5 .25}{.25}](f2',f1'){}
\drawedge[linecolor=red,dash={.5 .25}{.25}](p',f2'){}
\end{picture}\end{center}
\caption{Case~2.3.}\label{fig:case2.3}
\end{figure}

We observe the following properties:
\begin{enumerate}
\item[(a)] $\{p,f_2\}$ is a~colliding pair focused by $s$ to $f_2$.
This is the only pair among all colliding pairs that are focused.

\item[(b)] All states from $Q_M$ whose mapping is different in $t$ and $s$ belong to the same orbit of a~cycle in $s$.

\item[(c)] $s$ has exactly one cycle, namely $(f_1,f_2)$, and it is of length 2.
Moreover, one state in the cycle, which is $f_1$, has in-degree~1, and the other one, which is $f_2$, has in-degree~2.
\end{enumerate}

\textit{External injectivity}:
To see that $s$ is distinct from the transformations of Case~2.1, observe that in $s$ the only colliding pair is focused to $f_2$, which lies in a~cycle but is not the smallest state of the states of cycles.
On the other hand, from~(a) of Case~2.1 the transformations of that case have only one colliding pair focused to a~state in a~cycle, and this is the smallest state from the states of cycles.

Since $s$ has a~cycle, it is different from the transformations of Case~2.2.

\textit{Internal injectivity}:
Let $\e{t}$ be any transformation that fits in this case and results in the same $s$.
From~(c), there is a~single state in the unique cycle that has in-degree 2 and this is $f_1$. Hence $\e{f}_1 = f_1$, and so $\e{f}_2 = f_2$.
From~(a), the unique focused colliding pair is $\{p,f_2\}$, so $\{\e{p},\e{f}_2\}=\{p,f_2\}$ and $\e{p} = p$.
Hence $0\e{t} = 0t$, $p\e{t} = pt = n-1$, $f_1 t = f_1 \e{t} = f_1$, and $f_2 t = f_2 \e{t} = f_2$.
Since the other transitions in $s$ are defined exactly as in $t$ and $\e{t}$, we have $\e{t} = t$.

\textbf{Case~2.4}: $t$ does not fit in any of the previous cases, but there exists $x \in Q\setminus \{0\}$ of in-degree $0$ such that $xt \notin \{x,n-2,n-1\}$.\\
Let $x$ be the smallest state among the states satisfying the conditions and with the largest $\ell\ge 1$ such that $xt^\ell \notin \{xt^{\ell-1},n-2,n-1\}$.
Since $xt \notin \{x,n-2,n-1\}$ and $t$ does not have a~cycle, $x$ and $\ell$ are well defined.
We know that $xt^{\ell+1} \in \{xt^\ell,n-2,n-1\}$, and $x$ has in-degree $0$.
Within this case we have the following subcases covering all possibilities for $t$:

\textbf{Subcase~2.4.1}: $\ell\ge 2$ and $xt^{\ell+1} = n-1$.\\
Let $s$ be the transformation illustrated in Fig.~\ref{fig:subcase2.4.1} and defined by
\begin{center}
  $0 s = n-1$, $p s = xt^\ell$,\\
  $q s = q t$ for the other states $q\in Q$.
\end{center}
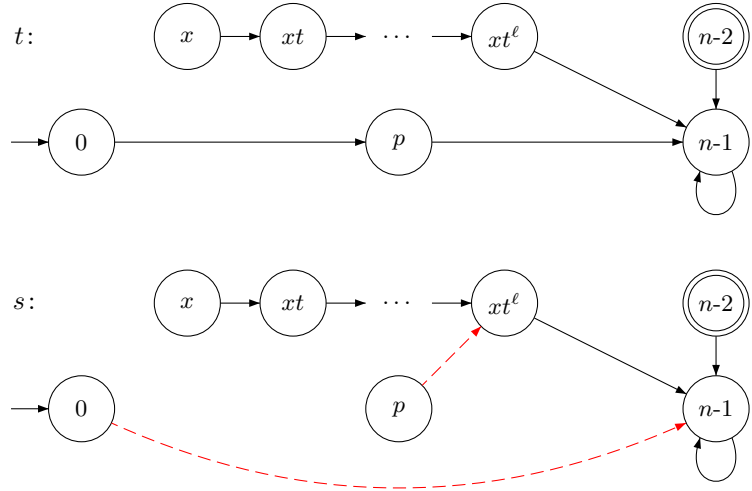
\begin{figure}[ht]
\unitlength 10pt\small
\gasset{Nh=2.5,Nw=2.5,Nmr=1.25,ELdist=0.3,loopdiam=1.5}
\begin{center}\begin{picture}(28,10)(0,-4)
\node[Nframe=n](name)(0,4){\normalsize$t\colon$}
\node(0)(2,0){0}\imark(0)
\node(p)(14,0){$p$}
\node(n-1)(26,0){$n$-$1$}
\node(n-2)(26,4){$n$-$2$}\rmark(n-2)
\node(x)(6,4){$x$}
\node(xt)(10,4){$xt$}
\node[Nframe=n](xdots)(14,4){$\dots$}
\node(xt^ell)(18,4){$xt^\ell$}
\drawedge(0,p){}
\drawedge(p,n-1){}
\drawedge(n-2,n-1){}
\drawloop[loopangle=270](n-1){}
\drawedge(x,xt){}
\drawedge(xt,xdots){}
\drawedge(xdots,xt^ell){}
\drawedge(xt^ell,n-1){}
\end{picture}
\begin{picture}(28,10)(0,-4)
\node[Nframe=n](name)(0,4){\normalsize$s\colon$}
\node(0')(2,0){0}\imark(0')
\node(p')(14,0){$p$}
\node(n-1')(26,0){$n$-$1$}
\node(n-2')(26,4){$n$-$2$}\rmark(n-2')
\node(x')(6,4){$x$}
\node(xt')(10,4){$xt$}
\node[Nframe=n](xdots')(14,4){$\dots$}
\node(xt^ell')(18,4){$xt^\ell$}
\drawedge[curvedepth=-3,linecolor=red,dash={.5 .25}{.25}](0',n-1'){}
\drawedge(n-2',n-1'){}
\drawloop[loopangle=270](n-1'){}
\drawedge[linecolor=red,dash={.5 .25}{.25}](p,xt^ell'){}
\drawedge(x',xt'){}
\drawedge(xt',xdots'){}
\drawedge(xdots',xt^ell'){}
\drawedge(xt^ell',n-1'){}
\end{picture}\end{center}
\caption{Subcase~2.4.1.}\label{fig:subcase2.4.1}
\end{figure}

We observe the following properties:
\begin{enumerate}
\item[(a)] $\{xt^{\ell-1}, p\}$ is a~colliding pair focused by $s$ to $xt^\ell$.

\item[(b)] $p$ is the only state from $Q_M$ whose mapping is different in $t$ and $s$, and $p$ is mapped to a~state mapped to $n-1$.

\item[(c)] $s$ does not have any cycles.
\end{enumerate}

\textit{External injectivity}:
Since $s$ does not have any cycles, $s$ is different from the transformations of Case~2.1 and Case~2.3.

From~(a), we have a~focused colliding pair in the orbit of $n-1$. Thus, $s$ is different from the transformations of Case~2.2, where all states in focused colliding pairs are in the orbit of a~fixed point different from $n-1$ (Property~(b) of Case~2.2).

\textit{Internal injectivity}:
Let $\e{t}$ be any transformation that fits in this subcase and results in the same $s$.
From~(b), all colliding pairs that are focused contain $p$.
If there are at least two such pairs, then $p$ is uniquely determined as the unique common state.
If there is only one such pair, then by~(a) it is $\{xt^{\ell-1}, p\}$, and $p$ is determined as the state of in-degree $0$, since $xt^{\ell-1}$ has in-degree $\ge 1$.
Hence, $\e{p} = p$, and since the other transitions in $s$ are defined exactly as in $t$ and $\e{t}$, we have $\e{t} = t$.

\textbf{Subcase~2.4.2}: $\ell=1$, $xt^2 = n-1$, and $xt$ has in-degree $>1$.\\
Let $y$ be the smallest state different from $x$ and such that $yt = xt$.
Note that $y$ has in-degree $0$, as otherwise, it would contradict the choice of $x$ since there would be a~state satisfying the conditions for $x$ with a~larger $\ell$.
Also, $x < y$, as otherwise we would choose $y$ as $x$.
Let $s$ be the transformation illustrated in Fig.~\ref{fig:subcase2.4.2} and defined by
\begin{center}
  $0 s = n-1$, $p s = y$,\\
  $(xt) s = x$, $x s = y$,\\
  $q s = q t$ for the other states $q\in Q$.
\end{center}
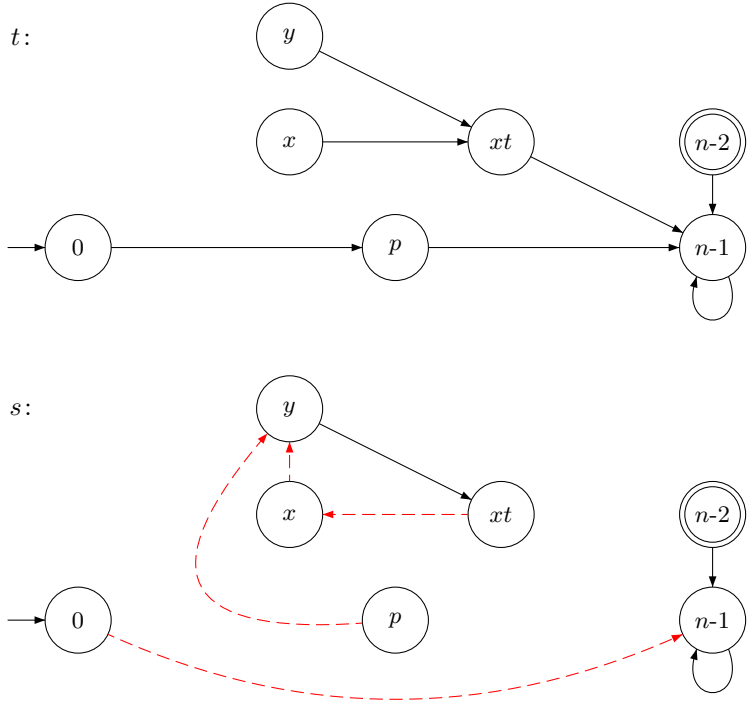
\begin{figure}[ht]
\unitlength 10pt\small
\gasset{Nh=2.5,Nw=2.5,Nmr=1.25,ELdist=0.3,loopdiam=1.5}
\begin{center}\begin{picture}(28,14)(0,-4)
\node[Nframe=n](name)(0,8){\normalsize$t\colon$}
\node(0)(2,0){0}\imark(0)
\node(p)(14,0){$p$}
\node(n-1)(26,0){$n$-$1$}
\node(n-2)(26,4){$n$-$2$}\rmark(n-2)
\node(x)(10,4){$x$}
\node(xt)(18,4){$xt$}
\node(y)(10,8){$y$}
\drawedge(0,p){}
\drawedge(p,n-1){}
\drawedge(n-2,n-1){}
\drawloop[loopangle=270](n-1){}
\drawedge(x,xt){}
\drawedge(xt,n-1){}
\drawedge(y,xt){}
\end{picture}
\begin{picture}(28,14)(0,-4)
\node[Nframe=n](name)(0,8){\normalsize$s\colon$}
\node(0')(2,0){0}\imark(0')
\node(p')(14,0){$p$}
\node(n-1')(26,0){$n$-$1$}
\node(n-2')(26,4){$n$-$2$}\rmark(n-2')
\node(x')(10,4){$x$}
\node(xt')(18,4){$xt$}
\node(y')(10,8){$y$}
\drawedge[curvedepth=-3,linecolor=red,dash={.5 .25}{.25}](0',n-1'){}
\drawedge(n-2',n-1'){}
\drawloop[loopangle=270](n-1'){}
\drawedge[curvedepth=6,linecolor=red,dash={.5 .25}{.25}](p',y'){}
\drawedge[linecolor=red,dash={.5 .25}{.25}](xt',x'){}
\drawedge[linecolor=red,dash={.5 .25}{.25}](x',y'){}
\drawedge(y',xt'){}
\end{picture}\end{center}
\caption{Subcase~2.4.2.}\label{fig:subcase2.4.2}
\end{figure}

We observe the following properties.
\begin{enumerate}
\item[(a)] $\{p,x\}$ is a~colliding pair focused by $s$ to $y$.

\item[(b)] All states from $Q_M$ whose mapping is different in $t$ and $s$ belong to the same orbit of a~cycle of length $3$ in $s$.

\item[(c)] $s$ contains exactly one cycle, namely $(x,y,xt)$.
Furthermore, $y$ has in-degree $2$ and is preceded in this cycle by $x$ of in-degree $1$.
\end{enumerate}

\textit{External injectivity}:
To see that $s$ is different from the transformations of Case~2.1, observe that by~(a) we have a~colliding pair focused to $y$, which is from a~cycle, but is not the smallest state from the states in cycles since $x < y$.

On the other hand, in Case~2.1 all colliding pairs focused to a~state in a~cycle are focused to the smallest state of all states in cycles (Property~(a) of Case~2.1). In Case~2.3, the transformation has a~cycle, but this cycle has length $2$.

Since $s$ has a~cycle, it is different from the transformations of Case~2.2 and Subcase~2.4.1 (recall that fixed points have not been defined as cycles).

\textit{Internal injectivity}:
Let $\e{t}$ be any transformation that fits in this subcase and results in the same $s$.
From~(c), in $s$ we have a~unique cycle of length 3, and this cycle is $(x,y,xt)$.
Since $y$ is uniquely determined as the state of in-degree 2 preceded in the cycle by the state of in-degree 1, we have $\e{y} = y$.
Then also $\e{x} = x$ and $x\e{t} = xt$. State $p$ is the only state outside the cycle mapped to $y$, hence $\e{p} = p$.
We have $0t = 0\e{t} = p$, $pt = p\e{t} = n-1$, and $xt^2 = x{\e{t}}^2 = n-1$.
Since the other transitions in $s$ are defined exactly as in $t$ and $\e{t}$, we have $\e{t} = t$.

\textbf{Subcase~2.4.3}: $\ell=1$, $xt^2 = n-1$, and $xt$ has in-degree $1$.\\
We split the subcase into two subsubcases: (i) $p < xt$ and (ii) $p > xt$.
Let $s$ be the transformation illustrated in Fig.~\ref{fig:subcase2.4.3} and defined by
\begin{center}
  $0 s = n-1$, $p s = x$,\\
  $(xt) s = x$,\\
  $x s = n-2$ (i), $x s = n-1$ (ii),\\
  $q s = q t$ for the other states $q\in Q$.
\end{center}
\begin{figure}[ht]
\unitlength 10pt\small
\gasset{Nh=2.5,Nw=2.5,Nmr=1.25,ELdist=0.3,loopdiam=1.5}
\begin{center}\begin{picture}(28,10)(0,-4)
\node[Nframe=n](name)(0,4){\normalsize$t\colon$}
\node(0)(2,0){0}\imark(0)
\node(p)(14,0){$p$}
\node(n-1)(26,0){$n$-$1$}
\node(n-2)(26,4){$n$-$2$}\rmark(n-2)
\node(x)(10,4){$x$}
\node(xt)(18,4){$xt$}
\drawedge(0,p){}
\drawedge(p,n-1){}
\drawedge(n-2,n-1){}
\drawloop[loopangle=270](n-1){}
\drawedge(x,xt){}
\drawedge(xt,n-1){}
\end{picture}
\begin{picture}(28,11)(0,-4)
\node[Nframe=n](name)(0,4){\normalsize$s\colon$}
\node(0')(2,0){0}\imark(0')
\node(p')(14,0){$p$}
\node(n-1')(26,0){$n$-$1$}
\node(n-2')(26,4){$n$-$2$}\rmark(n-2')
\node(x')(10,4){$x$}
\node(xt')(18,4){$xt$}
\drawedge[curvedepth=-3,linecolor=red,dash={.5 .25}{.25}](0',n-1'){}
\drawedge(n-2',n-1'){}
\drawloop[loopangle=270](n-1'){}
\drawedge[linecolor=red,dash={.5 .25}{.25}](p',x'){}
\drawedge[linecolor=red,dash={.5 .25}{.25}](xt',x'){}
\drawedge[curvedepth=2,linecolor=red,dash={.1 .1}{.1}](x',n-2'){(i)}
\drawedge[curvedepth=-.5,linecolor=red,dash={.1 .1}{.1},ELside=r](x',n-1'){(ii)}
\end{picture}\end{center}
\caption{Subcase~2.4.3.}\label{fig:subcase2.4.3}
\end{figure}
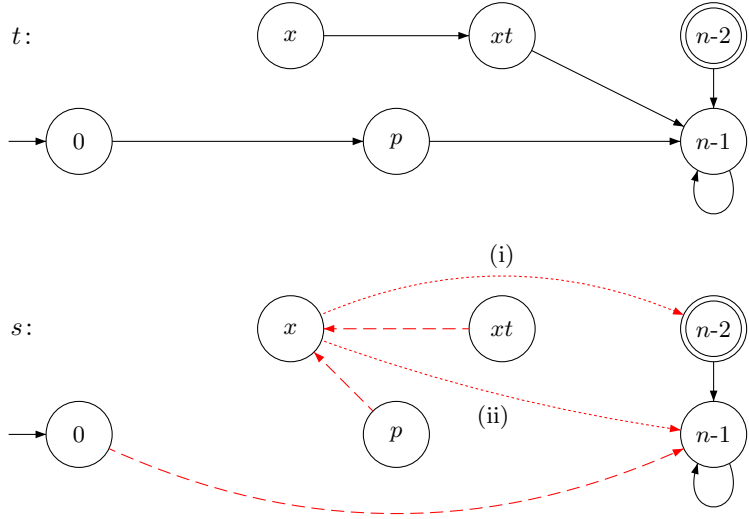

We observe the following properties:
\begin{enumerate}
\item[(a)] $\{p,xt\}$ is a~colliding pair focused by $s$ to $x$.
Both states from this pair have in-degree 0.

\item[(b)] All states from $Q_M$ whose mapping is different in $t$ and $s$ are from the orbit of $n-1$,
and $p$ and $xt$ are the only such states that are not mapped to $n-2$ nor to $n-1$.

\item[(c)] $s$ does not have any cycles.
\end{enumerate}

\textit{External injectivity}:
Since $s$ does not have any cycles, it is different from the transformations of Case~2.1, Case~2.3, and Subcase~2.4.2.

By~(b) all colliding pairs that are focused have states from the orbit of $n-1$, whereas the transformations of Case~2.2 focus a~colliding pair to a~fixed point.

Let $\e{t}$ be a~transformation that fits in Subcase~2.4.1 and results in the same $s$.
By Lemma~\ref{lem:orbits}, the orbits from Properties~(b) for both $t$ and $\e{t}$ must be the same, so $x = \e{x}\e{t}^{\e{\ell}}$.
But in $s$, to $x$ only states of in-degree 0 are mapped, whereas to $\e{x}\e{t}^{\e{\ell}}$ state $\e{x}\e{t}^{\e{\ell-1}}$ is mapped, which has in-degree at least 1.

\textit{Internal injectivity}:
Let $\e{t}$ be any transformation that fits in this subcase and results in the same $s$.
From~(a) and~(b), $\{p,xt\}$ is the unique colliding pair focused to a~state different from $n-2$; hence $\{p,xt\} = \{\e{p},\e{x}\e{t}\}$.
The pair is focused to $x$, hence $\e{x} = x$.
If $x$ is mapped to $n-2$, then we have subsubcase~(i) and $p$ is the smaller state in the colliding pair.
If $x$ is mapped to $n-1$, then we have subsubcase~(ii) and $p$ is the larger state in the colliding pair.
Hence $p = \e{p}$ and $xt = x\e{t}$. We have $0t = 0\e{t} = p$ and $(xt)t = (xt)\e{t} = n-1$.
Since the other transitions in $s$ are defined exactly as in $t$ and $\e{t}$, we have $\e{t} = t$.

\textbf{Subcase 2.4.4}: $xt^{\ell+1} = n-2$.\\
Let $s$ be the transformation illustrated in Fig.~\ref{fig:subcase2.4.4} and defined by
\begin{center}
  $0 s = n-1$, $p s = n-2$,\\
  $q s = q t$ for the other states $q\in Q$.
\end{center}
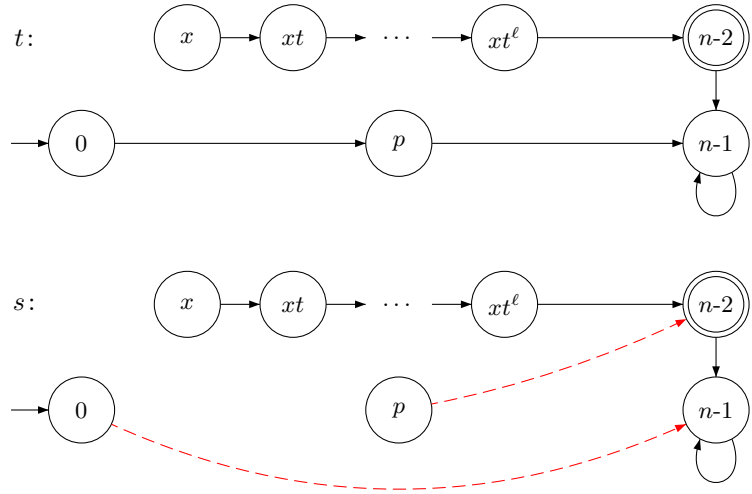
\begin{figure}[ht]
\unitlength 10pt\small
\gasset{Nh=2.5,Nw=2.5,Nmr=1.25,ELdist=0.3,loopdiam=1.5}
\begin{center}\begin{picture}(28,10)(0,-4)
\node[Nframe=n](name)(0,4){\normalsize$t\colon$}
\node(0)(2,0){0}\imark(0)
\node(p)(14,0){$p$}
\node(n-1)(26,0){$n$-$1$}
\node(n-2)(26,4){$n$-$2$}\rmark(n-2)
\node(x)(6,4){$x$}
\node(xt)(10,4){$xt$}
\node[Nframe=n](xdots)(14,4){$\dots$}
\node(xt^ell)(18,4){$xt^\ell$}
\drawedge(0,p){}
\drawedge(p,n-1){}
\drawedge(n-2,n-1){}
\drawloop[loopangle=270](n-1){}
\drawedge(x,xt){}
\drawedge(xt,xdots){}
\drawedge(xdots,xt^ell){}
\drawedge(xt^ell,n-2){}
\end{picture}
\begin{picture}(28,10)(0,-4)
\node[Nframe=n](name)(0,4){\normalsize$s\colon$}
\node(0')(2,0){0}\imark(0')
\node(p')(14,0){$p$}
\node(n-1')(26,0){$n$-$1$}
\node(n-2')(26,4){$n$-$2$}\rmark(n-2')
\node(x')(6,4){$x$}
\node(xt')(10,4){$xt$}
\node[Nframe=n](xdots')(14,4){$\dots$}
\node(xt^ell')(18,4){$xt^\ell$}
\drawedge[curvedepth=-3,linecolor=red,dash={.5 .25}{.25}](0',n-1'){}
\drawedge(n-2',n-1'){}
\drawloop[loopangle=270](n-1'){}
\drawedge(x',xt'){}
\drawedge(xt',xdots'){}
\drawedge(xdots',xt^ell'){}
\drawedge(xt^ell',n-2'){}
\drawedge[curvedepth=-.5,linecolor=red,dash={.5 .25}{.25}](p',n-2'){}
\end{picture}\end{center}
\caption{Subcase~2.4.4.}\label{fig:subcase2.4.4}
\end{figure}

We observe the following properties:
\begin{enumerate}
\item[(a)] $\{xt^\ell, p\}$ is a~colliding pair focused by $s$ to $n-2$.

\item[(b)] $p$ is the only state from $Q_M$ whose mapping is different in $t$ and $s$.

\item[(c)] $s$ does not contain any cycles.
\end{enumerate}

\textit{External injectivity}:
Since $s$ does not contain any cycles, it is different from the transformations of Case~2.1, Case~2.3, and Subcase~2.4.2.

From~(b), all focused colliding pairs contain $p$ and so are mapped to $n-2$ in $s$.
Hence, $s$ is different from the transformations of Case~2.2, Subcase~2.4.1, and Subcase~2.4.3.

\textit{Internal injectivity}:
Let $\e{t}$ be any transformation that fits in this subcase and results in the same $s$.
If there are two focused colliding pairs, then $p$ is uniquely determined as the common state in these pairs.
If there is only one such pair, then $p$ is the state of in-degree $0$, as the other state is $xt^\ell$, which has in-degree $\ge 1$.
Hence, $\e{p} = p$. We have $0t = 0\e{t} = p$ and $pt = p\e{t} = n-1$.
Since the other transitions in $s$ are defined exactly as in $t$ and $\e{t}$, we have $\e{t} = t$.

\textbf{Subcase 2.4.5}: $xt^{\ell+1} = xt^\ell$.\\
Let $s$ be the transformation illustrated in Fig.~\ref{fig:subcase2.4.5} and defined by
\begin{center}
  $0 s = n-1$, $p s = xt^\ell$,\\
  $q s = q t$ for the other states $q\in Q$.
\end{center}
\begin{figure}[ht]
\unitlength 10pt\small
\gasset{Nh=2.5,Nw=2.5,Nmr=1.25,ELdist=0.3,loopdiam=1.5}
\begin{center}\begin{picture}(28,11)(0,-4)
\node[Nframe=n](name)(0,6){\normalsize$t\colon$}
\node(0)(2,0){0}\imark(0)
\node(p)(14,0){$p$}
\node(n-1)(26,0){$n$-$1$}
\node(n-2)(26,4){$n$-$2$}\rmark(n-2)
\node(x)(6,4){$x$}
\node(xt)(10,4){$xt$}
\node[Nframe=n](xdots)(14,4){$\dots$}
\node(xt^ell)(18,4){$xt^\ell$}
\drawedge(0,p){}
\drawedge(p,n-1){}
\drawedge(n-2,n-1){}
\drawloop[loopangle=270](n-1){}
\drawedge(x,xt){}
\drawedge(xt,xdots){}
\drawedge(xdots,xt^ell){}
\drawloop(xt^ell){}
\end{picture}
\begin{picture}(28,11)(0,-4)
\node[Nframe=n](name)(0,6){\normalsize$s\colon$}
\node(0')(2,0){0}\imark(0')
\node(p')(14,0){$p$}
\node(n-1')(26,0){$n$-$1$}
\node(n-2')(26,4){$n$-$2$}\rmark(n-2')
\node(x')(6,4){$x$}
\node(xt')(10,4){$xt$}
\node[Nframe=n](xdots')(14,4){$\dots$}
\node(xt^ell')(18,4){$xt^\ell$}
\drawedge[curvedepth=-3,linecolor=red,dash={.5 .25}{.25}](0',n-1'){}
\drawedge(n-2',n-1'){}
\drawloop[loopangle=270](n-1'){}
\drawedge(x',xt'){}
\drawedge(xt',xdots'){}
\drawedge(xdots',xt^ell'){}
\drawloop(xt^ell'){}
\drawedge[linecolor=red,dash={.5 .25}{.25}](p',xt^ell'){}
\end{picture}\end{center}
\caption{Subcase~2.4.5.}\label{fig:subcase2.4.5}
\end{figure}
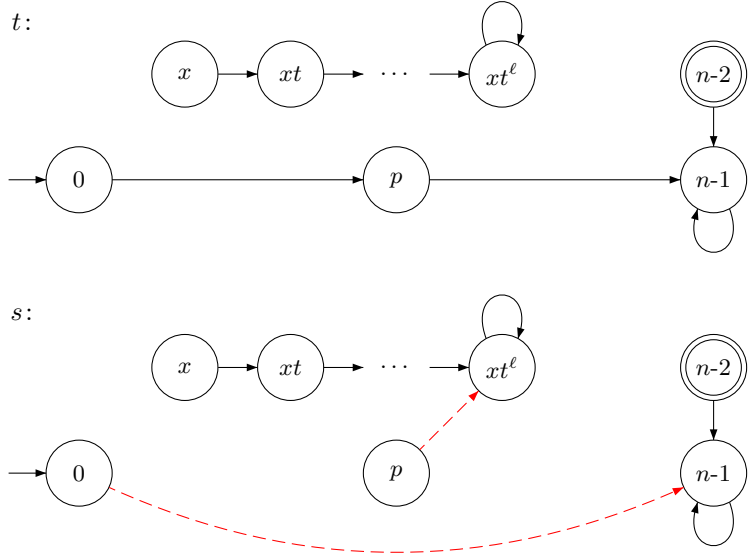

We observe the following properties:
\begin{enumerate}
\item[(a)] $\{p, xt^\ell\}$ is a~colliding pair focused by $s$ to the fixed point $xt^\ell$, which has in-degree at least $3$.

\item[(b)] $p$ is the only state from $Q_M$ whose mapping is different in $t$ and $s$.

\item[(c)] $s$ does not contain any cycles.
\end{enumerate}

\textit{External injectivity}:
Since $s$ does not contain any cycles, it is different from the transformations of Case~2.1, Case~2.3, and Subcase~2.4.2.

Let $\e{t}$ be a~transformation that fits in Case~2.2 and results in the same $s$.
By Lemma~\ref{lem:orbits}, the orbits from Properties~(b) for both $t$ and $\e{t}$ must be the same, so $xt^\ell = \e{p}$.
But $xt^\ell$ has in-degree at least 3, whereas $\e{p}$ has in-degree 2, which yields a~contradiction.

Since the orbits from Properties~(b) of the transformations of Subcase~2.4.1, Subcase~2.4.3, and Subcase~2.4.4 contain $n-1$, by Lemma~\ref{lem:orbits} they are different from $s$.

\textit{Internal injectivity}:
Let $\e{t}$ be any transformation that fits in this subcase and results in the same $s$.
By Lemma~\ref{lem:orbits}, the orbits from Properties~(b) for both $t$ and $\e{t}$ must be the same, so we obtain that $xt^\ell = \e{x}\e{t}^{\e{\ell}}$.
If $t \neq \e{t}$, then by~(b) $p \neq \e{p}$, and also $p\e{t} = \e{p}t = xt^\ell$, as otherwise $t$ and $\e{t}$ would not result in the same $s$.
Then, $\{\e{p},xt^\ell\}$ is a~colliding pair because of $\e{t}$.
But $\e{p}t = (xt^{\ell})t = xt^\ell$, so this colliding pair is focused by $t$.
Hence, it must be that $t = \e{t}$.

\textbf{Case 2.5}: $t$ does not fit in any of the previous cases.\\
First we observe that there exists exactly one fixed point $f \neq n-1$, and every state $q \in Q \setminus \{0,f\}$ is mapped either to $n-2$ or to $n-1$:
All transformations that fit in Supercase~2 and have a~cycle or with $p t \neq n-1$ are covered in Case~2.1 or~2.2.
If there are two fixed points of in-degree $1$ then $t$ is covered in Case~2.3.
If there is a~state $x \in Q\setminus \{0\}$ such that $xt \notin \{x,n-2,n-1\}$, then, since there are no cycles, there exists such a~state of in-degree $0$, thus $t$ is covered in Case~2.5.
Hence, every state $q\in Q\setminus \{0\}$ must either be a~fixed point or be mapped to $n-2$ or $n-1$, and there can be at most one fixed point.
If there is no fixed point, then $t \in \Wbf(n)$ (transformation of Type~3) and so it falls into Supercase~1.

\textbf{Subcase~2.5.1}: There are at least two states from all $r_1,r_2,\ldots,r_u \in Q \setminus \{p\}$ such that $r_i t = n-1$ for all $i$.\\
Assume that $r_1 < r_2 < \dots < r_u$.
Let $s$ be the transformation illustrated in Fig.~\ref{fig:subcase2.5.1} and defined by
\begin{center}
  $0 s = n-1$, $p s = f$,\\
  $r_i s = r_{i+1}$ for $1\le i\le u-1$,\\
  $r_u s = r_1$,\\
  $q s = q t$ for the other states $q\in Q$.
\end{center}
\begin{figure}[ht]
\unitlength 10pt\small
\gasset{Nh=2.5,Nw=2.5,Nmr=1.25,ELdist=0.3,loopdiam=1.5}
\begin{center}\begin{picture}(28,11)(0,-4)
\node[Nframe=n](name)(0,6){\normalsize$t\colon$}
\node(0)(2,0){0}\imark(0)
\node(p)(14,0){$p$}
\node(n-1)(26,0){$n$-$1$}
\node(n-2)(26,4){$n$-$2$}\rmark(n-2)
\node(f)(8,4){$f$}
\node(r1)(14,4){$r_1$}
\node[Nframe=n](rdots)(18,4){$\dots$}
\node(ru)(22,4){$r_u$}
\drawedge(0,p){}
\drawedge(p,n-1){}
\drawedge(n-2,n-1){}
\drawloop[loopangle=270](n-1){}
\drawloop(f){}
\drawedge[curvedepth=-.2](r1,n-1){}
\drawedge[curvedepth=0,exo=.2](rdots,n-1){}
\drawedge[curvedepth=0,exo=.5](ru,n-1){}
\end{picture}
\begin{picture}(28,11)(0,-4)
\node[Nframe=n](name)(0,6){\normalsize$s\colon$}
\node(0')(2,0){0}\imark(0')
\node(p')(14,0){$p$}
\node(n-1')(26,0){$n$-$1$}
\node(n-2')(26,4){$n$-$2$}\rmark(n-2')
\node(f')(8,4){$f$}
\node(r1')(14,4){$r_1$}
\node[Nframe=n](rdots')(18,4){$\dots$}
\node(ru')(22,4){$r_u$}
\drawedge[curvedepth=-3,linecolor=red,dash={.5 .25}{.25}](0',n-1'){}
\drawedge(n-2',n-1'){}
\drawloop[loopangle=270](n-1'){}
\drawloop(f'){}
\drawedge[linecolor=red,dash={.5 .25}{.25}](p',f'){}
\drawedge[linecolor=red,dash={.5 .25}{.25}](r1',rdots'){}
\drawedge[linecolor=red,dash={.5 .25}{.25}](rdots',ru'){}
\drawedge[curvedepth=-2,linecolor=red,dash={.5 .25}{.25}](ru',r1'){}
\end{picture}\end{center}
\caption{Subcase~2.5.1.}\label{fig:subcase2.5.1}
\end{figure}
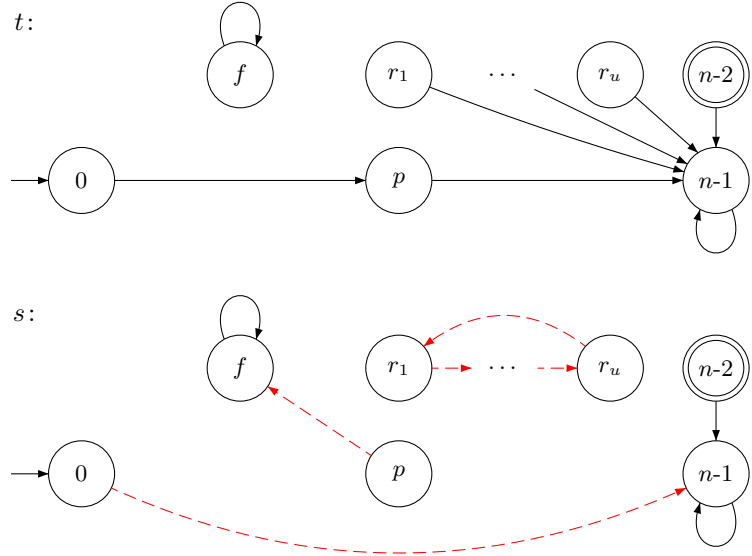

We observe the following properties:
\begin{enumerate}
\item[(a)] $\{p,f\}$ is a~colliding pair focused by $s$ to the fixed point $f$.
This is the only colliding pair that is focused by $s$.

\item[(c)] $s$ contains exactly one cycle.
\end{enumerate}

\textit{External injectivity}:
Since $s$ has a~cycle, it is different from the transformations of Case~2.2, Subcase~2.4.1, Subcase~2.4.3, Subcase~2.4.4, and Subcase~2.4.5.

From~(a) and~(c), $s$ has a~cycle and focuses a~colliding pair to a~state whose orbit is not the orbit of a~cycle.
Hence, $s$ is different from the transformations of Case~2.1, Case~2.3, and of Subcase~2.4.2,
where all colliding pairs that are focused by these transformations have states from the orbit of a~cycle (Properties~(b) of these (sub)cases).

\textit{Internal injectivity}:
Let $\e{t}$ be any transformation that fits in this subcase and results in the same $s$.
By~(a), $\{p,f\}$ is the unique colliding pair that is focused to the fixed point $f$, so $\e{p}=p$ and $\e{f}=f$.
Also, there is exactly one cycle formed by the states $r_i$, so $(r_1,r_2,\ldots,r_u) = (\e{r_1},\e{r_2},\ldots,\e{r_u})$.
It follows that $0t = 0\e{t} = p$, $ft = f\e{t} = f$, and $r_i t = \e{r_i} \e{t} = n-1$ for all $i$.
Since the other transitions in $s$ are defined exactly as in $t$ and $\e{t}$, we have $\e{t} = t$.

\textbf{Subcase~2.5.2}: $t$ does not fit in Subcase~2.5.1.\\
Because $n \ge 8$, we know that $Q \setminus \{0,p,f,n-2,n-1\}$ contains at least three states.
Since $t$ does not fit in Subcase~2.5.1, we have at least two states from all $q_1,q_2,\ldots,q_v \in Q_M \setminus \{p\}$ such that $q_i t = n-2$.
Assume that $q_1 < q_2 < \ldots < q_v$.
Let $s$ be the transformation illustrated in Fig.~\ref{fig:subcase2.5.2} and defined by
\begin{center}
  $0 s = n-1$, $p s = f$,\\
  $q_i s = q_{i-1}$ for $2\le i\le v$,\\
  $q_1 s = q_v$,\\
  $q s = q t$ for the other states $q\in Q$.
\end{center}
\begin{figure}[ht]
\unitlength 10pt\small
\gasset{Nh=2.5,Nw=2.5,Nmr=1.25,ELdist=0.3,loopdiam=1.5}
\begin{center}\begin{picture}(28,11)(0,-4)
\node[Nframe=n](name)(0,6){\normalsize$t\colon$}
\node(0)(2,0){0}\imark(0)
\node(p)(14,0){$p$}
\node(n-1)(26,0){$n$-$1$}
\node(n-2)(26,4){$n$-$2$}\rmark(n-2)
\node(f)(8,4){$f$}
\node(q1)(14,4){$q_1$}
\node[Nframe=n](qdots)(18,4){$\dots$}
\node(qv)(22,4){$q_v$}
\drawedge(0,p){}
\drawedge(p,n-1){}
\drawedge(n-2,n-1){}
\drawloop[loopangle=270](n-1){}
\drawloop(f){}
\drawedge[curvedepth=-3,exo=1](q1,n-2){}
\drawedge[curvedepth=-2](qdots,n-2){}
\drawedge[curvedepth=0](qv,n-2){}
\end{picture}
\begin{picture}(28,11)(0,-4)
\node[Nframe=n](name)(0,6){\normalsize$s\colon$}
\node(0')(2,0){0}\imark(0')
\node(p')(14,0){$p$}
\node(n-1')(26,0){$n$-$1$}
\node(n-2')(26,4){$n$-$2$}\rmark(n-2')
\node(f')(8,4){$f$}
\node(q1')(14,4){$q_1$}
\node[Nframe=n](qdots')(18,4){$\dots$}
\node(qv')(22,4){$q_v$}
\drawedge[curvedepth=-3,linecolor=red,dash={.5 .25}{.25}](0',n-1'){}
\drawedge(n-2',n-1'){}
\drawloop[loopangle=270](n-1'){}
\drawloop(f'){}
\drawedge[linecolor=red,dash={.5 .25}{.25}](p',f'){}
\drawedge[curvedepth=2,linecolor=red,dash={.5 .25}{.25}](q1',qv'){}
\drawedge[linecolor=red,dash={.5 .25}{.25}](qdots',q1'){}
\drawedge[linecolor=red,dash={.5 .25}{.25}](qv',qdots'){}
\end{picture}\end{center}
\caption{Subcase~2.5.2.}\label{fig:subcase2.5.2}
\end{figure}
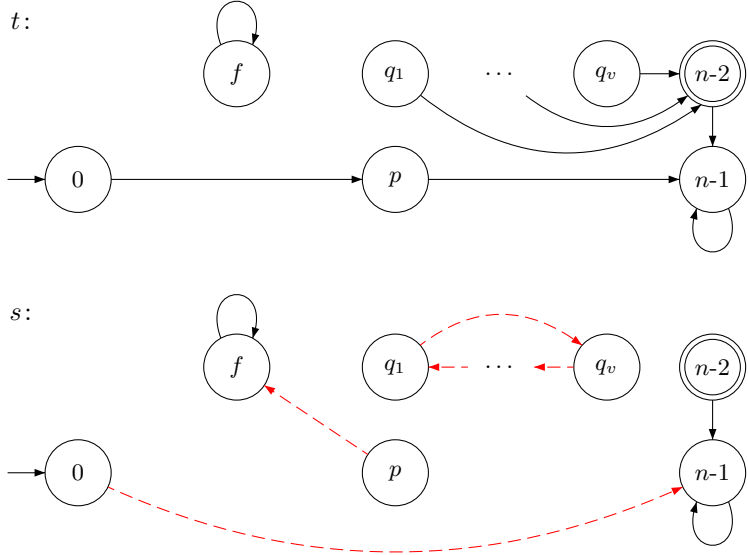

We observe the following properties:
\begin{enumerate}
\item[(a)] $\{p, f\}$ is a~colliding pair focused by $s$ to the fixed point $f$.
This is the only colliding pair that is focused by $s$.

\item[(c)] $s$ contains exactly one cycle.
\end{enumerate}

\textit{External injectivity}:
In the same way as in Subcase~2.5.1, $s$ is different from the transformations of Cases~2.1--2.4.

Now suppose that the same transformation $s$ is obtained in Subcase~2.5.1.
Since the unique cycles in both subcases go in opposite directions w.r.t. the ordering of the states, if they are equal then they must be of length $2$.
But then, since $n \ge 8$, we have at least one state in $Q_M$ being mapped to $n-1$ in $t$, and also in $s$.
But since $s$ is also obtained in Subcase~2.5.1, there are no such states besides $0$, $n-2$, and $n-1$, which yields a~contradiction.

\textit{Internal injectivity}:
Let $\e{t}$ be any transformation that fits in this subcase and results in the same $s$.
It follows in the same way as in~Subcase~2.5.1, that we have $0t = 0\e{t} = p$, $ft = f\e{t} = f$, and $q_i t = \e{q_i} \e{t} = n-2$ for all $i$.
Since the other transitions in $s$ are defined exactly as in $t$ and $\e{t}$, we have $\e{t} = t$.

\textbf{Supercase 3:} $t \notin \Wbf(n)$ and $pt^{k+1} = n-2$.\\
Here we have the chain
$$0 \stackrel{t}{\rightarrow} p \stackrel{t}{\rightarrow} pt \stackrel{t}{\rightarrow} \dots \stackrel{t}{\rightarrow} pt^k \stackrel{t}{\rightarrow} n-2 \stackrel{t}{\rightarrow} n-1.$$
We will always assign transformations $s$ such that $s$ together with $t$ generate a~transformation that focuses a~colliding pair, which distinguishes such transformations $s$ from those of Supercase~1.
Moreover, we will always have $0 s = n-2$, to distinguish $s$ from the transformations of Supercase~2.

For all the cases of Supercase~3, let $q_1,q_2,\ldots,q_v \in Q_M \setminus \{p\}$ be all the states such that $q_i t = n-2$, for all $i$.
Without loss of generality, we assume that $q_1 < q_2 < \dots < q_v$.

In contrast to Supercase~2, we have an~additional difficulty in constructions of $s$, which is that no state can be mapped to $n-2$ except state $0$.
On the other hand, the chains going through a~state $q_i$ and ending in $n-2$ are of length at most $k+1$ (i.e.\ they contain at most $k+2$ states including $n-2$).
Otherwise, if there is such a~chain of length at least $k+2$, then there would exist a~state $q \in Q_M$ such that $qt^{k+1}=n-2$, which means that the pair $\{p,qt\}$ is colliding because of $t$ and focused by $t^k$ to $n-1$, contradicting suffix-freeness.
This fact will allow give us more knowledge about the transformation $t$, helping to construct a~suitable $s$.
In particular, when $k=0$, all states $q_i$ have in-degree $0$.

We have the following cases covering all possibilities for $t$:

\textbf{Case 3.1}: $k=0$ and $t$ has a~cycle.\\
Let $r$ be the minimal among the states that appear in cycles of $t$, that is,
$$r = \min\{q\in Q \mid \text{q is in a~cycle of } t\}.$$
Let $s$ be the transformation illustrated in Fig.~\ref{fig:case3.1} and defined by
\begin{center}
  $0 s = n-2$, $p s = r$,\\
  $q_i s = p$ for $1\le i\le v$,\\
  $qs = qt$ for the other states $q\in Q$.
\end{center}
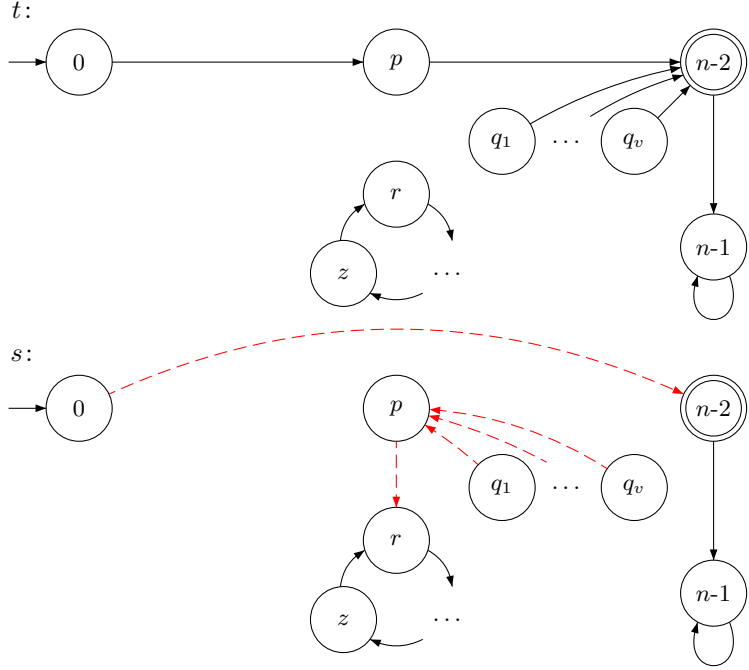
\begin{figure}[ht]
\unitlength 10pt\small
\gasset{Nh=2.5,Nw=2.5,Nmr=1.25,ELdist=0.3,loopdiam=1.5}
\begin{center}\begin{picture}(28,12)(0,-2)
\node[Nframe=n](name)(0,9){\normalsize$t\colon$}
\node(0)(2,7){0}\imark(0)
\node(p)(14,7){$p$}
\node(n-1)(26,0){$n$-$1$}
\node(n-2)(26,7){$n$-$2$}\rmark(n-2)
\node(z)(12,-1){$z$}
\node(r)(14,2){$r$}
\node[Nframe=n](rdots)(16,-1){$\dots$}
\node(q1)(18,4){$q_1$}
\node[Nframe=n](qdots)(20.5,4){$\dots$}
\node(qv)(23,4){$q_v$}
\drawedge(0,p){}
\drawedge(p,n-2){}
\drawedge(n-2,n-1){}
\drawloop[loopangle=270](n-1){}
\drawedge[curvedepth=.5](q1,n-2){}
\drawedge[curvedepth=.6,sxo=-.5,exo=1.5](qdots,n-2){}
\drawedge[curvedepth=0](qv,n-2){}
\drawedge[curvedepth=1](z,r){}
\drawedge[curvedepth=1](r,rdots){}
\drawedge[curvedepth=1](rdots,z){}
\end{picture}
\begin{picture}(28,12)(0,-1)
\node[Nframe=n](name)(0,9){\normalsize$s\colon$}
\node(0')(2,7){0}\imark(0')
\node(p')(14,7){$p$}
\node(n-1')(26,0){$n$-$1$}
\node(n-2')(26,7){$n$-$2$}\rmark(n-2')
\node(z')(12,-1){$z$}
\node(r')(14,2){$r$}
\node[Nframe=n](rdots')(16,-1){$\dots$}
\node(q1')(18,4){$q_1$}
\node[Nframe=n](qdots')(20.5,4){$\dots$}
\node(qv')(23,4){$q_v$}
\drawedge[curvedepth=3,linecolor=red,dash={.5 .25}{.25}](0',n-2'){}
\drawedge[linecolor=red,dash={.5 .25}{.25}](p',r'){}
\drawedge(n-2',n-1'){}
\drawloop[loopangle=270](n-1'){}
\drawedge[curvedepth=-.2,linecolor=red,dash={.5 .25}{.25}](q1',p'){}
\drawedge[curvedepth=-.3,syo=.5,linecolor=red,dash={.5 .25}{.25}](qdots',p'){}
\drawedge[curvedepth=-.8,linecolor=red,dash={.5 .25}{.25}](qv',p'){}
\drawedge[curvedepth=1](z',r'){}
\drawedge[curvedepth=1](r',rdots'){}
\drawedge[curvedepth=1](rdots',z'){}
\end{picture}\end{center}
\caption{Case~3.1.}\label{fig:case3.1}
\end{figure}

Let $z$ be the state from the cycle of $t$ such that $zt = r$. 
We observe the following properties:
\begin{enumerate}
\item[(a)] $\{p,z\}$ is a~colliding pair focused by $s$ to state $r$ in the cycle, which is the smallest state in a~cycle. 
This is the only colliding pair which is focused to a~state in a~cycle.

\item[(b)] All states from $Q_M$ whose mapping is different in $s$ and $t$ belong to the tree of $t$, and so to the orbit of a~cycle.

\item[(c)] $s$ has a~cycle.
\end{enumerate}

\textit{Internal injectivity}:
Let $\e{t}$ be any transformation that fits in this case and results in the same $s$; we will show that $\e{t}=t$.
From~(a), there is the unique colliding pair $\{p,z\}$ focused to a~state in a~cycle, hence $\{\e{p},\e{z}\} = \{p,z\}$.
Moreover, $p$ and $\e{p}$ are not in the cycle, whereas $z$ and $\e{z}$ are, so $\e{p}=p$ and $\e{z}=z$.
Since there is no state $q \neq 0$ such that $qt=p$, the only states mapped to $p$ by $s$ are $q_i$, hence $q_i = \e{q_i}$ for all $i$.
We know that $0t = 0\e{t} = p$, and $q_i t = q_i \e{t} = n-2$ for all $i$.
Since the other transitions in $s$ are defined exactly as in $t$ and $\e{t}$, we know that $\e{t}=t$.

\textbf{Case~3.2}: $t$ does not fit into any of the previous cases, $k=0$, and there exists a~state $x \in Q \setminus \{0\}$ such that $xt \notin \{x,n-2,n-1\}$.\\
Let $x$ be the smallest state among the states satisfying the conditions and with the largest $\ell\ge 1$ such that $xt^\ell \notin \{xt^{\ell-1},n-2,n-1\}$.
By the conditions of the case and since $t$ does not have a~cycle, $x$ is well-defined, and $\ell \ge 1$ and it is finite.

Note that $xt^{\ell+1} \neq n-2$, because $xt^\ell$ collides with $p$.
We have $xt^{\ell+1} \in \{xt^\ell,n-1\}$, and $x$ has in-degree 0.
Also note that, since $k=0$, all $q_i$ are of in-degree 0, because otherwise $pt=q_it=n-2$ would violate suffix-freeness.
We have the following subcases in this case that cover all possibilities for $t$:

\textbf{Subcase~3.2.1}: $\ell \ge 2$ and $xt^{\ell+1} = n-1$.\\
We have the following two subsubcases: (i) there exists $i$ such that $q_i < x$, and (ii) there is no such $i$.
Let $s$ be the transformation illustrated in Fig.~\ref{fig:subcase3.2.1} and defined by
\begin{center}
  $0 s = n-2$, $p s = xt^\ell$,\\
  $(xt^\ell) s = xt^\ell$ (i), $(xt^\ell) s = n-1$ (ii),\\
  $q_i s = p$ for $1\le i\le v$,\\
  $q s = q t$ for the other states $q\in Q$.
\end{center}
\begin{figure}[ht]
\unitlength 10pt\small
\gasset{Nh=2.5,Nw=2.5,Nmr=1.25,ELdist=0.3,loopdiam=1.5}
\begin{center}\begin{picture}(28,12)(0,-2)
\node[Nframe=n](name)(0,9){\normalsize$t\colon$}
\node(0)(2,7){0}\imark(0)
\node(p)(14,7){$p$}
\node(n-1)(26,0){$n$-$1$}
\node(n-2)(26,7){$n$-$2$}\rmark(n-2)
\node(x)(2,0){$x$}
\node(xt)(6,0){$xt$}
\node[Nframe=n](xdots)(10,0){$\dots$}
\node(xtl)(14,0){$xt^\ell$}
\node(q1)(18,4){$q_1$}
\node[Nframe=n](qdots)(20.5,4){$\dots$}
\node(qv)(23,4){$q_v$}
\drawedge(0,p){}
\drawedge(p,n-2){}
\drawedge(n-2,n-1){}
\drawloop[loopangle=270](n-1){}
\drawedge[curvedepth=.5](q1,n-2){}
\drawedge[curvedepth=.6,sxo=-.5,exo=1.5](qdots,n-2){}
\drawedge[curvedepth=0](qv,n-2){}
\drawedge(x,xt){}
\drawedge(xt,xdots){}
\drawedge(xdots,xtl){}
\drawedge(xtl,n-1){}
\end{picture}
\begin{picture}(28,12)(0,-1)
\node[Nframe=n](name)(0,9){\normalsize$s\colon$}
\node(0')(2,7){0}\imark(0')
\node(p')(14,7){$p$}
\node(n-1')(26,0){$n$-$1$}
\node(n-2')(26,7){$n$-$2$}\rmark(n-2')
\node(q1')(18,4){$q_1$}
\node[Nframe=n](qdots')(20.5,4){$\dots$}
\node(qv')(23,4){$q_v$}
\node(x')(2,0){$x$}
\node(xt')(6,0){$xt$}
\node[Nframe=n](xdots')(10,0){$\dots$}
\node(xtl')(14,0){$xt^\ell$}
\drawedge[curvedepth=3,linecolor=red,dash={.5 .25}{.25}](0',n-2'){}
\drawedge(n-2',n-1'){}
\drawloop[loopangle=270](n-1'){}
\drawedge[curvedepth=-.2,linecolor=red,dash={.5 .25}{.25}](q1',p'){}
\drawedge[curvedepth=-.3,syo=.5,linecolor=red,dash={.5 .25}{.25}](qdots',p'){}
\drawedge[curvedepth=-.8,linecolor=red,dash={.5 .25}{.25}](qv',p'){}
\drawedge[curvedepth=-3,linecolor=red,dash={.5 .25}{.25}](p',xtl'){}
\drawedge(x',xt'){}
\drawedge(xt',xdots'){}
\drawedge(xdots',xtl'){}
\drawloop[ELpos=80,linecolor=red,dash={.1 .1}{.1}](xtl'){(i)}
\drawedge[linecolor=red,dash={.1 .1}{.1}](xtl',n-1'){(ii)}
\end{picture}\end{center}
\caption{Subcase~3.2.1.}\label{fig:subcase3.2.1}
\end{figure}

We observe the following properties:
\begin{enumerate}
\item[(a)] $\{p, xt^{\ell-1}\}$ is a~colliding pair focused by $s$ to $xt^\ell$.

\item[(b)] All states from $Q_M$ whose mapping is different in $s$ and $t$ belong to the tree of $xt^\ell$,
which is either a~fixed point (i) or a~state mapped to $n-1$ (ii).

\item[(c)] $s$ does not contain any cycles.
\end{enumerate}

\textit{External injectivity}:
Since $s$ does not have any cycles, it is different from the transformations of Case~3.1.

\textit{Internal injectivity}:
Let $\e{t}$ be any transformation that fits in this subcase and results in the same $s$; we will show that $\e{t}=t$.
By Lemma~\ref{lem:orbits} the trees from~(b) of $t$ and $\e{t}$ must be the same, so $xt^\ell = \e{x}\e{t}^{\e{\ell}}$.
Also, the subsubcase is determined by $xt^\ell s$ and thus it is the same for both $t$ and $\e{t}$.

Consider all colliding pairs focused by $s$ to $xt^\ell$ that do not contain $xt^\ell$.
All of them contain $p$, so if there are two or more such pairs, then $\e{p} = p$.
Suppose that there is only one such pair $\{p,x^{\ell-1}\} = \{\e{p},\e{x}\e{t}^{\e{\ell}-1}\}$.
Note that $\ell = \e{\ell}$, as this is the length of a~longest path ending at $xt^{\ell} = \e{x}\e{t}^{\e{\ell}}$.
Also, only states $q_i$ are mapped to state $p$, and they all have in-degree 0.
If $\ell > 2$, then $p$ is distinguished from $xt^{\ell-1}$, since to $xt^{\ell-1}$ there is mapped $xt^{\ell-2}$ of in-degree $>0$; hence $p = \e{p}$.
Consider $\ell = 2$.
Let $U$ be the set of states that are mapped either to $p$ or to $xt^{\ell-1}$; then $\e{U} = U$.
The smallest state in $U$ is either a~state $q_i$ or $x$ (by the choice of $x$).
If the subsubcase is (i), then the smallest state in $U$ is $q_i$ and so is mapped to $p$, while in subsubcase (ii) it is $x$ mapped to $xt^{\ell}$.
Hence, the smallest state distinguishes $p$ from $xt^{\ell}$, and we have $p = \e{p}$ and $xt^{\ell-1} = \e{x}\e{t}^{\ell-1}$.
Then also $q_i = \e{q_i}$ for all $i$, since these are precisely the states mapped to $p = \e{p}$.
Summarizing, we know that $0t = 0\e{t} = p$, $pt = p\e{t} = n-2$, $(xt^\ell)t = (\e{x}\e{t}^{\ell})\e{t} = n-1$, and $q_i t = q_i \e{t} = n-2$.
Since the other transitions in $s$ are defined exactly as in $t$ and $\e{t}$, we have $\e{t}=t$.

\textbf{Subcase~3.2.2}: $\ell=1$, $xt^2 = n-1$, and $xt$ has in-degree at least $2$.\\
Let $y$ be the smallest state such that $yt = xt$ and $y \neq x$.
Note that $x < y$ and $y$ has in-degree 0.
Let $s$ be the transformation illustrated in Fig.~\ref{fig:subcase3.2.2} and defined by
\begin{center}
  $0 s = n-2$, $p s = y$,\\
  $(xt) s = x$, $x s = y$,\\
  $q_i s = p$ for all $i$,\\
  $q s = q t$ for the other states $q\in Q$.
\end{center}
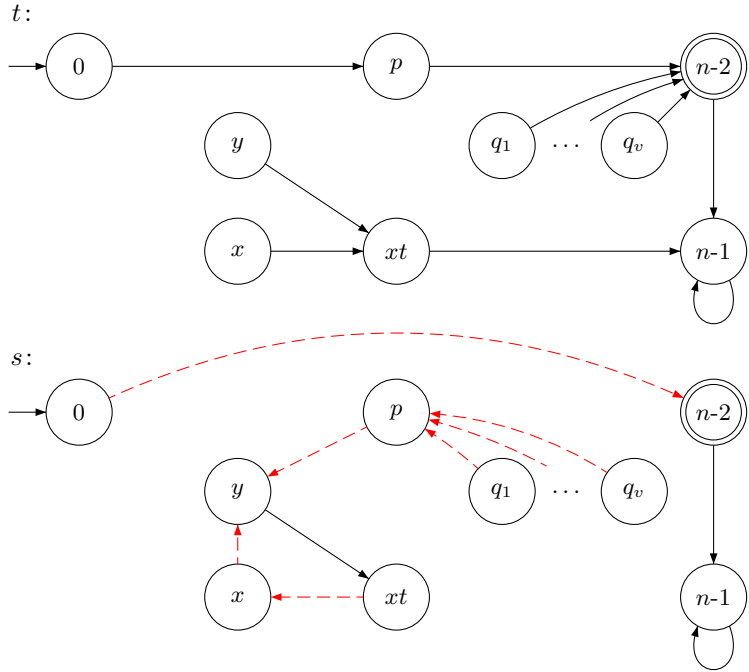
\begin{figure}[ht]
\unitlength 10pt\small
\gasset{Nh=2.5,Nw=2.5,Nmr=1.25,ELdist=0.3,loopdiam=1.5}
\begin{center}\begin{picture}(28,12)(0,-2)
\node[Nframe=n](name)(0,9){\normalsize$t\colon$}
\node(0)(2,7){0}\imark(0)
\node(p)(14,7){$p$}
\node(n-1)(26,0){$n$-$1$}
\node(n-2)(26,7){$n$-$2$}\rmark(n-2)
\node(x)(8,0){$x$}
\node(xt)(14,0){$xt$}
\node(y)(8,4){$y$}
\node(q1)(18,4){$q_1$}
\node[Nframe=n](qdots)(20.5,4){$\dots$}
\node(qv)(23,4){$q_v$}
\drawedge(0,p){}
\drawedge(p,n-2){}
\drawedge(n-2,n-1){}
\drawloop[loopangle=270](n-1){}
\drawedge[curvedepth=.5](q1,n-2){}
\drawedge[curvedepth=.6,sxo=-.5,exo=1.5](qdots,n-2){}
\drawedge[curvedepth=0](qv,n-2){}
\drawedge(x,xt){}
\drawedge(xt,n-1){}
\drawedge(y,xt){}
\end{picture}
\begin{picture}(28,12)(0,-1)
\node[Nframe=n](name)(0,9){\normalsize$s\colon$}
\node(0')(2,7){0}\imark(0')
\node(p')(14,7){$p$}
\node(n-1')(26,0){$n$-$1$}
\node(n-2')(26,7){$n$-$2$}\rmark(n-2')
\node(q1')(18,4){$q_1$}
\node[Nframe=n](qdots')(20.5,4){$\dots$}
\node(qv')(23,4){$q_v$}
\node(x')(8,0){$x$}
\node(xt')(14,0){$xt$}
\node(y')(8,4){$y$}
\drawedge[curvedepth=3,linecolor=red,dash={.5 .25}{.25}](0',n-2'){}
\drawedge(n-2',n-1'){}
\drawloop[loopangle=270](n-1'){}
\drawedge[curvedepth=-.2,linecolor=red,dash={.5 .25}{.25}](q1',p'){}
\drawedge[curvedepth=-.3,syo=.5,linecolor=red,dash={.5 .25}{.25}](qdots',p'){}
\drawedge[curvedepth=-.8,linecolor=red,dash={.5 .25}{.25}](qv',p'){}
\drawedge(y',xt'){}
\drawedge[linecolor=red,dash={.5 .25}{.25}](p',y'){}
\drawedge[linecolor=red,dash={.5 .25}{.25}](xt',x'){}
\drawedge[linecolor=red,dash={.5 .25}{.25}](x',y'){}
\end{picture}\end{center}
\caption{Subcase~3.2.2.}\label{fig:subcase3.2.2}
\end{figure}

We observe the following properties:
\begin{enumerate}
\item[(a)] $\{p,xt\}$ is a~colliding pair focused by $st$ to $xt$.
Note that in contrast to the previous cases, the focusing transformation here is $st$ instead of $s$.

\item[(b)] All states from $Q_M$ whose mapping is different in $s$ and $t$ belong to the same orbit of a~cycle.

\item[(c)] $s$ contains exactly one cycle, namely $(xt,x,y)$.
\end{enumerate}

\textit{External injectivity}:
Since all colliding pairs focused by $s$ must belong to the orbit from~(b), and the smallest state in the cycle of the orbit from~(b) is $x$ of in-degree 1, $s$ does not map a~colliding pair to it and thus it is different from the transformations of Case~3.1.

Since $s$ has a~cycle, it is different from the transformations of Subcase~3.3.1.

\textit{Internal injectivity}:
Let $\e{t}$ be any transformation that fits in this subcase and results in the same $s$; we will show that $\e{t}=t$.
All colliding pairs that are focused have states from the orbit of the cycle from Property~(b), hence $(xt,x,y) = (\e{x}\e{t},\e{x},\e{y})$.
Since $x$ and $\e{x}$ are the smallest states in the cycle, we have $x = \e{x}$, $y = \e{y}$, and $xt = \e{x}\e{t}$.
Since $y$ has in-degree 0 in $t$, $p$ is the only state outside the cycle that is mapped to $y$ in $s$; hence $p = \e{p}$.
Also, all states mapped to $p$ by $s$ are precisely the states $q_i$; hence $q_i = \e{q_i}$ for all $i$.
We know that $0t = 0\e{t} = p$, $pt = p\e{t} = n-2$, $xt = x\e{t}$, $(xt)t = (xt)\e{t} = n-1$, and $q_i = \e{q_i} = n-2$.
Since the other transitions in $s$ are defined exactly as in $t$ and $\e{t}$, we have $\e{t}=t$.

\textbf{Subcase~3.2.3}: $\ell=1$, $xt^2 = n-1$, and $xt$ has in-degree $1$.\\
We split the subcase into the following two subsubcases: (i) $v \ge 1$ or $p < xt$; (ii) $v=0$ and $p > xt$.
Let $s$ be the transformation illustrated in Fig.~\ref{fig:subcase3.2.3} and defined by
\begin{center}
  $0 s = n-2$, $p s = x$,\\
  $xt s = x$,\\
  $x s = x$ (i), $x s = n-1$ (ii),\\
  $q_i s = p$ for all $i$,\\
  $q s = q t$ for the other states $q\in Q$.
\end{center}
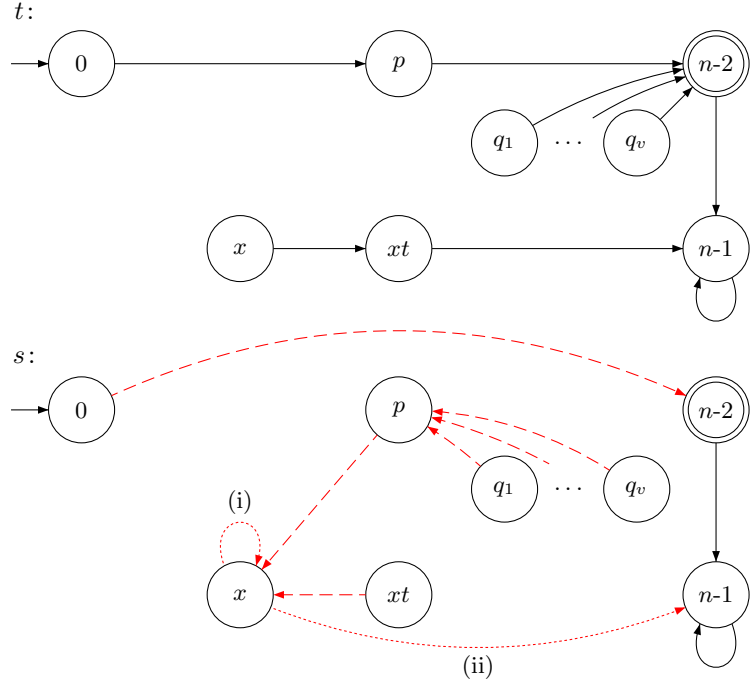
\begin{figure}[ht]
\unitlength 10pt\small
\gasset{Nh=2.5,Nw=2.5,Nmr=1.25,ELdist=0.3,loopdiam=1.5}
\begin{center}\begin{picture}(28,12)(0,-2)
\node[Nframe=n](name)(0,9){\normalsize$t\colon$}
\node(0)(2,7){0}\imark(0)
\node(p)(14,7){$p$}
\node(n-1)(26,0){$n$-$1$}
\node(n-2)(26,7){$n$-$2$}\rmark(n-2)
\node(x)(8,0){$x$}
\node(xt)(14,0){$xt$}
\node(q1)(18,4){$q_1$}
\node[Nframe=n](qdots)(20.5,4){$\dots$}
\node(qv)(23,4){$q_v$}
\drawedge(0,p){}
\drawedge(p,n-2){}
\drawedge(n-2,n-1){}
\drawloop[loopangle=270](n-1){}
\drawedge[curvedepth=.5](q1,n-2){}
\drawedge[curvedepth=.6,sxo=-.5,exo=1.5](qdots,n-2){}
\drawedge[curvedepth=0](qv,n-2){}
\drawedge(x,xt){}
\drawedge(xt,n-1){}
\end{picture}
\begin{picture}(28,13)(0,-2)
\node[Nframe=n](name)(0,9){\normalsize$s\colon$}
\node(0')(2,7){0}\imark(0')
\node(p')(14,7){$p$}
\node(n-1')(26,0){$n$-$1$}
\node(n-2')(26,7){$n$-$2$}\rmark(n-2')
\node(q1')(18,4){$q_1$}
\node[Nframe=n](qdots')(20.5,4){$\dots$}
\node(qv')(23,4){$q_v$}
\node(x')(8,0){$x$}
\node(xt')(14,0){$xt$}
\drawedge[curvedepth=3,linecolor=red,dash={.5 .25}{.25}](0',n-2'){}
\drawedge(n-2',n-1'){}
\drawloop[loopangle=270](n-1'){}
\drawedge[curvedepth=-.2,linecolor=red,dash={.5 .25}{.25}](q1',p'){}
\drawedge[curvedepth=-.3,syo=.5,linecolor=red,dash={.5 .25}{.25}](qdots',p'){}
\drawedge[curvedepth=-.8,linecolor=red,dash={.5 .25}{.25}](qv',p'){}
\drawedge[linecolor=red,dash={.5 .25}{.25}](p',x'){}
\drawedge[linecolor=red,dash={.5 .25}{.25}](xt',x'){}
\drawloop[linecolor=red,dash={.1 .1}{.1}](x'){(i)}
\drawedge[linecolor=red,dash={.1 .1}{.1},ELside=r,curvedepth=-2](x',n-1){(ii)}
\end{picture}\end{center}
\caption{Subcase~3.2.3.}\label{fig:subcase3.2.3}
\end{figure}

We observe the following properties:
\begin{enumerate}
\item[(a)] $\{p, xt\}$ is a~colliding pair focused by $s$ to $x$.

\item[(b)] All states from $Q_M$ whose mapping is different in $s$ and $t$ belong to the same tree of $x$,
which is either a~fixed point (i) or a~state mapped to $n-1$ (ii).

\item[(c)] $s$ does not contain any cycles.
\end{enumerate}

\textit{External injectivity}:
Since $s$ does not have any cycles, it is different from the transformations of Case~3.1 and Subcase~3.2.2.

Let $\e{t}$ be a~transformation that fits in Subcase~3.2.1 and results in the same $s$.
By Lemma~\ref{lem:orbits}, the trees from~(b) of both $t$ and $\e{t}$ must be the same, so $xt^\ell = \e{x}\e{t}^{\e{\ell}}$.
It follows that the subsubcases, which are determined by $xs$, are the same for both $t$ and $\e{t}$.
Note that $x$ has in-degree 2 in $s$, one of the states from this pair (i.e.\ $xt$) has in-degree $0$, and the other one ($\e{x}\e{t}^{\e{\ell}-1}$) has in-degree at least 1.
If the subsubcase is~(i), then $\e{p}$ has in-degree at least 1, and so both the states have in-degree at least 1, which yields a~contradiction.
If the subsubcase is~(ii), then $p$ has in-degree 0, and so both the states have in-degree 0, which yields a~contradiction.

\textit{Internal injectivity}:
Let $\e{t}$ be any transformation that fits in this subcase and results in the same $s$; we will show that $\e{t}=t$.
By Lemma~\ref{lem:orbits} we know that $x = \e{x}$, and so also $\{p,xt\} = \{\e{p},x\e{t}\}$.
The subsubcase for both $t$ and $\e{t}$ is determined by $xs$ and so must be the same.
If the subsubcase is~(i), then $p$ has in-degree $\ge 1$ or it is smaller than $xt$; hence $p = \e{p}$ and $xt = x\e{t}$.
If the subsubcase is~(ii), then both $p$ and $xt$ have in-degree $0$ and $p$ is larger than $xt$; hence again $p = \e{p}$ and $xt = x\e{t}$.
Also, $q_i = \e{q_i}$ as these are precisely all the states mapped to $p$ by $s$.
We know that $0t = 0\e{t} = p$, $pt = p\e{t} = n-2$, $(xt)t = (xt)\e{t} = n-1$, and $q_i t = q_i \e{t} = n-2$ for all $i$.
Since the other transitions in $s$ are defined exactly as in $t$ and $\e{t}$, we have $\e{t}=t$.

\textbf{Subcase~3.2.4}: $xt^{\ell} = xt^{\ell+1}$.\\
Let $s$ be the transformation illustrated in Fig.~\ref{fig:subcase3.2.4} and defined by
\begin{center}
  $0 s = n-2$, $p s = xt^\ell$,\\
  $(x t^i) s = x t^{i-1}$ for $1\le i\le \ell$,\\
  $x s = p$,\\
  $q_i s = x$ for $1\le i\le v$,\\
  $q s = q t$ for the other states $q\in Q$.
\end{center}
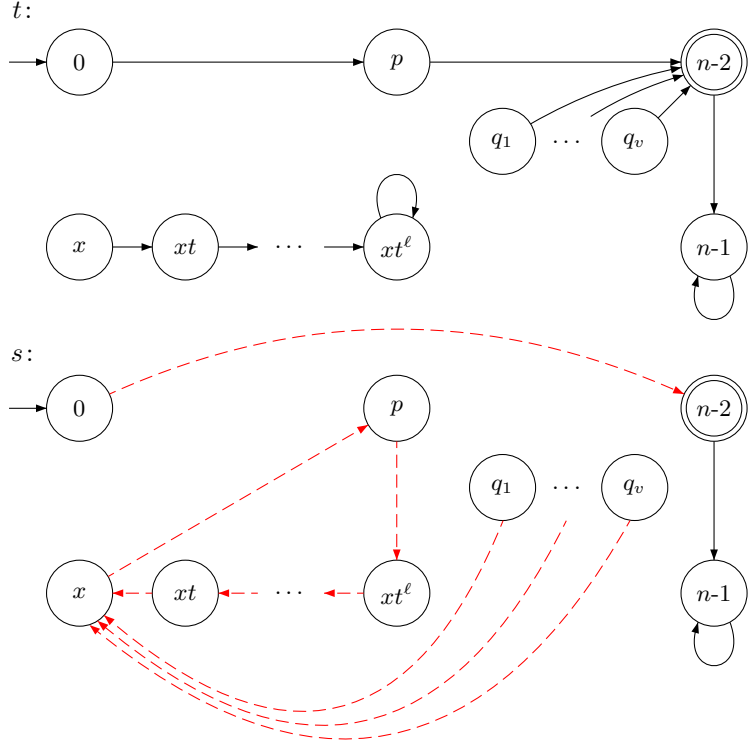
\begin{figure}[ht]
\unitlength 10pt\small
\gasset{Nh=2.5,Nw=2.5,Nmr=1.25,ELdist=0.3,loopdiam=1.5}
\begin{center}\begin{picture}(28,12)(0,-2)
\node[Nframe=n](name)(0,9){\normalsize$t\colon$}
\node(0)(2,7){0}\imark(0)
\node(p)(14,7){$p$}
\node(n-1)(26,0){$n$-$1$}
\node(n-2)(26,7){$n$-$2$}\rmark(n-2)
\node(x)(2,0){$x$}
\node(xt)(6,0){$xt$}
\node[Nframe=n](xdots)(10,0){$\dots$}
\node(xtl)(14,0){$xt^\ell$}
\node(q1)(18,4){$q_1$}
\node[Nframe=n](qdots)(20.5,4){$\dots$}
\node(qv)(23,4){$q_v$}
\drawedge(0,p){}
\drawedge(p,n-2){}
\drawedge(n-2,n-1){}
\drawloop[loopangle=270](n-1){}
\drawedge[curvedepth=.5](q1,n-2){}
\drawedge[curvedepth=.6,sxo=-.5,exo=1.5](qdots,n-2){}
\drawedge[curvedepth=0](qv,n-2){}
\drawedge(x,xt){}
\drawedge(xt,xdots){}
\drawedge(xdots,xtl){}
\drawloop(xtl){}
\end{picture}
\begin{picture}(28,16)(0,-5)
\node[Nframe=n](name)(0,9){\normalsize$s\colon$}
\node(0')(2,7){0}\imark(0')
\node(p')(14,7){$p$}
\node(n-1')(26,0){$n$-$1$}
\node(n-2')(26,7){$n$-$2$}\rmark(n-2')
\node(q1')(18,4){$q_1$}
\node[Nframe=n](qdots')(20.5,4){$\dots$}
\node(qv')(23,4){$q_v$}
\node(x')(2,0){$x$}
\node(xt')(6,0){$xt$}
\node[Nframe=n](xdots')(10,0){$\dots$}
\node(xtl')(14,0){$xt^\ell$}
\drawedge[curvedepth=3,linecolor=red,dash={.5 .25}{.25}](0',n-2'){}
\drawedge(n-2',n-1'){}
\drawloop[loopangle=270](n-1'){}
\drawedge[curvedepth=6.5,sxo=.5,linecolor=red,dash={.5 .25}{.25}](q1',x'){}
\drawedge[curvedepth=7,sxo=.5,exo=-.5,linecolor=red,dash={.5 .25}{.25}](qdots',x'){}
\drawedge[curvedepth=7.5,sxo=.5,exo=-1,linecolor=red,dash={.5 .25}{.25}](qv',x'){}
\drawedge[linecolor=red,dash={.5 .25}{.25}](p',xtl'){}
\drawedge[linecolor=red,dash={.5 .25}{.25}](xtl',xdots'){}
\drawedge[linecolor=red,dash={.5 .25}{.25}](xdots',xt'){}
\drawedge[linecolor=red,dash={.5 .25}{.25}](xt',x'){}
\drawedge[linecolor=red,dash={.5 .25}{.25}](x',p'){}
\end{picture}\end{center}
\caption{Subcase~3.2.4.}\label{fig:subcase3.2.4}
\end{figure}

We observe the following properties:
\begin{enumerate}
\item[(a)] $\{p,xt^\ell\}$ is a~colliding pair focused by $st$ to $xt^\ell$.

\item[(b)] All states from $Q_M$ whose mapping is different in $s$ and $t$ belong to the same orbit of a~cycle.

\item[(c)] $s$ contains exactly one cycle, namely $(p,xt^\ell,xt^{\ell-1},\ldots,x)$. 
\end{enumerate}

\textit{External injectivity}:
Let $\e{t}$ be a~transformation that fits in Case~3.1 and results in the same $s$.
Then $\e{t}$ must have the cycle $(p,xt^\ell,xt^{\ell-1},\ldots,x)$, since it exists in $s$ and the construction of Case~3.1 does not introduce any new cycles.
But then $0t\e{t}t = xt^{\ell}$ and $(xt^{\ell})t\e{t}t = xt^{\ell}$. Since $p$ collides with $xt^{\ell}$, $t$ and $\e{t}$ cannot be both in $T(\cD_n)$.

Since $s$ has a~cycle, it is different from the transformations of Subcase~3.2.1 and Subcase~3.2.3.

Now let $\e{t}$ be a~transformation that fits in Subcase~3.2.2 and results in the same $s$.
Since $s$ contains exactly one cycle, it must be that $\ell=1$ and $(p,xt,x) = (\e{x},\e{y},\e{x}\e{t})$.
We have the following three possibilities:
If $p = \e{x}$, $xt = \e{y}$, and $x = \e{x}\e{t}$, then $\e{t}$ focuses the colliding pair $\{p,xt\} = \{\e{x},\e{y}\}$; hence $t$ and $\e{t}$ cannot be both in $T(\cD_n)$.
If $p = \e{y}$, then we have a~contradiction with that $p$ has in-degree 1 and $\e{y}$ has in-degree 2.
Finally, suppose that $p = \e{x}\e{t}$, $xt = \e{x}$, and $x = \e{y}$.
Then $x = \e{y}$ must have in-degree $2$, and there is $q_1 = \e{p}$ (and $v = 1$).
But $\{\e{p},\e{x}\e{t}\} = \{q_1,p\}$ is a~colliding pair because of $\e{t}$, and it is focused to $n-2$ by $t$; hence $t$ and $\e{t}$ cannot be both in $T(\cD_n)$.

\textit{Internal injectivity}:
Let $\e{t}$ be any transformation that fits in this subcase and results in the same $s$; we will show that $\e{t}=t$.
By~(c), we know that $\ell = \e{\ell}$ and $(p,xt^\ell,xt^{\ell-1},\ldots,x) = (\e{p},\e{x}\e{t}^{\ell},\e{x}\e{t}^{\ell-1},\ldots,\e{x})$.

First suppose that $p = \e{p}$.
Then also $x = \e{x}$, $xt^\ell = x\e{t}^{\ell}$, $xt^{\ell-1} = x\e{t}^{\ell-1}$, and so on for the states of the cycle.
We know that $q_i = \e{q_i}$ for all $i$.
Hence, $0t = 0\e{t} = p$, $pt = p\e{t} = n-2$, $xt^i = x\e{t}^i$ for all $i$, and $q_i t = q_i\e{t} = n-2$.
Since the other transitions in $s$ are defined exactly as in $t$ and $\e{t}$, we have $\e{t}=t$.

Now suppose that $p \neq \e{p}$.
So $p = \e{x}\e{t}^i$ for some $i$.
Note that $p$ collides with all states $xt,\ldots,xt^\ell$, and $\e{p}$ collides with all states $\e{x}\e{t},\ldots,\e{x}\e{t}^{\ell}$.
If $\ell \ge 2$, then there exists $\e{x}\e{t}^j$ with $j \ge 1$ that is different from $p$ and collides with $p$.
But then $\e{t}^\ell$ focuses both these states to $\e{x}\e{t}^\ell$.
Finally consider $\ell = 1$.
If $p = \e{x}$ then $\{x,xt\} = \{\e{p},\e{x}\e{t}\}$, which is a~colliding pair because of $\e{t}$ that is focused by $t$ to $xt$.
On the other hand, if $p = \e{x}\e{t}$, then $xt = \e{x}$, and so $\{p,xt\} = \{\e{x}\e{t},\e{x}\}$ is a~colliding pair because of $t$ that is focused by $\e{t}$ to $\e{x}\e{t}$.
Hence, $t$ and $\e{t}$ cannot be both in $T(\cD_n)$.

\textbf{Case~3.3}: $t$ does not fit into any of the previous cases, $k=0$, and there exist at least two fixed points of in-degree 1.\\
Let the two smallest fixed points of in-degree 1 be the states $f_1$ and $f_2$, that is,
$$f_1 = \min\{q\in Q \mid q t = q, \forall_{q'\in Q \setminus \{q\}}\ q' t \neq q\},$$
$$f_2 = \min\{q\in Q\setminus\{f_1\} \mid q t = q, \forall_{q'\in Q \setminus \{q\}}\ q' t \neq q\}.$$
Let $s$ be the transformation illustrated in Fig.~\ref{fig:case3.3} and defined by
\begin{center}
  $0 s = n-2$, $f_1 s = f_2$, $f_2 s = f_1$, $p s = f_2$,\\
  $q_i s = p$ for $1\le i\le v$,\\
  $q s = q t$ for the other states $q\in Q$.
\end{center}
\begin{figure}[ht]
\unitlength 10pt\small
\gasset{Nh=2.5,Nw=2.5,Nmr=1.25,ELdist=0.3,loopdiam=1.5}
\begin{center}\begin{picture}(28,12)(0,-2)
\node[Nframe=n](name)(0,9){\normalsize$t\colon$}
\node(0)(2,7){0}\imark(0)
\node(p)(14,7){$p$}
\node(n-1)(26,0){$n$-$1$}
\node(n-2)(26,7){$n$-$2$}\rmark(n-2)
\node(q1)(18,4){$q_1$}
\node[Nframe=n](qdots)(20.5,4){$\dots$}
\node(qv)(23,4){$q_v$}
\node(f1)(8,0){$f_1$}
\node(f2)(14,0){$f_2$}
\drawedge(0,p){}
\drawedge(p,n-2){}
\drawedge(n-2,n-1){}
\drawloop[loopangle=270](n-1){}
\drawedge[curvedepth=.5](q1,n-2){}
\drawedge[curvedepth=.6,sxo=-.5,exo=1.5](qdots,n-2){}
\drawedge[curvedepth=0](qv,n-2){}
\drawloop(f1){}
\drawloop(f2){}
\end{picture}
\begin{picture}(28,14)(0,-1)
\node[Nframe=n](name)(0,9){\normalsize$s\colon$}
\node(0')(2,7){0}\imark(0')
\node(p')(14,7){$p$}
\node(n-1')(26,0){$n$-$1$}
\node(n-2')(26,7){$n$-$2$}\rmark(n-2')
\node(q1')(18,4){$q_1$}
\node[Nframe=n](qdots')(20.5,4){$\dots$}
\node(qv')(23,4){$q_v$}
\node(f1')(8,0){$f_1$}
\node(f2')(14,0){$f_2$}
\drawedge[curvedepth=3,linecolor=red,dash={.5 .25}{.25}](0',n-2'){}
\drawedge(n-2',n-1'){}
\drawloop[loopangle=270](n-1'){}
\drawedge[curvedepth=-.2,linecolor=red,dash={.5 .25}{.25}](q1',p'){}
\drawedge[curvedepth=-.3,syo=.5,linecolor=red,dash={.5 .25}{.25}](qdots',p'){}
\drawedge[curvedepth=-.8,linecolor=red,dash={.5 .25}{.25}](qv',p'){}
\drawedge[curvedepth=1,linecolor=red,dash={.5 .25}{.25}](f1',f2'){}
\drawedge[curvedepth=1,linecolor=red,dash={.5 .25}{.25}](f2',f1'){}
\drawedge[curvedepth=0,linecolor=red,dash={.5 .25}{.25}](p',f2'){}
\end{picture}\end{center}
\caption{Case~3.3.}\label{fig:case3.3}
\end{figure}
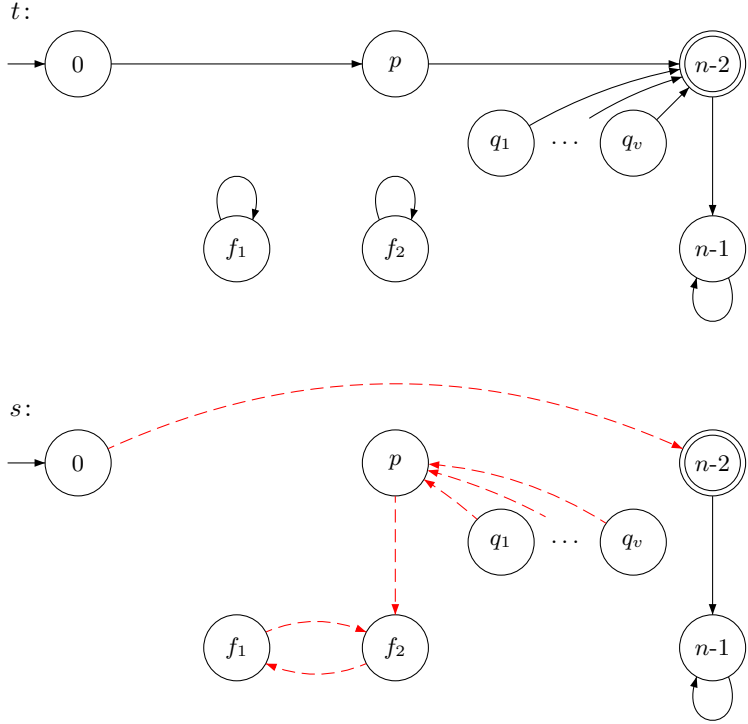

We observe the following properties:
\begin{enumerate}
\item[(a)] $\{p,f_1\}$ is a~colliding pair focused by $s$ to $f_2$.

\item[(b)] All states from $Q_M$ whose mapping is different in $s$ and $t$ belong to the same orbit of a~cycle.

\item[(c)] $s$ contains exactly one cycle, namely $(f_1,f_2)$.
\end{enumerate}

\textit{External injectivity}:
Since $\{p,f_1\}$ is the only colliding pair that is focused by $s$ to a~state in a~cycle, and $f_2$ is not the minimal state in the cycle, $s$ is different from the transformations of Case~3.1.

Since $s$ has a~cycle, it is different from the transformations of Subcase~3.2.1 and Subcase~3.2.3.
Also, since $s$ has exactly one cycle of length 2, it is different from the transformations of Subcase~3.2.2 and Subcase~3.2.4, which have a~cycle of length at least 3.

\textit{Internal injectivity}:
Let $\e{t}$ be any transformation that fits in this case and results in the same $s$; we will show that $\e{t}=t$.
From~(c), we know that $(f_1,f_2) = (\e{f_1},\e{f_2})$, and since $f_1$ has in-degree 1 and $f_2$ has in-degree 2 in $s$, we have $f_1 = \e{f_1}$ and $f_2 = \e{f_2}$.
Also $p = \e{p}$, as only $p$ and $f_2$ are mapped to $f_1$.
Then $q_i = \e{q_i}$ for all $i$, since these are precisely the states mapped to $p$ in $s$.
Hence $0t = 0\e{t}$, $pt = p\e{t} = n-1$, $f_1 t = f_1 \e{t} = f_1$, $f_2 t = f_2 \e{t} = f_2$, and $q_i t = q_i \e{t} = n-2$ for all $i$.
Since the other transitions in $s$ are defined exactly as in $t$ and $\e{t}$, we have $\e{t} = t$.

\textbf{Case~3.4}: $t$ does not fit into any of the previous cases and $k=0$.\\
In $t$, there is neither a~cycle (covered by Case~3.1) nor a~state $x \in Q_M$ such that $xt \notin \{x,n-1,n-2\}$ (covered by Case~3.2).
Hence, because $t \notin \Wbf(n)$, there must be a~fixed point $f$ of in-degree 1.
Because of Case~3.3, there is exactly one such fixed point.

Let $q_1 < \ldots < q_v$ be all the states from $Q_M \setminus \{p,f\}$ such that $q_i t = n-2$.
Let $r_1 < \ldots < r_u$ be all the states from $Q_M \setminus \{p,f\}$ such that $r_i t = n-1$.
All states $q_i$ and $r_i$ have in-degree 0 (covered by Case~3.2), and they are all the states besides $0,p,f,n-2,n-1$.
Because $n \ge 8$, we know that $v+u \ge 3$.
We have the following subcases that cover all possibilities for $t$:

\textbf{Subcase~3.4.1}: $v \ge 2$.\\
Let $s$ be the transformation illustrated in Fig.~\ref{fig:subcase3.4.1} and defined by
\begin{center}
  $0 s = n-2$, $p s = f$,\\
  $q_i s = q_{i+1}$ for $1\le i\le v-1$,\\
  $q_v s = q_1$,\\
  $r_i s = q_v$ for $1\le i\le u$,\\
  $q s = q t$ for the other states $q\in Q$.
\end{center}
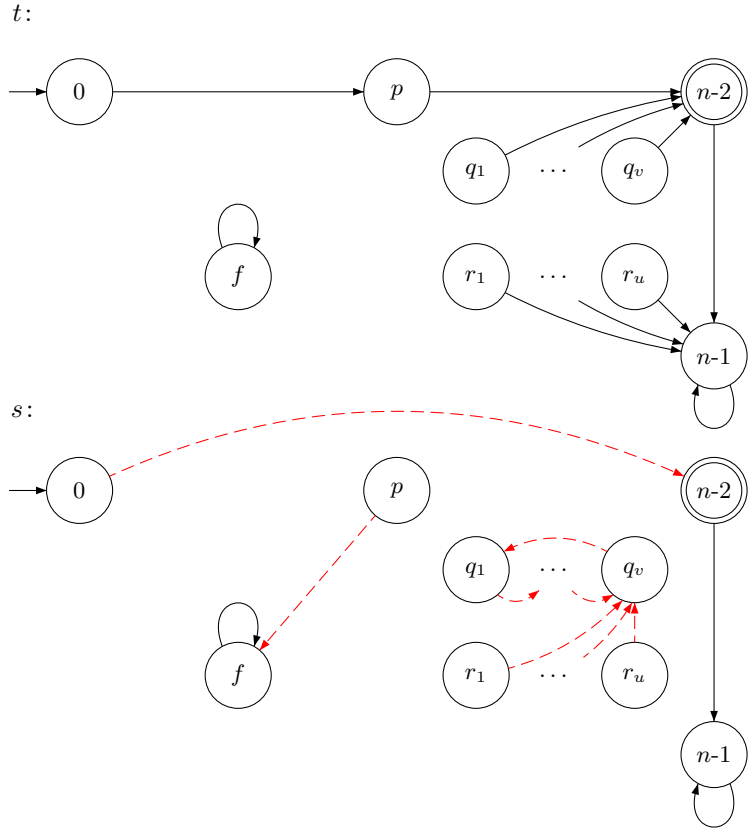
\begin{figure}[ht]
\unitlength 10pt\small
\gasset{Nh=2.5,Nw=2.5,Nmr=1.25,ELdist=0.3,loopdiam=1.5}
\begin{center}\begin{picture}(28,14)(0,-2)
\node[Nframe=n](name)(0,13){\normalsize$t\colon$}
\node(0)(2,10){0}\imark(0)
\node(p)(14,10){$p$}
\node(n-1)(26,0){$n$-$1$}
\node(n-2)(26,10){$n$-$2$}\rmark(n-2)
\node(q1)(17,7){$q_1$}
\node[Nframe=n](qdots)(20,7){$\dots$}
\node(qv)(23,7){$q_v$}
\node(r1)(17,3){$r_1$}
\node[Nframe=n](rdots)(20,3){$\dots$}
\node(ru)(23,3){$r_u$}
\node(f)(8,3){$f$}
\drawedge(0,p){}
\drawedge(p,n-2){}
\drawedge(n-2,n-1){}
\drawloop[loopangle=270](n-1){}
\drawedge[curvedepth=.5](q1,n-2){}
\drawedge[curvedepth=.6,sxo=-.5,exo=1.5](qdots,n-2){}
\drawedge[curvedepth=0](qv,n-2){}
\drawedge[curvedepth=-.5](r1,n-1){}
\drawedge[curvedepth=-.6,sxo=-.5,exo=1.5](rdots,n-1){}
\drawedge(ru,n-1){}
\drawloop(f){}
\end{picture}
\begin{picture}(28,15)(0,-2)
\node[Nframe=n](name)(0,13){\normalsize$s\colon$}
\node(0')(2,10){0}\imark(0')
\node(p')(14,10){$p$}
\node(n-1')(26,0){$n$-$1$}
\node(n-2')(26,10){$n$-$2$}\rmark(n-2')
\node(q1')(17,7){$q_1$}
\node[Nframe=n,Nh=2,Nw=2,Nmr=1](qdots')(20,7){$\dots$}
\node(qv')(23,7){$q_v$}
\node(r1')(17,3){$r_1$}
\node[Nframe=n](rdots')(20,3){$\dots$}
\node(ru')(23,3){$r_u$}
\node(f')(8,3){$f$}
\drawedge[curvedepth=3,linecolor=red,dash={.5 .25}{.25}](0',n-2'){}
\drawedge(n-2',n-1'){}
\drawloop[loopangle=270](n-1'){}
\drawedge[linecolor=red,dash={.5 .25}{.25}](p',f'){}
\drawloop(f){}
\drawedge[curvedepth=-1.2,linecolor=red,dash={.5 .25}{.25}](q1',qdots'){}
\drawedge[curvedepth=-1.2,linecolor=red,dash={.5 .25}{.25}](qdots',qv'){}
\drawedge[curvedepth=-1.2,linecolor=red,dash={.5 .25}{.25}](qv',q1'){}
\drawedge[curvedepth=-.8,exo=.5,linecolor=red,dash={.5 .25}{.25}](r1',qv'){}
\drawedge[curvedepth=-.5,exo=.5,linecolor=red,dash={.5 .25}{.25}](rdots',qv'){}
\drawedge[linecolor=red,dash={.5 .25}{.25}](ru',qv'){}
\end{picture}\end{center}
\caption{Subcase~3.4.1.}\label{fig:subcase3.4.1}
\end{figure}

We observe the following properties:
\begin{enumerate}
\item[(a)] $\{p,f\}$ is a~colliding pair focused by $s$ to $f$.
This is the only colliding pair that is focused by $s$ to a~fixed point.

\item[(c)] $s$ contains exactly one cycle, namely $(q_1,\ldots,q_v)$.
\end{enumerate}

\textit{External injectivity}:
Observe that all states in the unique cycle have in-degree 1 except possibly $q_v$.
Thus, no colliding pair of states is focused to the smallest state $q_1$ in the cycle.
This distinguishes $s$ from the transformations of Case~3.1.

Since $s$ has a~cycle, it is different from the transformations of Subcase~3.2.1 and Subcase~3.2.3.
Also, $s$ is different from the transformations of Subcase~3.2.2, Subcase~3.2.4, and Case~3.3, which do not focus a~colliding pair to a~fixed point, because the orbits from their Properties~(b) do not have a~fixed point.

\textit{Internal injectivity}:
Let $\e{t}$ be any transformation that fits in this subcase and results in the same $s$; we will show that $\e{t}=t$.
By~(c), we know that $\e{q_i} = q_i$ for all $i$.
Then all states mapped by $s$ to $q_1$ must be $r_i$, hence $\e{r_i} = r_i$ for all $i$.
By~(a) and since the fixed point is distinguished in the colliding pair, we obtain that $\e{p} = p$ and $\e{f} = f$.
We know that $0t = 0\e{t} = p$, $pt = p\e{t} = n-2$, $q_i t = q_i \e{t} = n-2$ and $r_i t = r_i \e{t} = n-1$ for all $i$.
Since the other transitions in $s$ are defined exactly as in $t$ and $\e{t}$, we have $\e{t} = t$.

\textbf{Subcase~3.4.2}: $v = 1$.\\
We have $u \ge 2$. Let $s$ be the transformation illustrated in Fig.~\ref{fig:subcase3.4.2} and defined by
\begin{center}
  $0 s = n-2$, $p s = f$,\\
  $q_1 s = f$,\\
  $r_i s = p$ for $1\le i\le u$,\\
  $q s = q t$ for other states $q\in Q$.
\end{center}
\begin{figure}[ht]
\unitlength 10pt\small
\gasset{Nh=2.5,Nw=2.5,Nmr=1.25,ELdist=0.3,loopdiam=1.5}
\begin{center}\begin{picture}(28,14)(0,-2)
\node[Nframe=n](name)(0,13){\normalsize$t\colon$}
\node(0)(2,10){0}\imark(0)
\node(p)(14,10){$p$}
\node(n-1)(26,0){$n$-$1$}
\node(n-2)(26,10){$n$-$2$}\rmark(n-2)
\node(q1)(17,7){$q_1$}
\node(r1)(17,3){$r_1$}
\node[Nframe=n](rdots)(20,3){$\dots$}
\node(ru)(23,3){$r_u$}
\node(f)(8,3){$f$}
\drawedge(0,p){}
\drawedge(p,n-2){}
\drawedge(n-2,n-1){}
\drawloop[loopangle=270](n-1){}
\drawedge[curvedepth=.3](q1,n-2){}
\drawedge[curvedepth=-.5](r1,n-1){}
\drawedge[curvedepth=-.6,sxo=-.5,exo=1.5](rdots,n-1){}
\drawedge[curvedepth=-0](ru,n-1){}
\drawloop(f){}
\end{picture}
\begin{picture}(28,15)(0,-2)
\node[Nframe=n](name)(0,13){\normalsize$s\colon$}
\node(0')(2,10){0}\imark(0')
\node(p')(14,10){$p$}
\node(n-1')(26,0){$n$-$1$}
\node(n-2')(26,10){$n$-$2$}\rmark(n-2')
\node(q1')(17,7){$q_1$}
\node(r1')(17,3){$r_1$}
\node[Nframe=n](rdots')(20,3){$\dots$}
\node(ru')(23,3){$r_u$}
\node(f')(8,3){$f$}
\drawedge[curvedepth=3,linecolor=red,dash={.5 .25}{.25}](0',n-2'){}
\drawedge(n-2',n-1'){}
\drawloop[loopangle=270](n-1'){}
\drawedge[linecolor=red,dash={.5 .25}{.25}](p',f'){}
\drawloop(f'){}
\drawedge[curvedepth=-4.5,linecolor=red,dash={.5 .25}{.25}](r1',p'){}
\drawedge[curvedepth=-4.5,linecolor=red,dash={.5 .25}{.25}](rdots',p'){}
\drawedge[curvedepth=-4.5,eyo=.5,linecolor=red,dash={.5 .25}{.25}](ru',p'){}
\drawedge[linecolor=red,dash={.5 .25}{.25}](q1',f'){}
\end{picture}\end{center}
\caption{Subcase~3.4.2.}\label{fig:subcase3.4.2}
\end{figure}
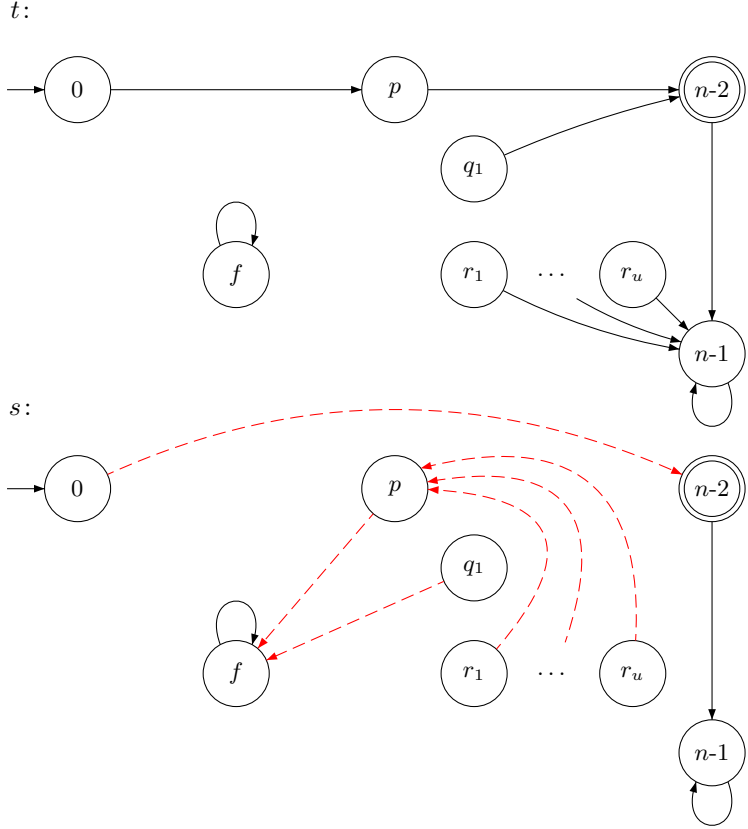

We observe the following properties:
\begin{enumerate}
\item[(a)] $\{p,f\}$ is a~colliding pair focused by $s$ to $f$.

\item[(b)] All states from $Q_M$ whose mapping is different in $s$ and $t$ belong to the same orbit of the fixed point $f$.

\item[(c)] $s$ does not contain any cycles.
\end{enumerate}

\textit{External injectivity}:
Since $s$ does not have any cycles, it is different from the transformations of Case~3.1, Subcase~3.2.2, Subcase~3.2.4, Case~3.3, and Subcase~3.4.1.

Let $\e{t}$ be a~transformation that fits in Subcase~3.2.1 and results in the same $s$.
By Lemma~\ref{lem:orbits}, the orbits from Properties~(b) for both $t$ and $\e{t}$ must be the same, so the subsubcase for $\e{t}$ is~(i), and necessarily $f = \e{x}\e{t}^{\e{\ell}}$.
We know that the states $\e{p}$ and $\e{x}\e{t}^{\e{\ell}-1}$ are mapped to $f$ and have in-degree at least 1.
This contradicts with that $p$ and $q_1$ are the only two states mapped to $f$, and $q_1$ has in-degree 0.

Let $\e{t}$ be a~transformation that fits in Subcase~3.2.3 and results in the same $s$.
By Lemma~\ref{lem:orbits}, the orbits from Properties~(b) for both $t$ and $\e{t}$ must be the same, so the subsubcase for $\e{t}$ is~(i), and necessarily $f = \e{x}$.
So $\{p,q_1\} = \{\e{p},\e{x}\e{t}\}$, but this is a~colliding pair because of $\e{t}$, which is focused to $n-2$ by $t$; hence, $t$ and $\e{t}$ cannot be both present in $T(\cD_n)$.

\textit{Internal injectivity}:
Let $\e{t}$ be any transformation that fits in this subcase and results in the same $s$; we will show that $\e{t}=t$.
Lemma~\ref{lem:orbits}, the orbits from Properties~(b) for both $t$ and $\e{t}$ must be the same, so we obtain that $f=\e{f}$.
So we have $\{p,q_1\} = \{\e{p},\e{q_1}\}$.
Since $q_1$ and $\e{q_1}$ have in-degree 0, and $p$ and $\e{p}$ have in-degree at least $2$, we have $q_1 = \e{q_1}$ and $p = \e{p}$.
Then $r_i = \e{r_i}$ for all $i$, as these are precisely the states mapped to $p$.
We know that $0t = 0\e{t} = p$, $pt = p\e{t} = n-2$, $q_1 t = q_1 \e{t} = n-2$, and $r_i t = r_i \e{t} = n-1$ for all $i$.
Since the other transitions in $s$ are defined exactly as in $t$ and $\e{t}$, we have $\e{t} = t$.

\textbf{Subcase~3.4.3}: $v = 0$.\\
Let $s$ be the transformation illustrated in Fig.~\ref{fig:subcase3.4.3} and defined by
\begin{center}
  $0 s = n-2$, $p s = f$,\\
  $r_1 s = p$,\\
  $r_i s = f$ for $2\le i\le u$,\\
  $q s = q t$ for other states $q\in Q$.
\end{center}
\begin{figure}[ht]
\unitlength 10pt\small
\gasset{Nh=2.5,Nw=2.5,Nmr=1.25,ELdist=0.3,loopdiam=1.5}
\begin{center}\begin{picture}(28,14)(0,-2)
\node[Nframe=n](name)(0,9){\normalsize$t\colon$}
\node(0)(2,7){0}\imark(0)
\node(p)(14,7){$p$}
\node(n-1)(26,0){$n$-$1$}
\node(n-2)(26,7){$n$-$2$}\rmark(n-2)
\node(f)(8,0){$f$}
\node(r1)(17,3){$r_1$}
\node[Nframe=n](rdots)(20,3){$\dots$}
\node(ru)(23,3){$r_u$}
\drawedge(0,p){}
\drawedge(p,n-2){}
\drawedge(n-2,n-1){}
\drawloop[loopangle=270](n-1){}
\drawedge[curvedepth=-.5](r1,n-1){}
\drawedge[curvedepth=-.6,sxo=-.5,exo=1.5](rdots,n-1){}
\drawedge[curvedepth=0](ru,n-1){}
\drawloop(f){}
\end{picture}
\begin{picture}(28,12)(0,-1)
\node[Nframe=n](name)(0,9){\normalsize$s\colon$}
\node(0')(2,7){0}\imark(0')
\node(p')(14,7){$p$}
\node(n-1')(26,0){$n$-$1$}
\node(n-2')(26,7){$n$-$2$}\rmark(n-2')
\node(f')(8,0){$f$}
\node(r1')(17,3){$r_1$}
\node[Nframe=n](rdots')(20,3){$\dots$}
\node(ru')(23,3){$r_u$}
\drawedge[curvedepth=3,linecolor=red,dash={.5 .25}{.25}](0',n-2'){}
\drawedge(n-2',n-1'){}
\drawloop[loopangle=270](n-1'){}
\drawedge[linecolor=red,dash={.5 .25}{.25}](p',f'){}
\drawloop(f'){}
\drawedge[linecolor=red,dash={.5 .25}{.25}](r1',p'){}
\drawedge[curvedepth=2.7,sxo=.5,eyo=.5,linecolor=red,dash={.5 .25}{.25}](rdots',f'){}
\drawedge[curvedepth=3,linecolor=red,dash={.5 .25}{.25}](ru',f'){}
\end{picture}\end{center}
\caption{Subcase~3.4.3.}\label{fig:subcase3.4.3}
\end{figure}
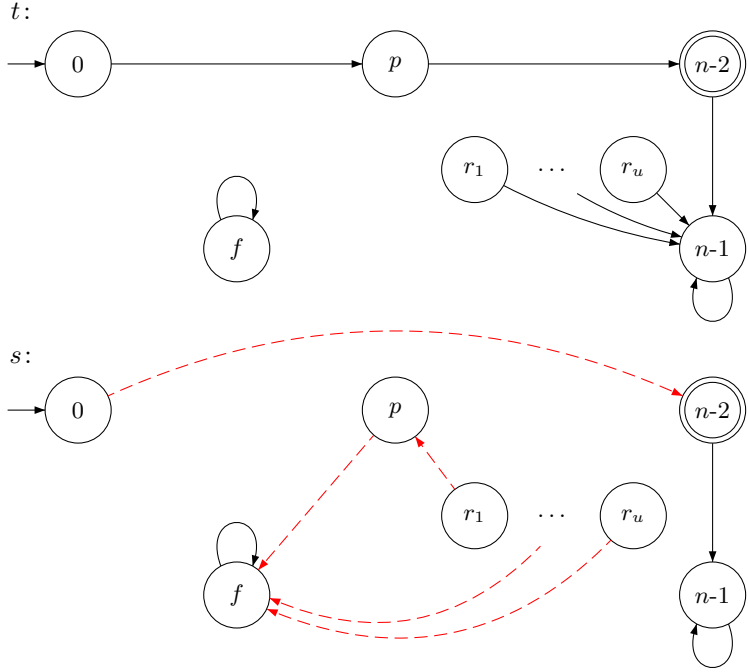

We observe the following properties:
\begin{enumerate}
\item[(a)] $\{p,f\}$ is a~colliding pair focused by $s$ to $f$.

\item[(b)] All states from $Q_M$ whose mapping is different in $s$ and $t$ belong to the same orbit of the fixed point $f$, which has in-degree $u+1 \ge 4$.

\item[(c)] $s$ does not contain any cycles.
\end{enumerate}

\textit{External injectivity}:
Since $s$ does not have any cycles, it is different from the transformations of Case~3.1, Subcase~3.2.2, Subcase~3.2.4, Case~3.3, and Subcase~3.4.1.

Let $\e{t}$ be a~transformation that fits in Subcase~3.2.1 and results in the same $s$.
By Lemma~\ref{lem:orbits}, the orbits from Properties~(b) for both $t$ and $\e{t}$ must be the same, so the subsubcase for $\e{t}$ must be~(i), and necessarily $f = \e{x}\e{t}^{\e{\ell}}$.
We know that the states $\e{p}$ and $\e{x}\e{t}^{\e{\ell}-1}$ are mapped by $s$ to $f$ and have in-degree at least $1$, because $\hat{\ell} \ge 2$ and in subsubcase~(i) there exists some $\e{q_i}$ mapped by $s$ to $\e{p}$.
On the other hand, all states mapped to $f$ (except $f$ itself) are $p$ and $r_2,\ldots,r_u$, where all the states $r_i$ have in-degree 0, which yields a~contradiction.

To distinguish $s$ from the transformations of Subcase~3.2.3 and of Subcase~3.4.2, observe that if they focus a~colliding pair to a~fixed point, then this fixed point have in-degree 3, but $s$ focuses a~colliding pair to the fixed point $f$ of in-degree at least 4.

\textit{Internal injectivity}:
Let $\e{t}$ be any transformation that fits in this subcase and results in the same $s$; we will show that $\e{t}=t$.
By Lemma~\ref{lem:orbits}, the orbits from Properties~(b) for both $t$ and $\e{t}$ must be the same, so we obtain that $f=\e{f}$.
We have $p = \e{p}$, as this is the unique state of in-degree 1 that is mapped to $f$.
Then $r_1 = \e{r_1}$ as this is the unique state mapped to $p$.
All states of in-degree 0 that mapped to $f$ are precisely $r_2,\ldots,r_u$; hence $r_i = \e{r_i}$ for all $i$.
We know that $0t = 0\e{t} = p$, $pt = p\e{t} = n-2$, and $r_i t = r_i \e{t} = n-1$ for all $i$.
Since the other transitions in $s$ are defined exactly as in $t$ and $\e{t}$, we have $\e{t} = t$.

\textbf{Case~3.5}: $k \ge 1$.\\
Let $q_1 < \ldots < q_v$ be all the states from $Q_M \setminus \{pt^k\}$ such that $q_i t = n-2$.
We split the case into the following three subcases covering all possibilities for $t$:

\textbf{Subcase~3.5.1}: $v = 0$ and $pt^k$ has in-degree $1$.\\
Let $s$ be the transformation illustrated in Fig.~\ref{fig:subcase3.5.1} and defined by
\begin{center}
  $0 s = n-2$, $p s = p$,\\
  $p t^i s = p t^{i-1}$ for $1\le i\le k$,\\
  $q s = qt$ for the other states $q\in Q$.
\end{center}
\begin{figure}[ht]
\unitlength 10pt\small
\gasset{Nh=2.5,Nw=2.5,Nmr=1.25,ELdist=0.3,loopdiam=1.5}
\begin{center}\begin{picture}(28,11)(0,-1)
\node[Nframe=n](name)(0,9){\normalsize$t\colon$}
\node(0)(2,7){0}\imark(0)
\node(p)(8,7){$p$}
\node[Nframe=n](pdots)(14,7){$\dots$}
\node(pt^k)(20,7){$pt^k$}
\node(n-2)(26,7){$n$-$2$}\rmark(n-2)
\node(n-1)(26,2){$n$-$1$}
\drawedge(0,p){}
\drawedge(p,pdots){}
\drawedge(pdots,pt^k){}
\drawedge(pt^k,n-2){}
\drawedge(n-2,n-1){}
\drawloop[loopangle=270](n-1){}
\end{picture}
\begin{picture}(28,10)(0,0)
\node[Nframe=n](name)(0,9){\normalsize$s\colon$}
\node(0')(2,7){0}\imark(0')
\node(p')(8,7){$p$}
\node[Nframe=n](pdots')(14,7){$\dots$}
\node(pt^k')(20,7){$pt^k$}
\node(n-2')(26,7){$n$-$2$}\rmark(n-2')
\node(n-1')(26,2){$n$-$1$}
\drawedge[curvedepth=3,linecolor=red,dash={.5 .25}{.25}](0',n-2'){}
\drawedge(n-2',n-1'){}
\drawloop[loopangle=270](n-1'){}
\drawloop[loopangle=270,linecolor=red,dash={.5 .25}{.25}](p'){}
\drawedge[linecolor=red,dash={.5 .25}{.25}](pdots',p'){}
\drawedge[linecolor=red,dash={.5 .25}{.25}](pt^k',pdots'){}
\end{picture}\end{center}
\caption{Subcase~3.5.1.}\label{fig:subcase3.5.1}
\end{figure}
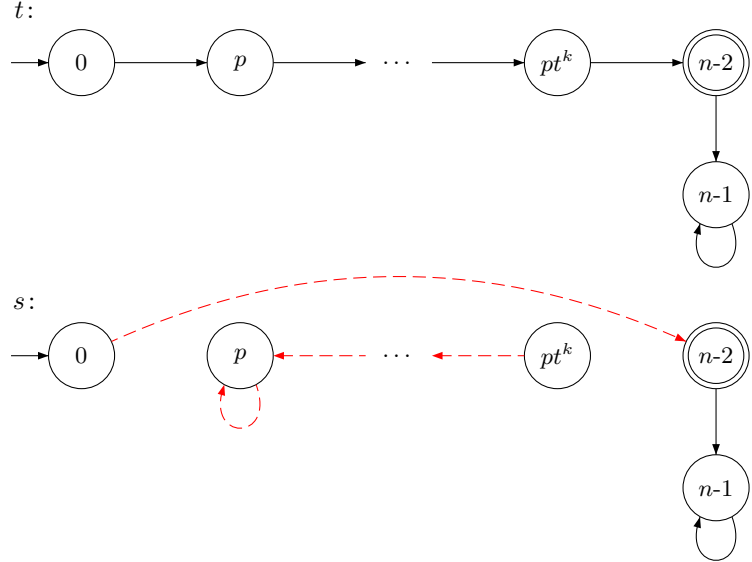

We observe the following properties:
\begin{enumerate}
\item[(a)] Pair $\{p,pt\}$ is a~colliding pair focused by $s$ to $p$.

\item[(b)] All states from $Q_M$ whose mapping is different in $s$ and $t$ belong to the orbit of the fixed point $p$, which has in-degree 2.
\end{enumerate}

\textit{External injectivity}:
Since the orbits from Properties~(b) for the transformations of Case~3.1, Subcase~3.2.2, Subcase~3.2.4, and Case~3.3 have cycles, and the orbit from~(b) of this subcase has a~fixed point, by Lemma~\ref{lem:orbits} $s$ is different from these transformations.
Similarly, the orbits from Properties~(b) for the transformations of Subcase~3.2.1, Subcase~3.2.3, Subcase~3.4.2, and Subcase~3.4.3 have a~fixed point of in-degree at least 3 or they are orbits of $n-1$, so by Lemma~\ref{lem:orbits} $s$ is different from these transformations.

Let $\e{t}$ be a~transformation that fits in Subcase~3.4.1 and results in the same $s$.
Since $\{\e{f},\e{p}\}$ is the only colliding pair that is focused to a~fixed point, it must be that $p = \e{f}$ and $pt = \e{p}$.
States $\e{q_i}$ form a~cycle in $s$, and since it is in a~different orbit from that from~(b), the cycle must be also present in $t$.
Hence, states $\e{q_i}$ collide with $pt = \e{p}$, and, in particular, $\{\e{q_1},\e{p}\}$ is a~colliding pair focused to $n-2$ by $\e{t}$, and so $t$ and $\e{t}$ cannot be both present in $T(\cD_n)$.

\textit{Internal injectivity}:
This follows exactly in the same way as in Case~2.2.

\textbf{Subcase~3.5.2}: $v = 0$ and $p t^k$ has in-degree at least $2$.\\
Let $y$ be the smallest state such that $yt = pt^k$ and $y \neq pt^{k-1}$.\\
Let $s$ be the transformation illustrated in Fig.~\ref{fig:subcase3.5.2} and defined by
\begin{center}
  $0 s = n-2$, $p s = y$,\\
  $y s = n-1$,\\
  $pt^i s = pt^{i-1}$ for $1\le i\le k$,\\
  $q s = q t$ for the other states $q\in Q$.
\end{center}
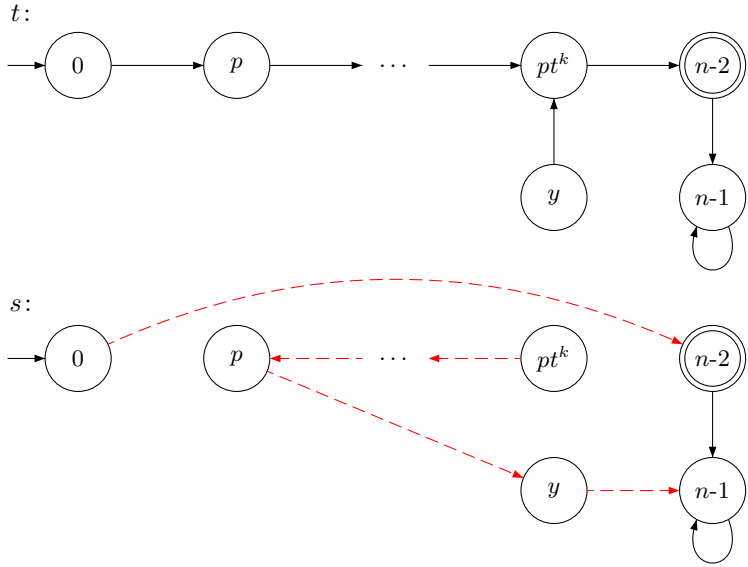
\begin{figure}[ht]
\unitlength 10pt\small
\gasset{Nh=2.5,Nw=2.5,Nmr=1.25,ELdist=0.3,loopdiam=1.5}
\begin{center}\begin{picture}(28,11)(0,-1)
\node[Nframe=n](name)(0,9){\normalsize$t\colon$}
\node(0)(2,7){0}\imark(0)
\node(p)(8,7){$p$}
\node[Nframe=n](pdots)(14,7){$\dots$}
\node(pt^k)(20,7){$pt^k$}
\node(y)(20,2){$y$}
\node(n-2)(26,7){$n$-$2$}\rmark(n-2)
\node(n-1)(26,2){$n$-$1$}
\drawedge(0,p){}
\drawedge(p,pdots){}
\drawedge(pdots,pt^k){}
\drawedge(pt^k,n-2){}
\drawedge(y,pt^k){}
\drawedge(n-2,n-1){}
\drawloop[loopangle=270](n-1){}
\end{picture}
\begin{picture}(28,11)(0,-1)
\node[Nframe=n](name)(0,9){\normalsize$s\colon$}
\node(0')(2,7){0}\imark(0')
\node(p')(8,7){$p$}
\node[Nframe=n](pdots')(14,7){$\dots$}
\node(pt^k')(20,7){$pt^k$}
\node(y')(20,2){$y$}
\node(n-2')(26,7){$n$-$2$}\rmark(n-2')
\node(n-1')(26,2){$n$-$1$}
\drawedge[curvedepth=3,linecolor=red,dash={.5 .25}{.25}](0',n-2'){}
\drawedge(n-2',n-1'){}
\drawloop[loopangle=270](n-1'){}
\drawedge[linecolor=red,dash={.5 .25}{.25}](y',n-1'){}
\drawedge[linecolor=red,dash={.5 .25}{.25}](p',y'){}
\drawedge[linecolor=red,dash={.5 .25}{.25}](pdots',p'){}
\drawedge[linecolor=red,dash={.5 .25}{.25}](pt^k',pdots'){}
\end{picture}\end{center}
\caption{Subcase~3.5.2.}\label{fig:subcase3.5.2}
\end{figure}

We observe the following properties:
\begin{enumerate}
\item[(a)] Pair $\{p,pt^k\}$ is a~colliding pair focused by $st$ to $pt^k$.

\item[(b)] All states from $Q_M$ whose mapping is different in $t$ and $s$ belong to the tree of $y$ in $s$, where $y$ is mapped to $n-1$.
\end{enumerate}

\textit{External injectivity}:
Since the orbits from Properties~(b) for the transformations of Case~3.1, Subcase~3.2.2, Subcase~3.2.4, and Case~3.3 have cycles, and the orbit from~(b) of this subcase is the orbit of $n-1$, by Lemma~\ref{lem:orbits} $s$ is different from these transformations.
Similarly, the orbits from Properties~(b) for the transformations of Subcase~3.2.1~(i), Subcase~3.2.3~(i), Subcase~3.4.2, Subcase~3.4.3, and Subcase~3.5.1 have a~fixed point from $Q_M$, so by Lemma~\ref{lem:orbits} $s$ is different from these transformations.
Since the transformations of Subcase~3.4.1 focus a~colliding pair to a~fixed point, they are also different from $s$.

Let $\e{t}$ be a~transformation from Subcase~3.2.1~(ii) that results in the same $s$.
By Lemma~\ref{lem:orbits}, the trees from Properties~(b) for both $t$ and $\e{t}$ must be the same, and so it must be that $y = \e{x}\e{t}^{\e{\ell}}$.
First observe that $p \neq \e{p}$, because otherwise $p$ and $pt=\e{q_i}$ for some $i$ would form a~colliding pair because of $t$, which is focused by $\e{t}$ to $n-2$.
So $p$ must be another state mapped by $s$ to $y = \e{x}\e{t}^{\e{\ell}}$, and so also by $\e{t}$.
It follows that all states $p,pt,\ldots,pt^k$ are mapped by $\e{t}$ in the same way as by $s$.
But then $p \e{t} t = p s t = pt^k$ and $(pt^k)\e{t} t = (pt^k)s t = pt^k$, so the colliding pair $\{p,pt^k\}$ is focused by $\e{t}t$, which yields a~contradiction.

Let $\e{t}$ be a~transformation from Subcase~3.2.3~(ii) that results in the same $s$.
By Lemma~\ref{lem:orbits}, the trees from Properties~(b) for both $t$ and $\e{t}$ must be the same, and so $\e{x} = y$.
But $\e{p}$ and $\e{x}\e{t}$ are the only states mapped to $y$ in $s$, and they both have in-degree 0, whereas $p$ is also mapped to $y$ in $s$ and has in-degree 1, which yields a~contradiction.

\textit{Internal injectivity}:
Let $\e{t}$ be any transformation that fits in this subcase and results in the same $s$; we will show that $\e{t}=t$.
By Lemma~\ref{lem:orbits}, the trees from Property~(b) must be the same, so $y = \e{y}$.
Since in $s$ all the states besides $p$ that are mapped to $y$ are also mapped to $y$ in $t$, it follows that $p \e{t} = y$ and $\e{p} t = y$.
Note that for $i$, $0 \le i \le \min\{k,\e{k}\}$, the distance in $s$ from $p t^i$ and from $\e{p} \e{t}^i$ to $y$ is $i+1$.
Hence, if $i \neq j$ then $p t^i \neq \e{p} \e{t}^j$.

\textbf{Subcase~3.5.3}: $v \ge 1$.\\
We define $c$ to be the largest distance in $t$ from a~state $q \in Q$ to some state $q_i$, that is,
$$c = \max\{d \in \mathbb{N} \mid \exists q\in Q\text{ such that }q t^d = q_i\text{ for some }i\}.$$
Notice that $c \le k$, because otherwise, when $qt^{k+1}=q_i$ for some $q$, the state $qt$ would be colliding with $p$ and the pair $\{qt,p\}$ would be focused by $t^{k+1}$ to $n-2$ (see also the observation at the beginning of Supercase~3).
Define
$$x = \min\{q\in Q \mid q t^c = q_i\text{ for some }i\},$$
that is, $x$ is the smallest state w.r.t. the ordering of the states among the furthest states from some $q_i$.
Let $q_m$ be that state $q_i$, which is the first state $q_i$ in the path from $x$.
Notice that if all $q_i$ have in-degree 0, then $c = 0$ and $x = q_m = q_1$.

Let $s$ be the transformation illustrated in Fig.~\ref{fig:subcase3.5.3} and defined by
\begin{center}
  $0 s = n-2$, $p s = x$,\\
  $p t^i s = p t^{i-1}$ for $1\le i \le k$,\\
  $q_i s = q_{i+1}$ for $1\le i\le v-1$,\\
  $q_v s = q_1$,\\
  $q s = q t$, for the other states $q\in Q$.
\end{center}
\begin{figure}[ht]
\unitlength 10pt\small
\gasset{Nh=2.5,Nw=2.5,Nmr=1.25,ELdist=0.3,loopdiam=1.5}
\begin{center}\begin{picture}(28,15)(0,-1)
\node[Nframe=n](name)(0,13){\normalsize$t\colon$}
\node(0)(2,11){0}\imark(0)
\node(p)(8,11){$p$}
\node[Nframe=n](pdots)(14,11){$\dots$}
\node(pt^k)(20,11){$pt^k$}
\node(q1)(14,6){$q_1$}
\node[Nframe=n](qdots1)(16.5,6){$\dots$}
\node(qi)(19,6){$q_m$}
\node[Nframe=n](qdots2)(21.5,6){$\dots$}
\node(qv)(24,6){$q_v$}
\node(x)(8,2){$x$}
\node[Nframe=n](xdots)(14,2){$\dots$}
\node(n-2)(26,11){$n$-$2$}\rmark(n-2)
\node(n-1)(26,2){$n$-$1$}
\drawedge(0,p){}
\drawedge(p,pdots){}
\drawedge(pdots,pt^k){}
\drawedge(pt^k,n-2){}
\drawedge(n-2,n-1){}
\drawloop[loopangle=270](n-1){}
\drawedge(x,xdots){}
\drawedge[curvedepth=-2,exo=.2](xdots,qi){}
\drawedge[curvedepth=.6](q1,n-2){}
\drawedge[curvedepth=.3,sxo=-1](qdots1,n-2){}
\drawedge[curvedepth=.2,sxo=-.5](qi,n-2){}
\drawedge[curvedepth=.1](qdots2,n-2){}
\drawedge[curvedepth=0](qv,n-2){}
\end{picture}
\begin{picture}(28,14)(0,0)
\node[Nframe=n](name)(0,14){\normalsize$s\colon$}
\node(0')(2,11){0}\imark(0')
\node(p')(8,11){$p$}
\node[Nframe=n](pdots')(14,11){$\dots$}
\node(pt^k')(20,11){$pt^k$}
\node(q1')(14,6){$q_1$}
\node[Nframe=n,Nh=2,Nw=2,Nmr=1](qdots1')(16.5,6){$\dots$}
\node(qi')(19,6){$q_m$}
\node[Nframe=n,Nh=2,Nw=2,Nmr=1](qdots2')(21.5,6){$\dots$}
\node(qv')(24,6){$q_v$}
\node[Nframe=n](xdots')(14,2){$\dots$}
\node(x')(8,2){$x$}
\node(n-2')(26,11){$n$-$2$}\rmark(n-2')
\node(n-1')(26,2){$n$-$1$}
\drawedge[curvedepth=3,linecolor=red,dash={.5 .25}{.25}](0',n-2'){}
\drawedge(n-2',n-1'){}
\drawloop[loopangle=270](n-1'){}
\drawedge[linecolor=red,dash={.5 .25}{.25}](pdots',p'){}
\drawedge[linecolor=red,dash={.5 .25}{.25}](pt^k',pdots'){}
\drawedge[linecolor=red,dash={.5 .25}{.25}](p',x'){}
\drawedge(x',xdots'){}
\drawedge[curvedepth=-2,exo=.2](xdots',qi'){}
\drawedge[curvedepth=-1.5,linecolor=red,dash={.5 .25}{.25}](q1',qdots1'){}
\drawedge[curvedepth=-1.5,linecolor=red,dash={.5 .25}{.25}](qdots1',qi'){}
\drawedge[curvedepth=-1.5,linecolor=red,dash={.5 .25}{.25}](qi',qdots2'){}
\drawedge[curvedepth=-1.5,linecolor=red,dash={.5 .25}{.25}](qdots2',qv'){}
\drawedge[curvedepth=-2,linecolor=red,dash={.5 .25}{.25}](qv',q1'){}
\end{picture}\end{center}
\caption{Subcase~3.5.3.}\label{fig:subcase3.5.3}
\end{figure}
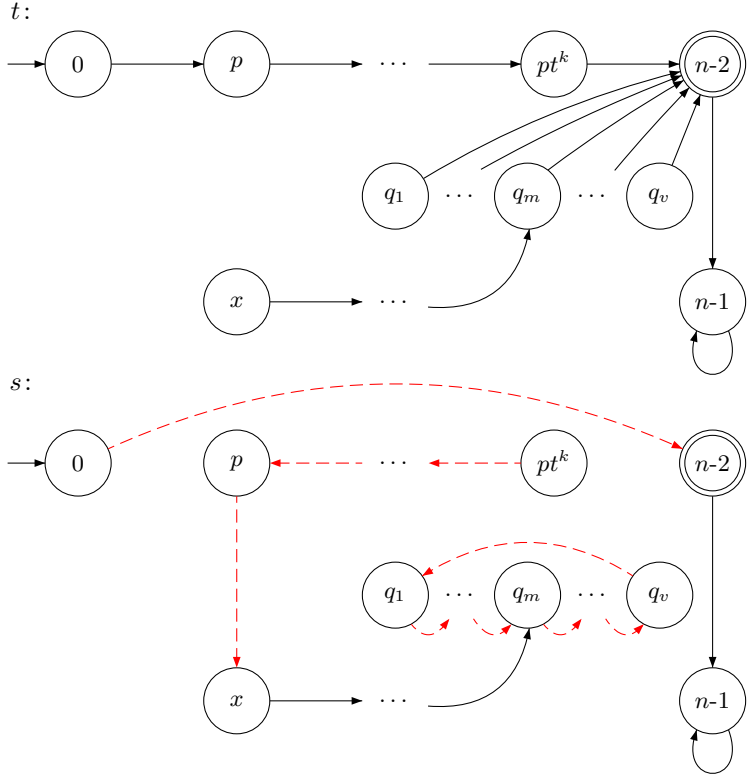

We observe the following properties:
\begin{enumerate}
\item[(a)] $\{p, pt\}$ is a~colliding pair focused by $s^{c+2}t$ to $n-2$.

\item[(b)] All states from $Q_M$ whose mapping is different in $s$ and $t$ belong to the same orbit of a~cycle (if $v \ge 2$) or a~fixed point (if $v=1$).

\item[(d)] Every longest path in $s$ from some state not in a~cycle to the first reachable $q_i$ contain both $p$ and $x$, and this $q_i$ is $q_m$.

\textit{Proof}: If such a~path contains $x$, then it does not contain $p,\ldots,pt^k$, and so would exist also in $t$.
But then, by the choice of $x$, its length could be at most $c$, whereas the path from $pt^k$ to $q_m$ is of length $k+c$.
Thus, every such a~path contain $x$ and so $p$, since $x$ has in-degree 1, and ends in $q_m$.
\end{enumerate}

\textit{External injectivity}:
Let $\e{t}$ be a~transformation that fits in Case~3.1 and results in the same $s$.
By Lemma~\ref{lem:orbits}, the orbits from Properties~(b) for both $t$ and $\e{t}$ are the same.
Let $y$ be the state mapped to $q_m$ in the path in $s$ from $p$ to $q_m$.
If $\e{p} \neq y$, then by the construction of $s$ in Case~3.1, all states in the tree of $y$ are mapped in $s$ in the same way as in $\e{t}$.
Hence, $\{p,pt\}$ is focused by $\e{t}^{c+2}t$ to $n-2$, which yields a~contradiction.
If $\e{p} = y$, then $p=\e{p}$, since to $\e{p}$ only the states $\e{q_i}$ are mapped, which have in-degree 0, and $p$ has in-degree 1.
Hence $k=1$, $p=y=\e{p}$, and $pt = \e{q_i}$ for some $i$. However, $\{p,pt\} = \{\e{p},\e{q_i}\}$ is a~colliding pair because of $t$ that is focused by $\e{t}$ to $n-2$, which yields a~contradiction.

Let $\e{t}$ be a~transformation that fits in Subcase~3.2.1 and results in the same $s$.
By Lemma~\ref{lem:orbits}, the orbits from Properties~(b) for both $t$ and $\e{t}$ are the same, so necessarily the subsubcase for $\e{t}$ must be~(i) and $v = 1$.
Since $\e{x}\e{t}^{\e{\ell}}$ has in-degree $\ge 3$ in $s$, it cannot be $x$, because $x$ has in-degree at most 2 (only $p$ is mapped to $x$ and, in the case $x=q_m=q_1$, additionally a~state $q_i$).
Thus $\e{p}$ and $\e{x}\e{t}^{\e{\ell}-1}$ are mapped in $s$ in the same way as in $t$.
But $\{\e{p},\e{x}\e{t}^{\e{\ell}-1}\}$ is a~colliding pair because of $\e{t}$, which is focused by $t$ to $\e{x}\e{t}^{\e{\ell}}$, which yields a~contradiction.

Let $\e{t}$ be a~transformation that fits in Subcase~3.2.2 and results in the same $s$.
By Lemma~\ref{lem:orbits}, the orbits from Properties~(b) for both $t$ and $\e{t}$ are the same, so necessarily $v=3$ and $\{\e{x},\e{y},\e{x}\e{t}\} = \{q_1,q_2,q_3\}$.
Observe that among the states mapped by $s$ to a~state in the cycle $(\e{x},\e{y},\e{x}\e{t})$, only $\e{p}$ can have in-degree larger than 0.
It follows that $\e{p} = p$, and we obtain a~contradiction exactly as for Case~3.1.

Let $\e{t}$ be a~transformation that fits in Subcase~3.2.3 and results in the same $s$.
By Lemma~\ref{lem:orbits}, the orbits from Properties~(b) for both $t$ and $\e{t}$ are the same, so necessarily the subsubcase for $\e{t}$ must be~(i) and $v=1$.
Since $\e{x}\e{t}$ has in-degree 1 in $\e{t}$, it has in-degree 0 in $s$, so it cannot be $p$.
Therefore $p = \e{p}$, but then we obtain a~contradiction exactly as for Case~3.1.

Let $\e{t}$ be a~transformation that fits in Subcase~3.2.4 and results in the same $s$.
By Lemma~\ref{lem:orbits}, the orbits from Properties~(b) for both $t$ and $\e{t}$ are the same, so necessarily $(\e{p},\e{x}\e{t}^{\e{\ell}},\ldots,\e{x})$ is the cycle formed by all states $q_i$.
But $\{\e{p},\e{x}\e{t}^{\e{\ell}}\}$ is a~colliding pair because of $\e{t}$, which is focused by $t$ to $n-2$; this yields a~contradiction.

Let $\e{t}$ be a~transformation that fits in Case~3.3 and results in the same $s$.
By Lemma~\ref{lem:orbits}, the orbits from Properties~(b) for both $t$ and $\e{t}$ are the same, so necessarily $v=2$ and $(f_1,f_2) = (q_1,q_2)$.
Then $p = \e{p}$, and again we obtain a~contradiction exactly as for Case~3.1.

Let $\e{t}$ be a~transformation that fits in Subcase~3.4.1 and results in the same $s$.
In $s$ there is exactly one orbit of a~fixed point from $Q_M$ and exactly one orbit of a~cycle.
But neither of them cannot be the orbit from~(b) of this subcase, since $\e{p}$ and states $\e{r_i}$ have in-degree 0 in $s$ so they cannot be $p$; this yields a~contradiction.

Let $\e{t}$ be a~transformation that fits in either Subcase~3.4.2 or Subcase~3.4.3 and results in the same $s$.
By Lemma~\ref{lem:orbits}, the orbits from Properties~(b) for both $t$ and $\e{t}$ are the same, so necessarily $v=1$ and $\e{f} = q_1 = q_m$.
Then $p = \e{p}$, as $\e{p}$ is the only state with non-zero in-degree in $s$ that is mapped to $\e{f}$.
So also $x = \e{f}$.
But there is another state mapped by $s$ to $p$ ($\e{q_1}$ or $\e{r_2}$, depending on the subcase), and it is mapped to $x$ also by $t$.
However, this contradicts that $x$ has in-degree 0 in $t$.

Let $\e{t}$ be a~transformation that fits in Subcase~3.5.1 and results in the same $s$.
By Lemma~\ref{lem:orbits}, the orbits from Properties~(b) for both $t$ and $\e{t}$ are the same, so necessarily $v=1$ and $\e{p} = q_1 = q_m$.
Consider the following path $s$, which contains all the states from $Q_M$ that are mapped differently in $t$ and $s$:
$$p t^k \stackrel{s}{\rightarrow} p t^{k-1} \stackrel{s}{\rightarrow} \dots \stackrel{s}{\rightarrow} p \stackrel{s}{\rightarrow} x \stackrel{s}{\rightarrow} \dots \stackrel{s}{\rightarrow} q_1.$$
Consider the second path in $s$, which contains all the states from $Q_M$ that are mapped differently in $\e{t}$ and $s$:
$$q_1 \e{t}^{\e{k}} \stackrel{s}{\rightarrow} q_1 \e{t}^{\e{k}-1} \stackrel{s}{\rightarrow} \dots \stackrel{s}{\rightarrow} q_1.$$
Let $y$ be the first common state in these paths; $y$ exists since both paths end up in $q_1$.
Note that $\e{t}$ reverses the second path.
We consider all possibilities for $y$, depending on where it occurs in the first chain:
\begin{itemize}
\item $y = p t^k$.
Then $y = q_1 \e{t}^j$ for some $j \ge 1$, so $\{y,\e{p}\} = \{pt^k,q_1\}$ is a~colliding pair because of $\e{t}$, which is focused by $t$ to $n-2$.
\item $y = p t^h$ for $1 \le h \le k-1$.
Then $(p t^h) s = p t^{h-1}$ so $(p t^{h-1})\e{t} = p t^h$, since $pt^{h-1}$ is in the second path and $\e{t}$ reverses it.
Also, $(p t^{h+1})\e{t} = (p t^{h+1})s = p t^h$, since $p t^{h+1}$ does not belong to the second path.
But then $\{p t^{h-1},p t^{h+1}\}$ is a~colliding pair because of $t$, which is focused by $\e{t}$ to $pt^h$.
\item $y = p$.
Since in $s$ only state $pt$ is mapped to $p$ and $p\e{t} \neq pt$, it must be that $p\e{t} = n-2$, as otherwise $(p\e{t})s = p$.
Therefore $p = q_1 \e{t}^{\e{k}}$.
But $q_1 \e{t}^{\e{k}}$ has in-degree 1 in $\e{t}$ from the conditions of Subcase~3.5.1, so it has in-degree 0 in $s$, which yields a~contradiction with in-degree 1 of $p$ in $s$.
\item $y$ is a~state in the path in $s$ from $x$ to $q_1$.
Then $q_1 \e{t}^j = y$ for some $j \le c$.
Remind that $c \le k$, so $j \le k$.
Since $y \notin \{p,pt,\ldots,pt^k\}$, the distance in $s$ from $pt^k$ to $y$ is at least $k+1 \ge j+1$.
It follows that there is a~state $z$ from the first chain such that $z s^{j+1} = z \e{t}^{j+1} = y$.
However, we also know that $0\e{t}^{j+1} = q_1 \e{t}^j = y$, hence $\e{t}$ cannot be in $T(\cD_n)$.
\end{itemize}
We obtained a~contradiction in every case, so $t$ and $\e{t}$ cannot be both in $T(\cD_n)$.

Let $\e{t}$ be a~transformation that fits in Subcase~3.5.2 and results in the same $s$.
However, by Lemma~\ref{lem:orbits}, the orbits from Properties~(b) for both $t$ and $\e{t}$ must be the same, but for $\e{t}$ this is an~orbit of $n-1$.

\textit{Internal injectivity}:
Let $\e{t}$ be any transformation that fits in this subcase and results in the same $s$; we will show that $\e{t}=t$.
By Lemma~\ref{lem:orbits}, the orbits from~(b) must be the same for both $t$ and $\e{t}$, hence $v = \e{v}$ and the sets of $q_i$ states are the same.

By~(d), both $p$ and $\e{p}$ are in every longest path to the first reachable $q_i$, so $q_m = \e{q_m}$.
Without loss of generality, state $\e{p}$ occurs not later than $p$, that is, we have $\e{p} s^j = p$ for some $j \ge 0$.
Since the path from $\e{x}$ to $q_m$ is the same in both $s$ and $\e{t}$, we have $\e{x} \e{t}^i = \e{x} s^i$ for all $i \ge 0$.

Consider the following path $P$ in $s$:
$$P = pt^k \stackrel{s}{\rightarrow} \dots \stackrel{s}{\rightarrow} p \stackrel{s}{\rightarrow} x \stackrel{s}{\rightarrow} \dots \stackrel{s}{\rightarrow} q_m.$$

First suppose that $P$ does not contain $\e{p}$.
Then also no state $\e{p} \e{t}^i$ for $1 \le i \le \e{k}$ would be in this path:
let $\e{p} \e{t}^i$ be such the state with the smallest $i$; then $(\e{p} \e{t}^i) s = \e{p} \e{t}^{i-1}$ would also be in this path, which is a~contradiction.
Hence, by the construction of $s$, this path is also present in $\e{t}$.
By the choice of $\e{x}$, the distance in $s$ from $\e{x}$ to $q_m$ is not smaller than the length of this path.
So we have $\e{c} \ge k+1+c$, which yields $\e{k} > k$ (because $\e{k} \ge \e{c}$).
Now observe that since in $s$ state $p$ is reachable from $\e{p}$, we have the following path in $s$:
$$\e{p} \e{t}^{\e{k}} \stackrel{s}{\rightarrow} \dots \stackrel{s}{\rightarrow} \e{p} \stackrel{s}{\rightarrow} \dots \stackrel{s}{\rightarrow} pt^i,$$
where $i$ is the smallest possible.
Then, by the construction of $s$, we have the following path in $t$:
$$\e{p} \e{t}^{\e{k}} \stackrel{t}{\rightarrow} \dots \stackrel{t}{\rightarrow} \e{p} \stackrel{t}{\rightarrow} \dots \stackrel{t}{\rightarrow} pt^i \stackrel{t}{\rightarrow} \dots \stackrel{t}{\rightarrow} pt^k.$$
This path has length at least $\e{k}+1 > k+1$.
Hence, there exists a~state $y \neq p$ in this path such that $y t^{k+1} = p t^k$.
This means that $\{p,yt\}$ is a~colliding pair because of $t$, which is focused by $t^k$ to $p t^k$.

There remains the case where $P$ contain $\e{p}$.
Since $\e{p}$ must occur before $p$ in $P$, we have $p t^h = \e{p}$ for some $h \ge 0$.

We claim that $p t^{h+i} = \e{p} \e{t}^{i}$ for all $i \ge 0$, which also implies $k = h+\e{k}$.
We use induction on $i$:
This holds for $i = 0$, and also for $i=1$, because $\e{k} \ge 1$ and the in-degree of $\e{p}$ is 1 in $s$.
For $i \ge 2$ assume that $p t^{h+j} = \e{p} \e{t}^{j}$ for all $j = 0,\ldots,i-1$.
Suppose for a~contradiction that $p t^{h+i} \neq \e{p} \e{t}^i$.
If $\e{p} \e{t}^i \neq n-2$, then $(\e{p} \e{t}^i)t = (\e{p} \e{t}^i)s = \e{p} \e{t}^{i-1} = p t^{h+i-1}$,
because in $s$ among the states mapped to $p t^{h+i-1}$, only $p t^{h+i}$ is mapped differently than in $t$.
Then, however, $\{\e{p} \e{t}^i,\e{p} \e{t}^{i-2}\}$ is a~colliding pair because of $\e{t}$ that is focused by $t$ to $\e{p} \e{t}^{i-1}$.
If $\e{p} \e{t}^i = n-2$, then, dually, $(p t^{h+i}) \e{t} = (p t^{h+i}) s = p t^{h+i-1} = \e{p} \e{t}^{i-1}$,
because in $s$, among the states mapped to $\e{p} \e{t}^{i-1}$, only $\e{p} \e{t}^i$ is mapped differently than in $\e{t}$.
Then, however, $\{p t^{h+i},p t^{h+i-2}\}$ is a~colliding pair because of $t$ that is focused by $\e{t}$ to $p t^{h+i-1}$.
Hence, the claim follows.

Suppose that $h \ge 1$.
Since the path in $s$ from $p$ to $q_m$ occurs also in $\e{t}$ and $t$, and is of length $c+1$, we have $p \e{t}^{c+2} = n-2$.
Note that $\e{k} \ge \e{c} = c+h \ge c+1$.
So there exists a~state $\e{p} \e{t}^{\e{k}-c-1} = p t^{h+\e{k}-c-1} \neq p$.
But this state collides with $p$ because of $t$, and the pair $\{p,\e{p} \e{t}^{\e{k}-c-1}\}$ is focused by $\e{t}^{c+2}$ to $n-2$.

Finally, if $h = 0$, then $0 t^i = 0 \e{t}^i$ for all $i \ge 0$, and $q_i t = q_i \e{t} = n-2$ for all $i$.
Since the other transitions in $s$ are defined exactly as in $t$ and $\e{t}$, we have $t = \e{t}$.
\end{proof}

\section{Uniqueness of maximal semigroups}

Here we show that $\Wbf(n)$ for $n \ge 6$ and $\Vbf(n)$ for $n \in \{3,4,5\}$ (whereas $\Wbf(n)=\Vbf(n)$ for $n \in \{3,4\}$) have not only the maximal sizes, but are also the unique largest semigroups up to renaming the states in a~minimal DFA $\cD_n = (Q,\Sigma,\delta,0,\{n-2\})$ of a~bifix-free language.

\begin{theorem}\label{thm:uniqueness}
If $n \ge 8$, and the transition semigroup $T(\cD_n)$ of a~minimal DFA $\cD_n$ of a~bifix-free language has at least one colliding pair, then
$$|T(\cD_n)| < |\Wbf(n)| = (n-1)^{n-3} + (n-2)^{n-3} + (n-3)2^{n-3}.$$
\end{theorem}
\begin{proof}
Assume that there is a~colliding pair $\{p_1,p_2\}$ with $p_1,p_2 \in Q_M$.
Since $n \ge 8$, there must be at least three states $r_1,r_2,r_3 \in Q_M \setminus \{p,q\}$.
Let $s \in \Wbf(n)$ be the transformation illustrated in Fig.~\ref{fig:uniqueness} and defined by:
\begin{center}
$0s = n-1$, $p_1 s = p_2$, $r_1 s = p_2$, $r_2 s = r_3$, $r_3 s = r_2$,\\
$q s = q$, for the other states $q \in Q$.
\end{center}
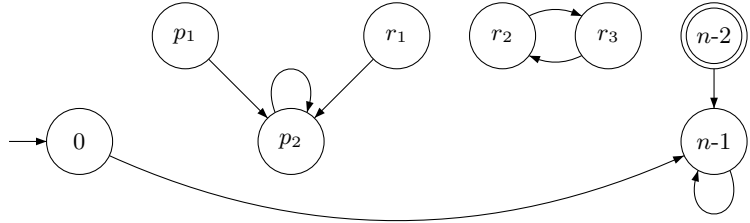
\begin{figure}[ht]
\unitlength 10pt\small
\gasset{Nh=2.5,Nw=2.5,Nmr=1.25,ELdist=0.3,loopdiam=1.5}
\begin{center}\begin{picture}(28,8)(0,-2)
\node(0)(2,0){0}\imark(0)
\node(p1)(6,4){$p_1$}
\node(p2)(10,0){$p_2$}
\node(r1)(14,4){$r_1$}
\node(r2)(18,4){$r_2$}
\node(r3)(22,4){$r_3$}
\node(n-1)(26,0){$n$-$1$}
\node(n-2)(26,4){$n$-$2$}\rmark(n-2)
\drawedge[curvedepth=-3](0,n-1){}
\drawedge(n-2,n-1){}
\drawloop[loopangle=270](n-1){}
\drawedge(p1,p2){}
\drawloop(p2){}
\drawedge(r1,p2){}
\drawedge[curvedepth=1](r2,r3){}
\drawedge[curvedepth=1](r3,r2){}
\end{picture}\end{center}
\caption{The transformation $s$ in the proof of Theorem~\ref{thm:uniqueness}.}\label{fig:uniqueness}
\end{figure}

Let $\varphi$ be the injective function from the proof of Theorem~\ref{thm:bifix-free_upper_bound}.
We will show that $s$ does not belong to $\varphi(T(\cD_n))$.

Since $\{p_1,p_2\}$ is focused by $s$ to $p_2$, $s$ is different from the transformations of Supercase~1.
Since $0s = n-1$, it is also different from the transformations of Supercase~3.

To see that it is different from all transformations of Supercase~2, notice that only the transformations of Case~2.1, Case~2.3, Subcase~2.4.2, Subcase~2.5.1, and Subcase 2.5.2 have a~cycle.
The transformations of Case~2.1, Case~2.3, and Subcase~2.4.2 have a~cycle with a~state with in-degree at least 2, whereas the single cycle $(r_2,r_3)$ in $s$ has both states of in-degree 1.
In the transformations of Subcase~2.5.1 and Subcase~2.5.2, there is only one fixed point from $Q_M$, and it has in-degree 2, whereas the single fixed point $p_2$ in $s$ has in-degree 3.

Thus, since $\varphi$ is injective and $\varphi(T(\cD_n)) \subseteq \Wbf(n)$, $s \in \Wbf(n)$ but $s \notin \varphi(T(\cD_n))$, it follows that $\varphi(T(\cD_n)) \subsetneq \Wbf(n)$, so $|T(\cD_n)| < |\Wbf(n)|$.
\end{proof}

\begin{corollary}
For $n \ge 8$, the transition semigroup $\Wbf(n)$ is the unique largest transition semigroup of a~minimal DFA of a~bifix-free language.
\end{corollary}
\begin{proof}
From Theorem~\ref{thm:uniqueness}, a~transition semigroup that has a~colliding pair cannot be largest.
From Proposition~\ref{pro:Wbf_unique}, $\Wbf(n)$ is the unique maximal transition semigroup that does not have colliding pairs of states.
\end{proof}

The following theorem solves the remaining cases of small semigroups:
\begin{theorem}
For $n \in \{6,7\}$, the largest transition semigroup of minimal DFAs of bifix-free languages is $\Wbf(n)$ and it is unique.
For $n = 5$, the largest transition semigroup of minimal DFAs of bifix-free languages is $\Vbf(n)$ and it is unique.
For $n \in \{3,4\}$, $\Wbf(n)=\Vbf(n)$ is the unique largest transition semigroup of minimal DFAs of bifix-free languages.
\end{theorem}
\begin{proof}
We have verified this with the help of computation, basing on the idea of conflicting pairs of transformations from~\cite[Theorem~20]{BLY12}.
The idea of the algorithm is described as follows, and the program is available at~\cite{SzWi18SyntacticComplexityOfBifixFreeArxiv}.

We say that two transformations $t_1,t_2 \in \Bbf(n)$ \emph{conflicts} if they cannot be both present in the transition semigroup of a~minimal DFA $\cD$ of a~bifix-free language, or they imply that all pairs of states from $Q_M$ are either colliding or focused.
In the latter case, by Proposition~\ref{pro:Vbf_unique} and Proposition~\ref{pro:Wbf_unique} we know that a~transition semigroup containing these transformations must be a~subsemigroup of $\Vbf(n)$ or $\Wbf(n)$, respectively.
Hence, we know that two conflicting transformations cannot be present in a~transition semigroup of size at least $\max\{\Vbf(n),\Wbf(n)\}$ which is different from $\Vbf(n)$ and $\Wbf(n)$.
Given a~set of transformations $B$, the \emph{graph of conflicts} is the graph $(B,E)$, where there is an~edge $(t_1,t_2) \in E$ if and only if $t_1$ conflicts with $t_2$.

Given an~$n$, our algorithm is as follows:
We keep a~subset $B_i \subseteq \Bbf(n)$ of transformations that can potentially be present in a~largest transition semigroup.
Starting with $B_0=\Bbf(n)$, we iteratively compute $B_{i+1} \subset B_i$, where $B_{i+1}$ is obtained from $B_i$ by removing some transformations.
This is done for $i=0,1,\ldots$ until we obtain $|B_{i+1}| = 0$.
If $B_{i+1} = B_i$ then the algorithm fails.

Given $B_i$, we compute $B_{i+1}$ by checking every transformation $t \in B$ and estimating how many pairwise non-conflicting transformations can we add to the set $\{t\}$.
Let $B' \subseteq B \setminus \{t\}$ be the set of all transformations that do not conflict with $t$.
The maximal number of pairwise non-conflicting transformations in $B'$ is the size of a~largest independent set in $B'$.
We only compute an~upper bound for it, since the problem is computationally hard.
Let $M$ be a~maximal matching in the graph of conflicts of $B'$; this can be computed by a~simple greedy algorithm in $O(|B'|^2)$ time.
Then $|B'|-|M|$ is an~upper bound for the size of a~largest independent set in $B'$, and so $1+|B'|-|M|$ is an~upper bound for the cardinality of a~maximal transition semigroup containing $t$ that is different from $\Vbf(n)$ and $\Wbf(n)$.
If this bound is smaller than $\max\{\Vbf(n),\Wbf(n)\}$, then we do not take $t$ into $B_{i+1}$; otherwise we keep $t$.

When $|B_i| = 0$, all transformations are rejected, which means that there are no transformations that can be present in a~transition semigroup of size at least $\max\{\Vbf(n),\Wbf(n)\}$ which is different from $\Vbf(n)$ and $\Wbf(n)$, so there are no such semigroups.

For $n=7$, two iterations were sufficient, and we obtained $|B_0| = 3653$, $|B_1|=1176$, and $|B_2|=0$; the computation took less than one minute.
\end{proof}

Since the largest transition semigroups are unique, from Propositions~\ref{pro:Wbf_alphabet_lower_bound} and~\ref{pro:Vbf_alphabet_lower_bound} we infer the sizes of the alphabets required in order to meet the bound for the syntactic complexity.
\begin{corollary}
To meet the bound for the syntactic complexity of bifix-free languages, $(n-2)^{n-3} + (n-3)2^{n-3} - 1$ letters are required and sufficient for $n \ge 6$, and $(n-2)!$ letters are required and sufficient for $n \in \{3,4,5\}$.
\end{corollary}


\section{Conclusions}

We have solved the problem of syntactic complexity of bifix-free languages and identified the largest semigroups for every number of states $n$.
In the main theorem, we used the method of injective function (cf.~\cite{BrSz14a,BrSz15SyntacticComplexityOfSuffixFree}) with new techniques and tricks for ensuring injectivity (in particular, Lemma~\ref{lem:orbits} and the constructions in Supercase~3). This stands as a~universal method for solving similar problems concerning the maximality of semigroups.
Our proof required an~extensive analysis of 23 (sub)cases and much more complicated injectivity arguments than those for suffix-free (12 cases), left ideals (5 subcases) and two-sided ideals (8 subcases).
The difficulty of applying the method grows quickly when the characterization of the class of languages gets more involved.

It may be surprising that we need a~witness with $(n-2)^{n-3} + (n-3)2^{n-3} - 1$ (for $n \ge 6$) letters to meet the bound for syntactic complexity of bifix-free languages, whereas in the case of prefix- and suffix-free languages only $n+1$ and five letters suffice, respectively (see \cite{BLY12,BrSz15SyntacticComplexityOfSuffixFree}).

Finally, our results enabled establishing the existence of most complex bifix-free languages (\cite{FeSz17ComplexityOfBifixFree,FeSz18ComplexityOfBifixFree}).

\section*{Acknowledgments}

We are grateful to anonymous reviewers for their careful reading of the proof.
This work was supported in part by the National Science Centre, Poland under project numbers 2014/15/B/ST6/00615 and 2017/25/B/ST6/01920 (Marek Szyku{\l}a).

\bibliographystyle{plainurl}
\providecommand{\noopsort}[1]{}

\newpage
\section*{Appendix: Map of the (sub)cases in the proof of Theorem~\ref{thm:bifix-free_upper_bound}}

{\setstretch{1.05}
\begin{enumerate}[leftmargin=*,widest=\textbf{Supercase~1}]
\item[\textbf{Supercase~1}:]
$t \in \Wbf(n)$.
\end{enumerate}

Let $p=0t$, and $k \ge 0$ be the largest integer such that $pt^k \notin \{n-2,n-1\}$.
\begin{enumerate}[leftmargin=*,widest=\textbf{Supercase~1}]
\item[\textbf{Supercase~2}:]
$t \notin \Wbf(n)$ and $pt^{k+1} = n-1$.
\begin{enumerate}[leftmargin=*,widest=\textbf{Case~2.1}]
\item[\textbf{Case~2.1}:] $t$ has a~cycle.
\item[\textbf{Case~2.2}:] $t$ has no cycles and $k \ge 1$.
\item[\textbf{Case~2.3}:] $t$ does not fit in any of the previous cases, and there exist at least two fixed points of in-degree 1.
\item[\textbf{Case~2.4}:] $t$ does not fit in any of the previous cases, and there exists $x \in Q\setminus \{0\}$ of in-degree $0$ such that $xt \notin \{x,n-2,n-1\}$.
Let $x$ be the smallest state such a~state with the largest $\ell\ge 1$ such that $xt^\ell \notin \{xt^{\ell-1},n-2,n-1\}$.
\begin{enumerate}[leftmargin=*,widest=\textbf{Subcase~2.4.1}]
\item[\textbf{Subcase~2.4.1}:] $\ell\ge 2$ and $xt^{\ell+1} = n-1$.
\item[\textbf{Subcase~2.4.2}:] $\ell=1$, $xt^2 = n-1$, and $xt$ has in-degree $>1$.
\item[\textbf{Subcase~2.4.3}:] $\ell=1$, $xt^2 = n-1$, and $xt$ has in-degree $1$.
\item[\textbf{Subcase 2.4.4}:] $xt^{\ell+1} = n-2$.
\item[\textbf{Subcase 2.4.5}:] $xt^{\ell+1} = xt^\ell$.
\end{enumerate}
\item[\textbf{Case 2.5}:] $t$ does not fit in any of the previous cases.
\begin{enumerate}[leftmargin=*,widest=\textbf{Subcase~2.5.1}]
\item[\textbf{Subcase~2.5.1}:] There are at least two states $r_1,r_2,\ldots,r_u$ from $Q \setminus \{0,p\}$ such that $r_i t = n-1$ for all $i$.
\item[\textbf{Subcase~2.5.2}:] $t$ does not fit in Subcase~2.5.1.
\end{enumerate}
\end{enumerate}
\medskip
\item[\textbf{Supercase 3}:] $t \notin \Wbf(n)$ and $pt^{k+1} = n-2$.
\begin{enumerate}[leftmargin=*,widest=\textbf{Case~3.1}]
\item[\textbf{Case~3.1}:] $k=0$ and $t$ has a~cycle.
\item[\textbf{Case~3.2}:] $t$ does not fit into any of the previous cases, $k=0$, and there exists a~state $x \in Q \setminus \{0\}$ such that $xt \notin \{x,n-1,n-2\}$.
Let $x$ be the smallest state such a~state with the largest $\ell\ge 1$ such that $xt^\ell \notin \{xt^{\ell-1},n-2,n-1\}$.
\begin{enumerate}[leftmargin=*,widest=\textbf{Subcase~3.2.1}]
\item[\textbf{Subcase~3.2.1}:] $\ell \ge 2$ and $xt^{\ell+1} = n-1$.
\item[\textbf{Subcase~3.2.2}:] $\ell=1$, $xt^2 = n-1$, and $xt$ has in-degree at least $2$.
\item[\textbf{Subcase~3.2.3}:] $\ell=1$, $xt^2 = n-1$, and $xt$ has in-degree $1$.
\item[\textbf{Subcase~3.2.4}:] $xt^{\ell} = xt^{\ell+1}$.
\end{enumerate}
\item[\textbf{Case~3.3}:] $t$ does not fit into any of the previous cases, $k=0$, and there exist at least two fixed points of in-degree 1.
\item[\textbf{Case~3.4}:] $t$ does not fit into any of the previous cases and $k=0$.
\begin{enumerate}[leftmargin=*,widest=\textbf{Subcase~3.4.1}]
\item[\textbf{Subcase~3.4.1}:] $v \ge 2$.
\item[\textbf{Subcase~3.4.2}:] $v = 1$.
\item[\textbf{Subcase~3.4.3}:] $v = 0$.
\end{enumerate}
\item[\textbf{Case~3.5}:] $k \ge 1$.
\begin{enumerate}[leftmargin=*,widest=\textbf{Subcase~3.5.1}]
\item[\textbf{Subcase~3.5.1}:] $v = 0$ and $pt^k$ has in-degree $1$.
\item[\textbf{Subcase~3.5.2}:] $v = 0$ and $p t^k$ has in-degree at least $2$.
\item[\textbf{Subcase~3.5.3}:] $v \ge 1$.
\end{enumerate}
\end{enumerate}
\end{enumerate}}


\begin{figure}[htb]
\unitlength 7.2pt\scriptsize
\gasset{Nh=2.5,Nw=2.5,Nmr=1.25,ELdist=0.5,loopdiam=1.5}
\begin{center}\begin{picture}(28,10)(0,-1)
\node[Nframe=n](name)(2,7){Case 2.1:}
\node(0)(2,0){0}\imark(0)
\node(p)(8,0){$p$}
\node[Nframe=n](pdots)(14,0){$\dots$}
\node(pt^k)(20,0){$pt^k$}
\node(n-1)(26,0){$n$-$1$}
\node(n-2)(26,4){$n$-$2$}\rmark(n-2)
\node(z)(12,4){$z$}
\node(r)(14,7){$r$}
\node[Nframe=n](rdots)(16,4){$\dots$}
\drawedge(0,p){}
\drawedge(p,pdots){}
\drawedge(pdots,pt^k){}
\drawedge(pt^k,n-1){}
\drawedge(n-2,n-1){}
\drawloop[loopangle=270](n-1){}
\drawedge[curvedepth=1](z,r){}
\drawedge[curvedepth=1](r,rdots){}
\drawedge[curvedepth=1](rdots,z){}
\drawedge[linecolor=red,dash={.5 .25}{.25},curvedepth=-2.5](0,n-1){}
\drawedge[linecolor=red,dash={.5 .25}{.25},curvedepth=3.5](p,r){}
\drawedge[linecolor=red,dash={.5 .25}{.25},curvedepth=-1.5](pdots,p){}
\drawedge[linecolor=red,dash={.5 .25}{.25},curvedepth=-1.5](pt^k,pdots){}
\end{picture}\begin{picture}(28,10)(-2,-1)
\node[Nframe=n](name)(2,7){Case 2.2:}
\node(0)(2,0){0}\imark(0)
\node(p)(8,0){$p$}
\node[Nframe=n](pdots)(14,0){$\dots$}
\node(pt^k)(20,0){$pt^k$}
\node(n-1)(26,0){$n$-$1$}
\node(n-2)(26,4){$n$-$2$}\rmark(n-2)
\drawedge(0,p){}
\drawedge(p,pdots){}
\drawedge(pdots,pt^k){}
\drawedge(pt^k,n-1){}
\drawedge(n-2,n-1){}
\drawloop[loopangle=270](n-1){}
\drawedge[linecolor=red,dash={.5 .25}{.25},curvedepth=-2.5](0,n-1){}
\drawedge[linecolor=red,dash={.5 .25}{.25},curvedepth=-1.5](pdots,p){}
\drawedge[linecolor=red,dash={.5 .25}{.25},curvedepth=-1.5](pt^k,pdots){}
\drawloop[linecolor=red,dash={.5 .25}{.25}](p){}
\end{picture}\end{center}

\begin{center}\begin{picture}(28,12)(0,-1)
\node[Nframe=n](name)(2,7){Case 2.3:}
\node(0)(2,0){0}\imark(0)
\node(p)(14,0){$p$}
\node(n-1)(26,0){$n$-$1$}
\node(n-2)(26,4){$n$-$2$}\rmark(n-2)
\node(f1)(10,4){$f_1$}
\node(f2)(18,4){$f_2$}
\drawedge(0,p){}
\drawedge(p,n-1){}
\drawedge(n-2,n-1){}
\drawloop[loopangle=270](n-1){}
\drawloop(f1){}
\drawloop(f2){}
\drawedge[linecolor=red,dash={.5 .25}{.25},curvedepth=-2.5](0,n-1){}
\drawedge(n-2,n-1){}
\drawloop[loopangle=270](n-1){}
\drawedge[linecolor=red,dash={.5 .25}{.25},curvedepth=1](f1,f2){}
\drawedge[linecolor=red,dash={.5 .25}{.25},curvedepth=1](f2,f1){}
\drawedge[linecolor=red,dash={.5 .25}{.25}](p,f2){}
\end{picture}\begin{picture}(28,12)(-2,-1)
\node[Nframe=n](name)(2,7){Subcase 2.4.1:}
\node(0)(2,0){0}\imark(0)
\node(p)(14,0){$p$}
\node(n-1)(26,0){$n$-$1$}
\node(n-2)(26,4){$n$-$2$}\rmark(n-2)
\node(x)(6,4){$x$}
\node(xt)(10,4){$xt$}
\node[Nframe=n](xdots)(14,4){$\dots$}
\node(xt^ell)(18,4){$xt^\ell$}
\drawedge(0,p){}
\drawedge(p,n-1){}
\drawedge(n-2,n-1){}
\drawloop[loopangle=270](n-1){}
\drawedge(x,xt){}
\drawedge(xt,xdots){}
\drawedge(xdots,xt^ell){}
\drawedge(xt^ell,n-1){}
\drawedge[linecolor=red,dash={.5 .25}{.25},curvedepth=-2.5](0,n-1){}
\drawedge[linecolor=red,dash={.5 .25}{.25}](p,xt^ell){}
\end{picture}\end{center}

\begin{center}\begin{picture}(28,13)(0,-1)
\node[Nframe=n](name)(2,8){Subcase 2.4.2:}
\node(0)(2,0){0}\imark(0)
\node(p)(14,0){$p$}
\node(n-1)(26,0){$n$-$1$}
\node(n-2)(26,4){$n$-$2$}\rmark(n-2)
\node(x)(10,4){$x$}
\node(xt)(18,4){$xt$}
\node(y)(10,8){$y$}
\drawedge(0,p){}
\drawedge(p,n-1){}
\drawedge(n-2,n-1){}
\drawloop[loopangle=270](n-1){}
\drawedge(x,xt){}
\drawedge(xt,n-1){}
\drawedge(y,xt){}
\drawedge[linecolor=red,dash={.5 .25}{.25},curvedepth=-2.5](0,n-1){}
\drawedge[linecolor=red,dash={.5 .25}{.25},curvedepth=5,syo=.5](p,y){}
\drawedge[linecolor=red,dash={.5 .25}{.25},curvedepth=1](xt,x){}
\drawedge[linecolor=red,dash={.5 .25}{.25}](x,y){}
\end{picture}\begin{picture}(28,13)(-2,-1)
\node[Nframe=n](name)(2,8){Subcase 2.4.3:}
\node(0)(2,0){0}\imark(0)
\node(p)(14,0){$p$}
\node(n-1)(26,0){$n$-$1$}
\node(n-2)(26,4){$n$-$2$}\rmark(n-2)
\node(x)(10,4){$x$}
\node(xt)(18,4){$xt$}
\drawedge(0,p){}
\drawedge(p,n-1){}
\drawedge(n-2,n-1){}
\drawloop[loopangle=270](n-1){}
\drawedge(x,xt){}
\drawedge(xt,n-1){}
\drawedge[linecolor=red,dash={.5 .25}{.25},curvedepth=-2.5](0,n-1){}
\drawedge[linecolor=red,dash={.5 .25}{.25}](p,x){}
\drawedge[linecolor=red,dash={.5 .25}{.25},curvedepth=-1](xt,x){}
\drawedge[linecolor=red,dash={.1 .1}{.1},curvedepth=3](x,n-2){(i)}
\drawedge[linecolor=red,dash={.1 .1}{.1},curvedepth=-.2,ELside=r,ELdist=.2](x,n-1){(ii)}
\end{picture}\end{center}

\begin{center}\begin{picture}(28,11)(0,-1)
\node[Nframe=n](name)(2,7){Subcase 2.4.4:}
\node(0)(2,0){0}\imark(0)
\node(p)(14,0){$p$}
\node(n-1)(26,0){$n$-$1$}
\node(n-2)(26,4){$n$-$2$}\rmark(n-2)
\node(x)(6,4){$x$}
\node(xt)(10,4){$xt$}
\node[Nframe=n](xdots)(14,4){$\dots$}
\node(xt^ell)(18,4){$xt^\ell$}
\drawedge(0,p){}
\drawedge(p,n-1){}
\drawedge(n-2,n-1){}
\drawloop[loopangle=270](n-1){}
\drawedge(x,xt){}
\drawedge(xt,xdots){}
\drawedge(xdots,xt^ell){}
\drawedge(xt^ell,n-2){}
\drawedge[linecolor=red,dash={.5 .25}{.25},curvedepth=-2.5](0,n-1){}
\drawedge[linecolor=red,dash={.5 .25}{.25},curvedepth=-.5](p,n-2){}
\end{picture}\begin{picture}(28,11)(-2,-1)
\node[Nframe=n](name)(2,7){Subcase 2.4.5:}
\node(0)(2,0){0}\imark(0)
\node(p)(14,0){$p$}
\node(n-1)(26,0){$n$-$1$}
\node(n-2)(26,4){$n$-$2$}\rmark(n-2)
\node(x)(6,4){$x$}
\node(xt)(10,4){$xt$}
\node[Nframe=n](xdots)(14,4){$\dots$}
\node(xt^ell)(18,4){$xt^\ell$}
\drawedge(0,p){}
\drawedge(p,n-1){}
\drawedge(n-2,n-1){}
\drawloop[loopangle=270](n-1){}
\drawedge(x,xt){}
\drawedge(xt,xdots){}
\drawedge(xdots,xt^ell){}
\drawloop(xt^ell){}
\drawedge[linecolor=red,dash={.5 .25}{.25},curvedepth=-2.5](0,n-1){}
\drawedge[linecolor=red,dash={.5 .25}{.25}](p,xt^ell){}
\end{picture}\end{center}

\begin{center}\begin{picture}(28,13)(0,-2)
\node[Nframe=n](name)(2,7){Subcase 2.5.1:}
\node(0)(2,0){0}\imark(0)
\node(p)(14,0){$p$}
\node(n-1)(26,0){$n$-$1$}
\node(n-2)(26,4){$n$-$2$}\rmark(n-2)
\node(f)(8,4){$f$}
\node(r1)(14,4){$r_1$}
\node[Nframe=n](rdots)(18,4){$\dots$}
\node(ru)(22,4){$r_u$}
\drawedge(0,p){}
\drawedge(p,n-1){}
\drawedge(n-2,n-1){}
\drawloop[loopangle=270](n-1){}
\drawloop(f){}
\drawedge[curvedepth=-.2](r1,n-1){}
\drawedge[curvedepth=0,exo=.2](rdots,n-1){}
\drawedge[curvedepth=0,exo=.5](ru,n-1){}
\drawedge[linecolor=red,dash={.5 .25}{.25},curvedepth=-2.5](0,n-1){}
\drawedge[linecolor=red,dash={.5 .25}{.25}](p,f){}
\drawedge[linecolor=red,dash={.5 .25}{.25}](r1,rdots){}
\drawedge[linecolor=red,dash={.5 .25}{.25}](rdots,ru){}
\drawedge[linecolor=red,dash={.5 .25}{.25},curvedepth=-2](ru,r1){}
\end{picture}\begin{picture}(28,13)(-2,-2)
\node[Nframe=n](name)(2,7){Subcase 2.5.2:}
\node(0)(2,0){0}\imark(0)
\node(p)(14,0){$p$}
\node(n-1)(26,0){$n$-$1$}
\node(n-2)(26,4){$n$-$2$}\rmark(n-2)
\node(f)(8,4){$f$}
\node(q1)(14,4){$q_1$}
\node[Nframe=n](qdots)(18,4){$\dots$}
\node(qv)(22,4){$q_v$}
\drawedge(0,p){}
\drawedge(p,n-1){}
\drawedge(n-2,n-1){}
\drawloop[loopangle=270](n-1){}
\drawloop(f){}
\drawedge[curvedepth=-3,exo=1](q1,n-2){}
\drawedge[curvedepth=-2,sxo=-1](qdots,n-2){}
\drawedge[curvedepth=0](qv,n-2){}
\drawedge[linecolor=red,dash={.5 .25}{.25},curvedepth=-2.5](0,n-1){}
\drawedge[linecolor=red,dash={.5 .25}{.25}](p,f){}
\drawedge[linecolor=red,dash={.5 .25}{.25},curvedepth=2](q1,qv){}
\drawedge[linecolor=red,dash={.5 .25}{.25}](qdots,q1){}
\drawedge[linecolor=red,dash={.5 .25}{.25}](qv,qdots){}
\end{picture}\end{center}
\caption{Map of the (sub)cases of Supercase~2 in the proof of Theorem~\ref{thm:bifix-free_upper_bound}.}
\end{figure}
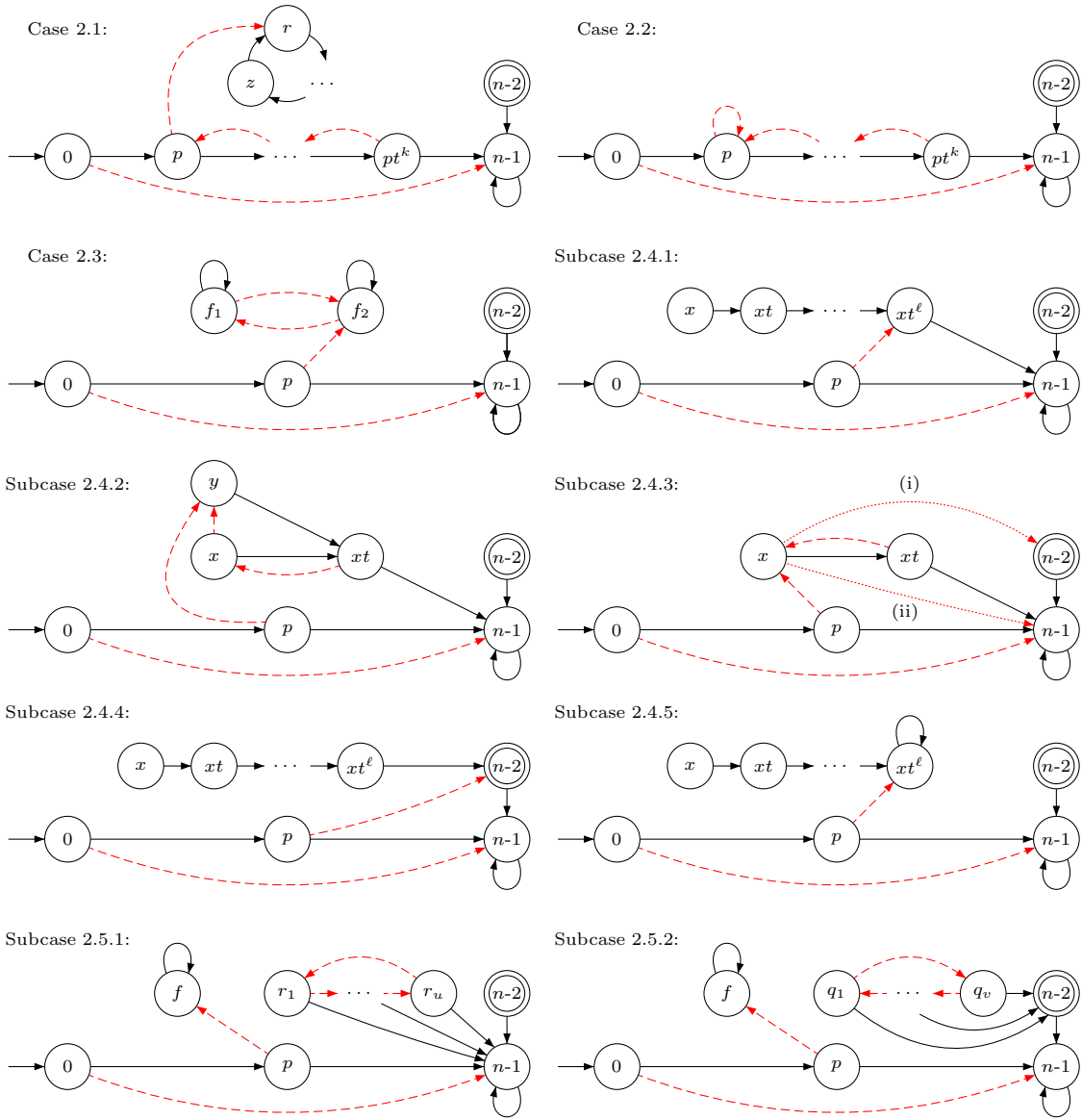


\begin{figure}[htb]
\unitlength 7.2pt\scriptsize
\gasset{Nh=2.5,Nw=2.5,Nmr=1.25,ELdist=0.5,loopdiam=1.5}
\begin{center}\begin{picture}(28,9.5)(0,-1)
\node[Nframe=n](name)(2,8.5){Case 3.1:}
\node(0)(2,6){0}\imark(0)
\node(p)(14,6){$p$}
\node(n-1)(26,0){$n$-$1$}
\node(n-2)(26,6){$n$-$2$}\rmark(n-2)
\node(z)(12,-1){$z$}
\node(r)(14,2){$r$}
\node[Nframe=n](rdots)(16,-1){$\dots$}
\node(q1)(18,3){$q_1$}
\node[Nframe=n](qdots)(20.5,3){$\dots$}
\node(qv)(23,3){$q_v$}
\drawedge(0,p){}
\drawedge(p,n-2){}
\drawedge(n-2,n-1){}
\drawloop[loopangle=270](n-1){}
\drawedge[curvedepth=.4](q1,n-2){}
\drawedge[curvedepth=.4,sxo=-1,exo=1.5](qdots,n-2){}
\drawedge(qv,n-2){}
\drawedge[curvedepth=1](z,r){}
\drawedge[curvedepth=1](r,rdots){}
\drawedge[curvedepth=1](rdots,z){}
\drawedge[linecolor=red,dash={.5 .25}{.25},curvedepth=2](0,n-2){}
\drawedge[linecolor=red,dash={.5 .25}{.25}](p,r){}
\drawedge[linecolor=red,dash={.5 .25}{.25}](q1,p){}
\drawedge[linecolor=red,dash={.5 .25}{.25},curvedepth=-.4,sxo=1,exo=-1.5](qdots,p){}
\drawedge[linecolor=red,dash={.5 .25}{.25},curvedepth=-.4](qv,p){}
\end{picture}\begin{picture}(28,9.5)(-2,-1)
\node[Nframe=n](name)(2,8.5){Subcase 3.2.1:}
\node(0)(2,6){0}\imark(0)
\node(p)(14,6){$p$}
\node(n-1)(26,0){$n$-$1$}
\node(n-2)(26,6){$n$-$2$}\rmark(n-2)
\node(x)(2,0){$x$}
\node(xt)(6,0){$xt$}
\node[Nframe=n](xdots)(10,0){$\dots$}
\node(xtl)(14,0){$xt^\ell$}
\node(q1)(18,3){$q_1$}
\node[Nframe=n](qdots)(20.5,3){$\dots$}
\node(qv)(23,3){$q_v$}
\drawedge(0,p){}
\drawedge(p,n-2){}
\drawedge(n-2,n-1){}
\drawedge[curvedepth=.4](q1,n-2){}
\drawedge[curvedepth=.4,sxo=-1,exo=1.5](qdots,n-2){}
\drawedge(qv,n-2){}
\drawedge(x,xt){}
\drawedge(xt,xdots){}
\drawedge(xdots,xtl){}
\drawedge[ELdist=.2](xtl,n-1){(ii)}
\drawedge[linecolor=red,dash={.5 .25}{.25},curvedepth=2](0,n-2){}
\drawedge[linecolor=red,dash={.5 .25}{.25}](q1,p){}
\drawedge[linecolor=red,dash={.5 .25}{.25},curvedepth=-.4,sxo=1,exo=-1.5](qdots,p){}
\drawedge[linecolor=red,dash={.5 .25}{.25},curvedepth=-.4](qv,p){}
\drawedge[linecolor=red,dash={.5 .25}{.25},curvedepth=-3](p,xtl){}
\drawloop[linecolor=red,dash={.1 .1}{.1},ELpos=80,ELdist=.2](xtl){(i)}
\end{picture}\end{center}

\begin{center}\begin{picture}(28,11)(0,-1)
\node[Nframe=n](name)(2,8.5){Subcase 3.2.2:}
\node(0)(2,6){0}\imark(0)
\node(p)(14,6){$p$}
\node(n-1)(26,0){$n$-$1$}
\node(n-2)(26,6){$n$-$2$}\rmark(n-2)
\node(x)(8,0){$x$}
\node(xt)(14,0){$xt$}
\node(y)(8,4){$y$}
\node(q1)(18,3){$q_1$}
\node[Nframe=n](qdots)(20.5,3){$\dots$}
\node(qv)(23,3){$q_v$}
\drawedge(0,p){}
\drawedge(p,n-2){}
\drawedge(n-2,n-1){}
\drawedge[curvedepth=.4](q1,n-2){}
\drawedge[curvedepth=.4,sxo=-1,exo=1.5](qdots,n-2){}
\drawedge(qv,n-2){}
\drawedge(x,xt){}
\drawedge(xt,n-1){}
\drawedge(y,xt){}
\drawedge[linecolor=red,dash={.5 .25}{.25},curvedepth=2](0,n-2){}
\drawedge[linecolor=red,dash={.5 .25}{.25}](q1,p){}
\drawedge[linecolor=red,dash={.5 .25}{.25},curvedepth=-.4,sxo=1,exo=-1.5](qdots,p){}
\drawedge[linecolor=red,dash={.5 .25}{.25},curvedepth=-.4](qv,p){}
\drawedge[linecolor=red,dash={.5 .25}{.25}](p,y){}
\drawedge[linecolor=red,dash={.5 .25}{.25},curvedepth=1](xt,x){}
\drawedge[linecolor=red,dash={.5 .25}{.25}](x,y){}
\end{picture}\begin{picture}(28,11)(-2,-1)
\node[Nframe=n](name)(2,8.5){Subcase 3.2.3:}
\node(0)(2,6){0}\imark(0)
\node(p)(14,6){$p$}
\node(n-1)(26,0){$n$-$1$}
\node(n-2)(26,6){$n$-$2$}\rmark(n-2)
\node(x)(8,0){$x$}
\node(xt)(14,0){$xt$}
\node(q1)(18,3){$q_1$}
\node[Nframe=n](qdots)(20.5,3){$\dots$}
\node(qv)(23,3){$q_v$}
\drawedge(0,p){}
\drawedge(p,n-2){}
\drawedge(n-2,n-1){}
\drawedge[curvedepth=.4](q1,n-2){}
\drawedge[curvedepth=.4,sxo=-1,exo=1.5](qdots,n-2){}
\drawedge(qv,n-2){}
\drawedge(x,xt){}
\drawedge(xt,n-1){}
\drawedge[linecolor=red,dash={.5 .25}{.25},curvedepth=2](0,n-2){}
\drawedge[linecolor=red,dash={.5 .25}{.25}](q1,p){}
\drawedge[linecolor=red,dash={.5 .25}{.25},curvedepth=-.4,sxo=1,exo=-1.5](qdots,p){}
\drawedge[linecolor=red,dash={.5 .25}{.25},curvedepth=-.4](qv,p){}
\drawedge[linecolor=red,dash={.5 .25}{.25}](p,x){}
\drawedge[linecolor=red,dash={.5 .25}{.25},curvedepth=-1](xt,x){}
\drawloop[linecolor=red,dash={.1 .1}{.1}](x){(i)}
\drawedge[linecolor=red,dash={.1 .1}{.1},ELdist=.2,curvedepth=-2](x,n-1){(ii)}
\end{picture}\end{center}

\begin{center}\begin{picture}(28,11.5)(0,-2)
\node[Nframe=n](name)(2,7.5){Subcase 3.2.4:}
\node(0)(2,5){0}\imark(0)
\node(p)(14,5){$p$}
\node(n-1)(26,0){$n$-$1$}
\node(n-2)(26,5){$n$-$2$}\rmark(n-2)
\node(x)(2,0){$x$}
\node(xt)(6,0){$xt$}
\node[Nframe=n](xdots)(10,0){$\dots$}
\node(xtl)(14,0){$xt^\ell$}
\node(q1)(18,2){$q_1$}
\node[Nframe=n](qdots)(20.5,2){$\dots$}
\node(qv)(23,2){$q_v$}
\drawedge(0,p){}
\drawedge(p,n-2){}
\drawedge(n-2,n-1){}
\drawloop[loopangle=270](n-1){}
\drawedge[curvedepth=.4](q1,n-2){}
\drawedge[curvedepth=.4,sxo=-1,exo=1.5](qdots,n-2){}
\drawedge(qv,n-2){}
\drawedge(x,xt){}
\drawedge(xt,xdots){}
\drawedge(xdots,xtl){}
\drawloop[loopangle=270](xtl){}
\drawedge[linecolor=red,dash={.5 .25}{.25},curvedepth=2](0,n-2){}
\drawedge[linecolor=red,dash={.5 .25}{.25},curvedepth=5.6,sxo=1,exo=-.2](q1,x){}
\drawedge[linecolor=red,dash={.5 .25}{.25},curvedepth=5.7,sxo=1,exo=-1](qdots,x){}
\drawedge[linecolor=red,dash={.5 .25}{.25},curvedepth=5.8,sxo=1,exo=-1.8](qv,x){}
\drawedge[linecolor=red,dash={.5 .25}{.25}](p,xtl){}
\drawedge[linecolor=red,dash={.5 .25}{.25},curvedepth=1](xtl,xdots){}
\drawedge[linecolor=red,dash={.5 .25}{.25},curvedepth=1](xdots,xt){}
\drawedge[linecolor=red,dash={.5 .25}{.25},curvedepth=1](xt,x){}
\drawedge[linecolor=red,dash={.5 .25}{.25}](x,p){}
\end{picture}\begin{picture}(28,11.5)(-2,-2)
\node[Nframe=n](name)(2,7.5){Case 3.3:}
\node(0)(2,5){0}\imark(0)
\node(p)(14,5){$p$}
\node(n-1)(26,0){$n$-$1$}
\node(n-2)(26,5){$n$-$2$}\rmark(n-2)
\node(q1)(18,2){$q_1$}
\node[Nframe=n](qdots)(20.5,2){$\dots$}
\node(qv)(23,2){$q_v$}
\node(f1)(8,0){$f_1$}
\node(f2)(14,0){$f_2$}
\drawedge(0,p){}
\drawedge(p,n-2){}
\drawedge(n-2,n-1){}
\drawloop[loopangle=270](n-1){}
\drawedge[curvedepth=.4](q1,n-2){}
\drawedge[curvedepth=.4,sxo=-1,exo=1.5](qdots,n-2){}
\drawedge(qv,n-2){}
\drawloop[loopangle=270](f1){}
\drawloop[loopangle=270](f2){}
\drawedge[linecolor=red,dash={.5 .25}{.25},curvedepth=2](0,n-2){}
\drawedge[linecolor=red,dash={.5 .25}{.25}](q1,p){}
\drawedge[linecolor=red,dash={.5 .25}{.25},curvedepth=-.4,sxo=1,exo=-1.5](qdots,p){}
\drawedge[linecolor=red,dash={.5 .25}{.25},curvedepth=-.4](qv,p){}
\drawedge[linecolor=red,dash={.5 .25}{.25},curvedepth=1](f1,f2){}
\drawedge[linecolor=red,dash={.5 .25}{.25},curvedepth=1](f2,f1){}
\drawedge[linecolor=red,dash={.5 .25}{.25},curvedepth=0](p,f2){}
\end{picture}\end{center}

\begin{center}\begin{picture}(28,15.5)(0,-1)
\node[Nframe=n](name)(2,12.5){Subcase 3.4.1:}
\node(0)(2,10){0}\imark(0)
\node(p)(14,10){$p$}
\node(n-1)(26,0){$n$-$1$}
\node(n-2)(26,10){$n$-$2$}\rmark(n-2)
\node(q1)(18,7){$q_1$}
\node[Nframe=n,Nw=1.5,Nh=1.5](qdots)(20.5,7){$\dots$}
\node(qv)(23,7){$q_v$}
\node(r1)(18,3){$r_1$}
\node[Nframe=n,Nw=1.5,Nh=1.5](rdots)(20.5,3){$\dots$}
\node(ru)(23,3){$r_u$}
\node(f)(8,3){$f$}
\drawedge(0,p){}
\drawedge(p,n-2){}
\drawedge(n-2,n-1){}
\drawloop[loopangle=270](n-1){}
\drawedge[curvedepth=.4](q1,n-2){}
\drawedge[curvedepth=.4,sxo=-1,exo=1.5](qdots,n-2){}
\drawedge(qv,n-2){}
\drawedge[curvedepth=-.4](r1,n-1){}
\drawedge[curvedepth=-.4,sxo=-1,exo=1.5](rdots,n-1){}
\drawedge(ru,n-1){}
\drawloop(f){}
\drawedge[linecolor=red,dash={.5 .25}{.25},curvedepth=2](0,n-2){}
\drawedge[linecolor=red,dash={.5 .25}{.25}](p,f){}
\drawedge[linecolor=red,dash={.5 .25}{.25},sxo=-.5,curvedepth=-1](q1,qdots){}
\drawedge[linecolor=red,dash={.5 .25}{.25},sxo=-.5,curvedepth=-1](qdots,qv){}
\drawedge[linecolor=red,dash={.5 .25}{.25},curvedepth=-1.5](qv,q1){}
\drawedge[linecolor=red,curvedepth=-.5,exo=1,dash={.5 .25}{.25}](r1,qv){}
\drawedge[linecolor=red,curvedepth=-.5,exo=.7,dash={.5 .25}{.25}](rdots,qv){}
\drawedge[linecolor=red,curvedepth=-.5,exo=0,dash={.5 .25}{.25}](ru,qv){}
\end{picture}\begin{picture}(28,15.5)(-2,-1)
\node[Nframe=n](name)(2,12.5){Subcase 3.4.2:}
\node(0)(2,10){0}\imark(0)
\node(p)(14,10){$p$}
\node(n-1)(26,0){$n$-$1$}
\node(n-2)(26,10){$n$-$2$}\rmark(n-2)
\node(q1)(17,7){$q_1$}
\node(r1)(17,3){$r_1$}
\node[Nframe=n](rdots)(20,3){$\dots$}
\node(ru)(23,3){$r_u$}
\node(f)(8,3){$f$}
\drawedge(0,p){}
\drawedge(p,n-2){}
\drawedge(n-2,n-1){}
\drawloop[loopangle=270](n-1){}
\drawedge[curvedepth=.3](q1,n-2){}
\drawedge[curvedepth=-.4](r1,n-1){}
\drawedge[curvedepth=-.4,sxo=-1,exo=1.5](rdots,n-1){}
\drawedge(ru,n-1){}
\drawloop(f){}
\drawedge[linecolor=red,dash={.5 .25}{.25},curvedepth=2](0,n-2){}
\drawedge[linecolor=red,dash={.5 .25}{.25}](p,f){}
\drawedge[linecolor=red,dash={.5 .25}{.25},curvedepth=-3.5,eyo=-.9,sxo=.3](r1,p){}
\drawedge[linecolor=red,dash={.5 .25}{.25},curvedepth=-3,eyo=-.7](rdots,p){}
\drawedge[linecolor=red,dash={.5 .25}{.25},curvedepth=-2.7,eyo=-.4](ru,p){}
\drawedge[linecolor=red,dash={.5 .25}{.25}](q1,f){}
\end{picture}\end{center}

\begin{center}\begin{picture}(28,10.5)(0,-1)
\node[Nframe=n](name)(2,8.5){Subcase 3.4.3:}
\node(0)(2,6){0}\imark(0)
\node(p)(14,6){$p$}
\node(n-1)(26,0){$n$-$1$}
\node(n-2)(26,6){$n$-$2$}\rmark(n-2)
\node(f)(8,0){$f$}
\node(r1)(18,3){$r_1$}
\node[Nframe=n,Nw=2,Nh=2](rdots)(20.5,3){$\dots$}
\node(ru)(23,3){$r_u$}
\drawedge(0,p){}
\drawedge(p,n-2){}
\drawedge(n-2,n-1){}
\drawloop[loopangle=270](n-1){}
\drawedge[curvedepth=-.4](r1,n-1){}
\drawedge[curvedepth=-.4,sxo=-1,exo=1.5](rdots,n-1){}
\drawedge(ru,n-1){}
\drawloop(f){}
\drawedge[linecolor=red,dash={.5 .25}{.25},curvedepth=2.5](0,n-2){}
\drawedge[linecolor=red,dash={.5 .25}{.25}](p,f){}
\drawedge[linecolor=red,dash={.5 .25}{.25}](r1,p){}
\drawedge[linecolor=red,dash={.5 .25}{.25},curvedepth=2,sxo=1](rdots,f){}
\drawedge[linecolor=red,dash={.5 .25}{.25},curvedepth=2,eyo=-.5](ru,f){}
\end{picture}\begin{picture}(28,10.5)(-2,-1)
\node[Nframe=n](name)(2,8.5){Subcase 3.5.1:}
\node(0)(2,6){0}\imark(0)
\node(p)(8,6){$p$}
\node[Nframe=n](pdots)(14,6){$\dots$}
\node(pt^k)(20,6){$pt^k$}
\node(n-2)(26,6){$n$-$2$}\rmark(n-2)
\node(n-1)(26,0){$n$-$1$}
\drawedge(0,p){}
\drawedge(p,pdots){}
\drawedge(pdots,pt^k){}
\drawedge(pt^k,n-2){}
\drawedge(n-2,n-1){}
\drawloop[loopangle=270](n-1){}
\drawedge[linecolor=red,dash={.5 .25}{.25},curvedepth=2.5](0,n-2){}
\drawedge[linecolor=red,dash={.5 .25}{.25},curvedepth=1](pdots,p){}
\drawedge[linecolor=red,dash={.5 .25}{.25},curvedepth=1](pt^k,pdots){}
\drawloop[linecolor=red,dash={.5 .25}{.25},loopangle=270](p){}
\end{picture}\end{center}

\begin{center}\begin{picture}(28,12)(0,-1)
\node[Nframe=n](name)(2,9.5){Subcase 3.5.2:}
\node(0)(2,7){0}\imark(0)
\node(p)(8,7){$p$}
\node[Nframe=n](pdots)(14,7){$\dots$}
\node(pt^k)(20,7){$pt^k$}
\node(y)(20,2){$y$}
\node(n-2)(26,7){$n$-$2$}\rmark(n-2)
\node(n-1)(26,2){$n$-$1$}
\drawedge(0,p){}
\drawedge(p,pdots){}
\drawedge(pdots,pt^k){}
\drawedge(pt^k,n-2){}
\drawedge(y,pt^k){}
\drawedge(n-2,n-1){}
\drawloop[loopangle=270](n-1){}
\drawedge[linecolor=red,dash={.5 .25}{.25},curvedepth=2.5](0,n-2){}
\drawedge[linecolor=red,dash={.5 .25}{.25}](y,n-1){}
\drawedge[linecolor=red,dash={.5 .25}{.25},curvedepth=-1,sxo=-1](p,y){}
\drawedge[linecolor=red,dash={.5 .25}{.25},curvedepth=1](pdots,p){}
\drawedge[linecolor=red,dash={.5 .25}{.25},curvedepth=1](pt^k,pdots){}
\end{picture}\begin{picture}(28,12)(-2,0)
\node[Nframe=n](name)(2,10.5){Subcase 3.5.3:}
\node(0)(2,8){0}\imark(0)
\node(p)(8,8){$p$}
\node[Nframe=n](pdots)(14,8){$\dots$}
\node(pt^k)(20,8){$pt^k$}
\node[Nw=2,Nh=2](q1)(14,3){$q_1$}
\node[Nframe=n,Nw=2,Nh=2](qdots1)(16.5,3){$\dots$}
\node[Nw=2,Nh=2](qi)(19,3){$q_m$}
\node[Nframe=n,Nw=2,Nh=2](qdots2)(21.5,3){$\dots$}
\node[Nw=2,Nh=2](qv)(24,3){$q_v$}
\node(x)(8,0){$x$}
\node[Nframe=n](xdots)(14,0){$\dots$}
\node(n-2)(26,8){$n$-$2$}\rmark(n-2)
\node(n-1)(26,0){$n$-$1$}
\drawedge(0,p){}
\drawedge(p,pdots){}
\drawedge(pdots,pt^k){}
\drawedge(pt^k,n-2){}
\drawedge(n-2,n-1){}
\drawloop[loopangle=270](n-1){}
\drawedge(x,xdots){}
\drawedge[curvedepth=-1.5,exo=.5](xdots,qi){}
\drawedge[curvedepth=.6](q1,n-2){}
\drawedge[curvedepth=.3,sxo=-1](qdots1,n-2){}
\drawedge[curvedepth=.2,sxo=-.5](qi,n-2){}
\drawedge[curvedepth=.1](qdots2,n-2){}
\drawedge[curvedepth=0](qv,n-2){}
\drawedge[linecolor=red,dash={.5 .25}{.25},curvedepth=2.5](0,n-2){}
\drawedge[linecolor=red,dash={.5 .25}{.25},curvedepth=1](pdots,p){}
\drawedge[linecolor=red,dash={.5 .25}{.25},curvedepth=1](pt^k,pdots){}
\drawedge[linecolor=red,dash={.5 .25}{.25}](p,x){}
\drawedge[linecolor=red,dash={.5 .25}{.25},curvedepth=-1.5](q1,qdots1){}
\drawedge[linecolor=red,dash={.5 .25}{.25},curvedepth=-1.5](qdots1,qi){}
\drawedge[linecolor=red,dash={.5 .25}{.25},curvedepth=-1.5](qi,qdots2){}
\drawedge[linecolor=red,dash={.5 .25}{.25},curvedepth=-1.5](qdots2,qv){}
\drawedge[linecolor=red,dash={.5 .25}{.25},curvedepth=-2.5,sxo=1,exo=-1](qv,q1){}
\end{picture}\end{center}
\caption{Map of the (sub)cases of Supercase~3 in the proof of Theorem~\ref{thm:bifix-free_upper_bound}.}
\end{figure}
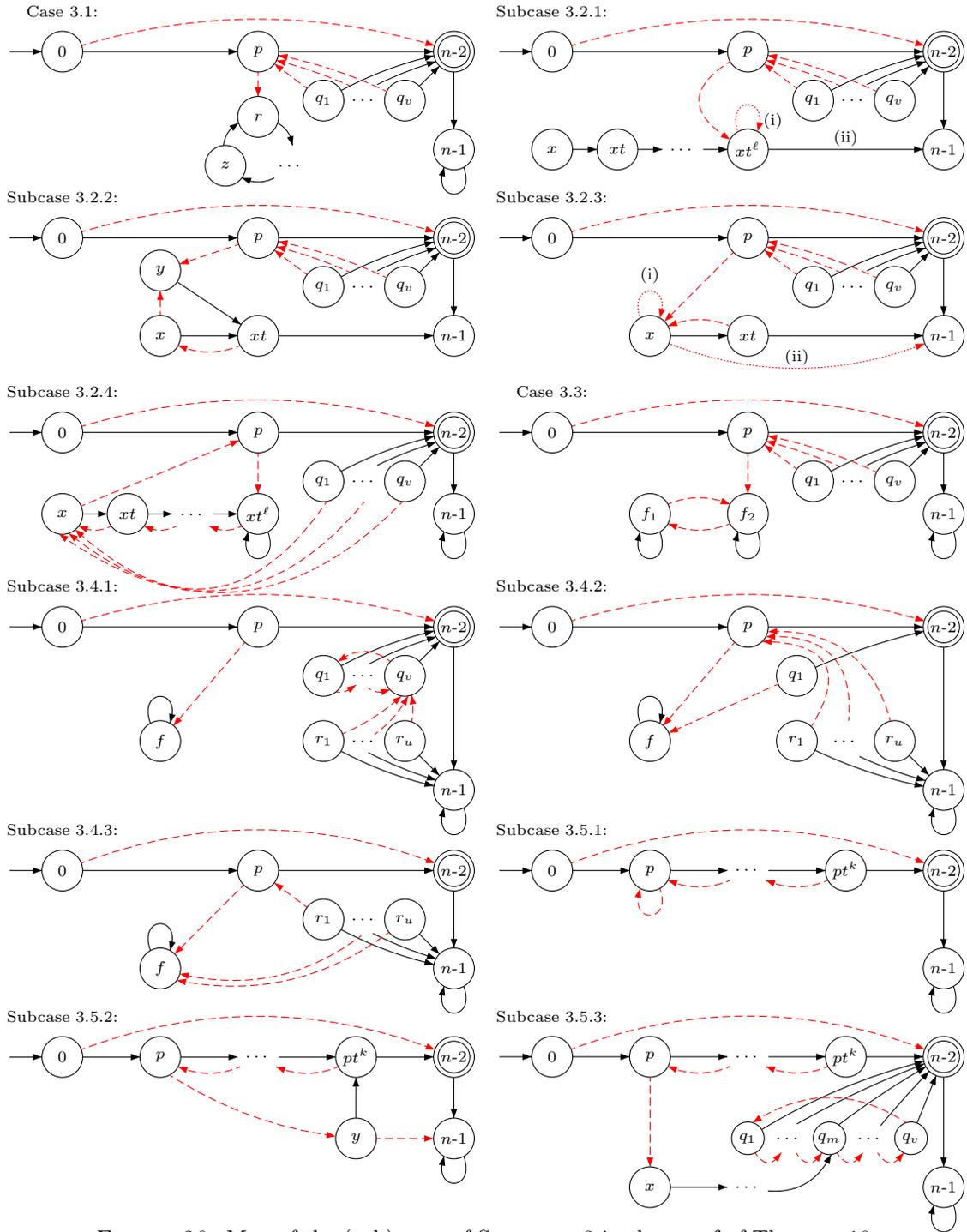

\end{document}